\documentclass[journal,compsoc]{IEEEtran}
\IEEEoverridecommandlockouts

\ifCLASSOPTIONcompsoc
  \usepackage[nocompress]{cite}
\else
  \usepackage{cite}
\fi

\usepackage{bm}
\usepackage{amssymb}
\usepackage{amsmath}
\usepackage{amsthm}
\newtheorem{theorem}{Theorem}
\newtheorem{lemma}{Lemma}
\newtheorem{corollary}{Corollary}
\newtheorem{definition}{Definition}

\newtheorem{observation}{Observation}
\newtheorem{proposition}{Proposition}
\newtheorem{example}{Example}

\newcommand{\minitab}[2][l]{\begin{tabular}{#1}#2\end{tabular}}

\usepackage{graphicx}
\usepackage{setspace}
\usepackage{dsfont}
\usepackage{cite}
\usepackage{subfigure}
\usepackage{caption}
\usepackage{algorithm}
\usepackage{algorithmic}
\usepackage{multirow}
\usepackage{stfloats}
\usepackage{url}
\usepackage{color}
\usepackage{comment}

\begin{document}

\title{Cooperative Wi-Fi Deployment: \\A One-to-Many Bargaining Framework}

\author{Haoran~Yu,~\IEEEmembership{Student Member,~IEEE,}
        Man~Hon~Cheung, and~Jianwei~Huang,~\IEEEmembership{Fellow,~IEEE}
\IEEEcompsocitemizethanks{\IEEEcompsocthanksitem H. Yu, M. H. Cheung, and J. Huang are with the Department of Information Engineering, The Chinese University of Hong Kong, Hong Kong. E-mail: \{yh012, mhcheung, jwhuang\}@ie.cuhk.edu.hk.}
\thanks{Part of the results appeared in {\it IEEE WiOpt 2015} \cite{yu2015cooperative}. This work is supported by the General Research Fund (Project Number CUHK 14202814) established under the University Grant Committee of the Hong Kong Special Administrative Region, China.}
}

\thispagestyle{empty}

\twocolumn

\IEEEtitleabstractindextext{

\begin{abstract}
We study the cooperation of the mobile network operator (MNO) and the venue owners (VOs) on the public Wi-Fi deployment. We consider a \emph{one-to-many bargaining} framework, where the MNO bargains with VOs sequentially to determine where to deploy Wi-Fi and how much to pay. {{Taking into account the negative externalities among different steps of bargaining, we analyze the following two cases:}} for the \emph{exogenous bargaining sequence} case, we compute the optimal bargaining solution on the cooperation decisions and payments under a predetermined bargaining sequence; for the \emph{endogenous bargaining sequence} case, the MNO decides the bargaining sequence to maximize its payoff. Through exploring the structural property of the optimal bargaining sequence, we design a low-complexity \emph{Optimal VO Bargaining Sequencing} (OVBS) algorithm to search the optimal sequence. More specifically, we categorize the VOs into three types based on the impact of the Wi-Fi deployment at their venues, and show that it is optimal for the MNO to bargain with these three types of VOs sequentially. Numerical results show that compared with the random and worst bargaining sequences, the optimal bargaining sequence improves the MNO's payoff by up to 14.8\% and 45.3\%, respectively.
\end{abstract}

}

\newpage 
\setcounter{page}{1}

\maketitle

\IEEEdisplaynontitleabstractindextext

\IEEEpeerreviewmaketitle

\IEEEraisesectionheading{\section{Introduction}\label{Section1}}

\subsection{Motivation}
\IEEEPARstart{T}{he} proliferation of mobile devices has lead to an explosive growth of global mobile data traffic, so the \emph{mobile network operators} ({MNOs}) are seeking innovative approaches to expand the network capacity and improve users' quality of experience. With the recent technology developments and standardization efforts (\emph{e.g.}, Hotspot 2.0 and 
{{the access network discovery and selection function \cite{4Ginteg}}}), Wi-Fi data offloading has emerged as an important approach to alleviate cellular congestion. A recent study \cite{lee2013mobile} showed that Wi-Fi has offloaded 65\% of total mobile traffic in the major cities in Korea. Furthermore, the Wireless Broadband Alliance's report \cite{WBA} estimated that the annual global Wi-Fi deployment rate will increase to $10.5$ million in 2018.

Instead of building their own Wi-Fi hotspots, many MNOs have been collaborating with \emph{venue owners} (VOs), which are the owners of public places such as shopping malls and stadiums, on hotspot installment \cite{WBA}. Since a large volume of cellular data traffic is generated from these crowded public places, MNOs are especially interested in deploying hotspots at these venues to relieve the traffic congestion. With the location information provided by Wi-Fi hotspots, MNOs can also earn profits by delivering context-aware mobile advertisements to mobile users.{\footnote{Although the MNO can also deliver advertisements through the cellular network, users are much more receptive to advertising through Wi-Fi due to their voluntary use of Wi-Fi \cite{Cisco2013}. Furthermore, Wi-Fi usually provides more accurate user localization, and is more suitable for supporting multimedia advertisement due to the higher data rate.}} Meanwhile, VOs also welcome the MNOs' help in building the carrier-grade Wi-Fi, which usually provides a higher capacity and better integration with the cellular network than a regular Wi-Fi \cite{Cisco2013}, hence significantly enhances the mobile users' experience and attracts more visitors to those Wi-Fi available venues. Moreover, the carrier-grade Wi-Fi can help both MNOs and VOs collect visitor analytics, provide location-based services, and promote products or activities \cite{WBA,Cisco2013}. Therefore, both MNOs and VOs benefit from the Wi-Fi deployment and have incentives to provide Wi-Fi service cooperatively. 
For example, AT\&T has been cooperating with some VOs (such as Starbucks) to install the public Wi-Fi networks \cite{ATTWiFismall}. 
Although this kind of MNO-VO cooperation is increasingly popular, the detailed economic interactions among MNOs and VOs still have not been sufficiently explored and understood by the existing literatures. This motivates us to extensively analyze both MNOs and VOs' strategies in the cooperative Wi-Fi deployment in this paper.
\subsection{Our Work}
We consider a case where both MNOs and VOs have considerable market power, and study the cooperative Wi-Fi deployment problem under the one-to-many bargaining framework.{\footnote{The case where different sides have unbalanced market power can be studied in the same framework as in this paper, using the asymmetric Nash bargaining formulation \cite{napel2002bilateral}.}} Specifically, a monopoly MNO bargains with multiple VOs sequentially, \emph{i.e.}, at each step, the MNO bargains with only one VO for deploying Wi-Fi at the corresponding venue.{\footnote{More precisely, the one-to-many bargaining contains several types. The most common type is the one-to-many bargaining with a \emph{sequential} bargaining protocol. Another type is the one-to-many bargaining with a \emph{concurrent} bargaining protocol, where the buyer bargains with multiple sellers concurrently \cite{gao2014bargaining}. In practice, conducting the concurrent bargaining is much more difficult than the sequential bargaining, as it requires the evaluation of simultaneous responses of all bargainers. In this paper, we focus on the sequential bargaining protocol in the one-to-many bargaining.}} We analyze the bargaining solution of each step, including the cooperation decision and payment, by using the \emph{Nash bargaining theory} \cite{nash1950bargaining}. Since the MNO's willingness to deploy new hotspots decreases as the number of deployed hotspots increases, the cooperation between the MNO and a particular VO imposes {a} \emph{negative externality} to the bargaining among the MNO and other VOs. Such an externality significantly complicates the analysis. There are very few literatures studying the one-to-many bargaining, especially under the Nash bargaining theory. Our work provides a systematic study on this problem.

In the first part of this paper, we study the \emph{exogenous} {bargaining sequence} scenario, where the MNO bargains with VOs sequentially according to a predetermined bargaining sequence.
We take into account the \emph{data offloading benefit}, \emph{Wi-Fi operation cost}, \emph{advertising profit}, and \emph{business revenue} of the MNO and VOs. In particular, we differentiate the MNO's \emph{data offloading benefit} at a venue during different time periods (\emph{e.g.}, daytime and nighttime). We would like to answer the following key questions: (i) \emph{Which VOs should the MNO cooperate with?} (ii) \emph{How much should the MNO pay these VOs?} We apply backward induction to compute the optimal bargaining solution on the cooperation decisions and payments.

In the second part of this paper, we study the \emph{endogenous} {bargaining sequence} scenario, where the MNO first determines the bargaining sequence and then bargains with VOs accordingly. We want to answer the following key question: \emph{Under what bargaining sequence can the MNO maximize its payoff?} Based on the analysis in the first part, we can compute the MNO's payoff under a fixed bargaining sequence. However, due to the complex structure of the one-to-many bargaining, we often cannot obtain the closed-form solution of such a payoff. Therefore, it is very challenging to directly compare the MNO's payoffs under all possible bargaining sequences and determine the optimal one.

To tackle the high complexity of the optimal sequencing problem, we first establish an important structural property of the one-to-many bargaining. More precisely, we categorize VOs into three types based on the impact of the Wi-Fi deployment at their venues. We show that there exists a group of optimal bargaining sequences, under which the MNO bargains with these three types of VOs sequentially. As a result, we design an \emph{Optimal VO Bargaining Sequencing} (OVBS) algorithm that searches for the optimal bargaining sequence from a significantly reduced set. In fact, the structural property we prove in this paper is general, and is valid for many other one-to-many bargaining problems. We further characterize two special system settings, where we can explicitly determine the optimal sequence without running OVBS.

In the third part of this paper, we study the influence of the bargaining sequence on the VOs' payoffs. Our analysis shows that: (i) When VOs are homogenous, it is beneficial for a VO to bargain with the MNO as early as possible; (ii) When VOs are heterogenous, ``the earlier the better'' is no longer true in general.

The main contributions of this paper are as follows:

\begin{itemize}
\item \emph{Study of the one-to-many bargaining with cooperation cost}: To the best of our knowledge, this is the first work studying the one-to-many bargaining with the \emph{cooperation cost} (\emph{i.e.}, Wi-Fi deployment and operation cost) under the Nash bargaining theory.
We show that with the cooperation cost, the bargaining sequence significantly influences the bargaining results. We analyze the one-to-many bargaining with both \emph{exogenous} and \emph{endogenous} bargaining sequences. The results in this paper are general enough to be applied in other one-to-many bargaining problems.

\item \emph{Modeling and analysis of the cooperative Wi-Fi deployment}: As far as we know, this is the first work studying the economic interactions among the MNO and VOs in terms of the cooperative Wi-Fi deployment. We show the negative externalities among different steps of negotiation, and analyze the bargaining results for any given bargaining sequence.

\item \emph{Low-complexity optimal bargaining sequence search algorithm}: Motivated by the fact that the bargaining sequence influences the bargaining results, we formulate the MNO's optimal bargaining sequencing problem. Then we prove an important structural property for the optimal bargaining sequence, and design a low-complexity OVBS algorithm to search the optimal sequence. Numerical results show that the optimal bargaining sequence improves the MNO's payoff over the random and worst bargaining sequences by up to 14.8\% and 45.3\%, respectively.

\item \emph{Study of the bargaining sequence's impact on VOs}: We prove that for homogenous VOs, bargaining with the MNO at earlier positions always improves their payoffs. However, for heterogenous VOs, earlier bargaining positions may decrease their payoffs. To the best of our knowledge, this is the first paper showing and explaining this feature.

\end{itemize}
\subsection{Literature Review}\label{literaturereview}
\subsubsection{Deployment of MNO's Wi-Fi Networks}
{{There are a few literatures studying the MNO's Wi-Fi access point deployment problem.}} Zheng \emph{et al.} in \cite{zheng2012sparse} proposed Wi-Fi access point deployment algorithms, which provide the worst-case guarantee to the interconnection gap for vehicular Internet access. 
Wang \emph{et al.} in \cite{wang2013exploiting} exploited users' mobility patterns to deploy Wi-Fi access points, aiming at maximizing the continuous Wi-Fi coverage for mobile users. 
{{Bulut \emph{et al.} in \cite{bulut2013wifi} analyzed some real user mobility traces and deployed Wi-Fi access points based on the density of users' data access requests. 
Liao \emph{et al.} in \cite{liao2011two} investigated the Wi-Fi access point deployment problem with the consideration of both the coverage and localization accuracy. 
Poularakis \emph{et al.} in \cite{george2016mobihoc} studied a joint Wi-Fi access point deployment and Wi-Fi service pricing problem. 
These works focused on a single MNO's Wi-Fi deployment decision, and did not consider the VOs, who may collaborate with the MNO and compensate the MNO's Wi-Fi deployment cost.}}
\subsubsection{Economics of VOs' Wi-Fi Networks}
{{There have been many literatures studying the mobile data offloading market, where the MNOs lease the VOs' (or resident users') Wi-Fi networks to offload the cellular data traffic. 
For example, Iosifidis \emph{et al.} in \cite{iosifidis2015double} designed an iterative double auction mechanism for an offloading market, where the MNOs compete to lease the VOs' Wi-Fi networks for data offloading. The authors proposed an efficient allocation and payment rule that maximizes the social welfare. References \cite{paris2015efficient,dong2014ideal,lu2014easybid} designed reverse auctions for an MNO to motivate the VOs to offload the cellular traffic. 
Gao \emph{et al}. in \cite{gao2014bargaining} applied a bargaining framework to study a similar Wi-Fi capacity trading problem. 
Furthermore, Yu \emph{et al.} in \cite{yu2016wifiad} focused on the VOs' optimal Wi-Fi monetization strategies by considering the Wi-Fi advertising technique. 
However, these works assumed that the Wi-Fi networks have already been deployed and are owned by the VOs. They did not study the VOs' cooperation with the MNO in deploying the Wi-Fi networks.}}

\subsubsection{One-to-Many Bargaining}
In terms of the \emph{one-to-many bargaining}, the most relevant works are \cite{gao2014bargaining}, \cite{moresi2008model}. Both papers studied the one-to-many bargaining under the Nash bargaining theory. However, since they did not consider the \emph{cooperation cost}, their conclusion was that the bargaining sequence does not affect the buyer's payoff, and their analysis was limited to the one-to-many bargaining with \emph{exogenous} sequence. In our work, we take into account the \emph{cooperation cost} (\emph{i.e.}, Wi-Fi deployment and operation cost), which complicates the one-to-many bargaining with \emph{exogenous} sequence. Such a consideration also motivates us to study the one-to-many bargaining with \emph{endogenous} sequence.
References \cite{li2010one,cai2000delay,cai2003inefficient} studied several one-to-many bargaining problems, where the buyer bargains with multiple sellers on a joint project that requires the cooperation from all the participants. It is different from our problem, as here the MNO may only cooperate with a subset of the VOs on the Wi-Fi deployment.

The rest of the paper is organized as follows. In Section \ref{Section2}, we introduce the system model. In Section \ref{Section3}, we analyze the bargaining between the MNO and a single VO. In Sections \ref{Section4} and \ref{Section5}, we study the one-to-many bargaining with exogenous and endogenous bargaining sequences, respectively. In Section \ref{sec:VOpayoff}, we investigate the impact of the bargaining sequence on the VOs. We provide the numerical results in Section \ref{sec:numerical}, and conclude the paper in Section \ref{sec:conclusion}.

\section{System Model}\label{Section2}
\subsection{Basic Settings}
We consider one mobile network operator (MNO), who operates multiple macrocells and bargains with venue owners (VOs) to deploy Wi-Fi access points. For simplicity, we assume that each venue (such as a cafe) has a limited space and hence is covered by only one cellular macrocell. Since deploying Wi-Fi at a particular venue only offloads traffic for the corresponding macrocell under our assumption and does not benefit other macrocells, the MNO can consider the Wi-Fi deployments for different macrocells separately. Without loss of generality, we study the MNO's strategy within one macrocell.

We consider a set ${\cal{N}}\triangleq \left\{1,2,\ldots,N\right\}$ of VOs, whose venues are non-overlapping but covered by the same macrocell. According to \cite{wang2015understanding}, the mobile traffic exhibits a periodical daily pattern. Hence, we divide a day equally into $T\in\left\{1,2,\ldots\right\}$ time periods, and assume that when Wi-Fi is deployed at venue $n$,{\footnote{To simplify the description, we use venue $n$ to refer to VO $n$'s venue.}} the expected amount of offloaded macrocell traffic during the $t$-th ($t=1,2,\ldots,T$) time period is $X_n^t\ge 0$. We define
\begin{align}
\bm{X}_n \triangleq \left(X_n^1,X_n^2,\ldots,X_n^T\right)
\end{align}
as the offloading vector of VO $n$. Each VO $n\in{\cal{N}}$ is further characterized by parameters $R_n$, $C_n$, and $A_n$:
\begin{itemize}
\item $R_n\ge0$ denotes the extra revenue that Wi-Fi creates for VO $n$'s business (\emph{e.g.}, via attracting more customers and collecting customer analytics);{\footnote{{{Different from ${\bm X}_n$, we aggregate the extra revenues obtained by VO $n$ during different time periods into a single parameter $R_n$. The reason is that VO $n$'s payoff is linear in $R_n$, as we will discuss in Section \ref{subsec:payoff}. Hence, considering the total value leads to the same result as considering different values in different time periods. Similar explanations apply for the definitions of parameters $C_n$ and $A_n$.}}}}
\item $C_n\ge0$ denotes the total cost for the MNO to deploy and operate Wi-Fi at venue $n$, including the installment fee, management cost, and backhaul cost;{\footnote{In practice, some VOs undertake the backhaul cost for the MNO. This can be easily incorporated into our analysis by properly redefining $R_n$ and $C_n$.}}
\item $A_n\ge0$ denotes the advertising profit to the MNO when Wi-Fi is deployed at venue $n$.{\footnote{Sometimes VOs promote their products via Wi-Fi, and we include the corresponding advertising profit in $R_n$.}}
\end{itemize}
We assume that the information of $\bm{X}_n$, $R_n$, $C_n$, and $A_n$ for all $n\in{\cal N}$ is known to the MNO and all VOs.{\footnote{{In practice, the MNO and VOs can estimate these parameters. For example, parameter ${\bm X}_n$ can be estimated by combining the results in \cite{wang2015understanding} and \cite{lee2013mobile}, which studied the spatial-temporal distribution of cellular traffic and the percentage of offloaded cellular traffic, respectively. Parameters $R_n$ and $A_n$ are mainly determined by the statistics like the number of customers and the customers' average sojourn time, which can be estimated by the method proposed in \cite{kim2006extracting}. Parameter $C_n$ can be estimated based on \cite{SenzaEco}, which showed the Wi-Fi hotspots' detailed capital expenditures (\emph{e.g.}, equipment fees) and operating expenses (\emph{e.g.}, backhaul costs, power costs, and maintenance fees).}}} 
This allows us to focus on studying the optimal bargaining decisions in this paper. In our future work, we will further analyze how incomplete and asymmetric information affects the cooperation among the MNO and VOs.

\subsection{MNO's Payoff, VO's Payoff, and Social Welfare}\label{subsec:payoff}
We use ${{b_n}}\in\left\{0,1\right\}$ to denote the bargaining outcome between the MNO and VO $n$: $b_n=1$ if they agree on the Wi-Fi deployment at venue $n$, and $b_n=0$ otherwise. We use $p_n\in{\mathbb R}$ to denote the MNO's payment to VO $n$.{\footnote{We allow $p_n$ to be negative, in which case VO $n$ pays the MNO. {{This will be the case when deploying Wi-Fi is more beneficial to VO $n$ than to the MNO.}}}} As we will see in Sections \ref{Section3} and \ref{Section4}, under the Nash bargaining solution, ${{p_n}}=0$ whenever ${{b_n}}=0$, \emph{i.e.,} there is no transfer if no agreement is reached.

To simplify the notations, we define
\begin{align}
{{\bm{b}}_n}\triangleq \left(b_1,b_2,\ldots,b_n\right){\rm~and~}{{\bm{p}}_n}\triangleq \left(p_1,p_2,\ldots,p_n\right)
\end{align}
as the bargaining outcomes and payments between the MNO and the \emph{first} $n\in{\cal N}$ VOs, respectively.

The \emph{MNO's payoff} depends on the offloading benefit, advertising profit, Wi-Fi deployment and operation cost, and its payment to VOs. Based on ${\bm b}_N$ and ${\bm p}_N$, the MNO's payoff is
\begin{align}
\nonumber
{U}\left( {\bm b}_N,{\bm p}_N\right)\triangleq & \sum\limits_{t=1}^T {f_t\left( {\sum\limits_{n = 1}^N {{b_n}{X_n^t}}} \right)}\\
& + \sum\limits_{n = 1}^N {{b_n}\left( {{A_n}- {C_n}} \right)}-\sum\limits_{n = 1}^N {{p_n}}.\label{MNOpayoff}
\end{align}
Here, $f_t\left( \cdot  \right),t=1,2,\ldots,T,$ is an increasing and concave function with $f_t\left( 0 \right) = 0$,{\footnote{{{Notice that the situation where function $f_t\left( \cdot  \right)$ is linear for all $t$ is a special case of our framework. In this case, there is no externality among different steps of bargaining, and the one-to-many bargaining problem degenerates to $N$ independent one-to-one bargaining between the MNO and each VO.}}}} and ${\sum_{n = 1}^N {{b_n}{X_n^t}}}$ is the MNO's total offloaded traffic from all the $N$ venues during the $t$-th time period. Hence, $f_t\left( {\sum_{n = 1}^N {{b_n}{X_n^t}}} \right)$ characterizes the \emph{offloading benefit} of the MNO during the $t$-th time period, and $\sum_{t=1}^T {f_t\left( {\sum_{n = 1}^N {{b_n}{X_n^t}}} \right)}$ is the MNO's total \emph{offloading benefit} of all time periods.{\footnote{Reference \cite{gao2014bargaining} used a similar function to characterize the MNO's serving cost reduction due to the data offloading. However, \cite{gao2014bargaining} did not consider the temporal heterogeneity of the offloaded traffic, while our work defines the offloading benefit function $f_t\left(\cdot\right)$ for each time period $t=1,2,\ldots,T$.}} Furthermore, $\sum_{n = 1}^N {{b_n} {A_n}}$ and $\sum_{n = 1}^N {{b_n} {C_n}}$ describe the MNO's total advertising profit and total cost, respectively. Term $\sum_{n = 1}^N {{p_n}}$ is the MNO's total payment to the VOs.

\emph{VO $n$'s payoff} depends on the revenue directly brought by Wi-Fi and the MNO's payment as
\begin{align}
{V_n}\left( {{b_n},{p_n}} \right) \triangleq {b_n}{R_n} + {p_n}.\label{VOpayoff}
\end{align}

The \emph{social welfare} is the aggregate payoff of the MNO and all VOs:
\begin{align}
\nonumber
\Psi \left( {\bm b}_N \right)& \triangleq U\left( {{{\bm b}_N},{{\bm p}_N}} \right) + \sum\limits_{n = 1}^N {{V_n}\left( {{b_n},{p_n}} \right)}\\
&=\sum\limits_{t=1}^{T} {f_t\left( {\sum\limits_{n = 1}^N {{b_n}{X_n^t}} } \right)} + \sum\limits_{n = 1}^N {{b_n} {Q_n}},\label{socialwelfare}
\end{align}
where for each VO $n\in\cal{N}$, we define
\begin{align}
{Q_n}\triangleq {R_n} + {A_n} - {C_n}.\label{Qn}
\end{align}
Here $Q_n$ captures the increase in social welfare by deploying Wi-Fi at venue $n$, excluding the data offloading effect. Hence, we call $Q_n$ as the \emph{net benefit} of deploying Wi-Fi at venue $n$ without considering the data offloading benefit.
We summarize the key notations in this paper in Table \ref{table:notation}, including some notations to be discussed in Sections \ref{Section3} and \ref{Section4}.

Since the payment terms are cancelled out in (\ref{socialwelfare}), the social welfare only depends on the bargaining outcomes ${{\bm b}_N}=\left(b_1,b_2,\ldots,b_N\right)$ between the MNO and $N$ VOs.

\begin{table}[t]
\centering
\caption{Main Notations}\label{table:notation}
\begin{tabular}{|p{1.8cm}|p{6cm}|}
\hline
{\minitab[c]{$n,{\cal N}$}} & {VO index and its feasible set}\\
\hline
{\minitab[c]{$t$}} & {Time period index}\\
\hline
{\minitab[c]{$X_n^t$}} & {Amount of offloaded traffic at venue $n$ during the $t$-th time period}\\
\hline
{\minitab[c]{$Q_n$}} & {Net benefit of deploying Wi-Fi at venue $n$ without data offloading effect}\\
\hline
{\minitab[c]{$f_t\left(\cdot\right)$}} & {MNO's data offloading benefit function for the $t$-th time period}\\
\hline
{\minitab[c]{${\bm b}_n$}} & {Bargaining outcomes between the MNO and the first $n$ VOs (\emph{Variables})}\\
\hline
{\minitab[c]{${\bm p}_n$}} & {Payments from the MNO to the first $n$ VOs (\emph{Variables})}\\
\hline
{\minitab[c]{${\bm \pi}_n$}} & {Payoffs of the first $n$ VOs (\emph{Variables})}\\
\hline
{\minitab[c]{$U\left({\bm b}_N,{\bm p}_N\right)$}} & {MNO's payoff function}\\
\hline
{\minitab[c]{$V_n\left({b}_n,{p}_n\right)$}} & {VO $n$'s payoff function}\\
\hline
{\minitab[c]{$\Psi\left({\bm b}_N\right)$}} & {Social welfare function}\\
\hline
{\minitab[c]{$U_n^0,V_n^0$}} & {MNO's and VO $n$'s disagreement points at step $n$}\\
\hline
{\minitab[c]{$U_n^1,V_n^1$}} & {MNO's and VO $n$'s payoffs at step $n$ under bargaining result $\left(b_n,\pi_n\right)$}\\
\hline
{\minitab[c]{$B_{m}^{s}\left({\bm b}_s\right)$}} & {Outcomes of the first $m$ steps when the MNO reaches ${\bm b}_s$ in the first $s$ steps}\\
\hline
{\minitab[c]{${b_k^*\left( {\bm b}_{k-1} \right)}$}} & {Outcome of step $k$ when the MNO reaches ${\bm b}_{k-1}$ in the first $k-1$ steps}\\
\hline
{\minitab[c]{$\pi_{k}^*\left( {\bm b}_{k-1} \right)$}} & {VO $k$'s payoff when the MNO reaches ${\bm b}_{k-1}$ in the first $k-1$ steps}\\
\hline
{\minitab[c]{${\hat {\bm b}}_N,{\hat {\bm \pi}}_N$}} & {NBS of all the $N$ steps}\\
\hline
{\minitab[c]{$U_0$}} & {MNO's eventual payoff after bargaining}\\
\hline
\end{tabular}
\end{table}

\section{One-To-One Bargaining}\label{Section3}
We first study a special case where there is only one VO, \emph{i.e.,} $\left|{\cal N}\right|=1$. We analyze the one-to-one bargaining under the Nash bargaining theory, which helps us better understand the more general results in the later sections.

The \emph{Nash bargaining solution} (NBS) \cite{nash1950bargaining} of the one-to-one bargaining solves the following problem:
\begin{align}
\begin{split}
&{\max} {\rm ~}\left( {U\left( {{b_1},{p_1}} \right) - U\left( {0,0} \right)} \right) \cdot \left( {{V_1}\left( {{b_1},{p_1}} \right) - {V_1}\left( {0,0} \right)} \right)\\
&{\rm s.t.~~}U\!\left( {{b_1},{p_1}} \right)\!-\!U\left( {0,0} \right)\! \ge 0,{V_1}\!\left( {{b_1},{p_1}} \right)\! -\! {V_1}\left( {0,0} \right)\! \ge 0,\\
&{\rm var.~~~~~~~~~~~~~~~~~}{b_1} \in \left\{ {0,1} \right\},{p_1} \in {\mathbb R}.\label{formulation1}
\end{split}
\end{align}
Here, $U\left( {0,0} \right)$ and ${V_1}\left( {0,0} \right)$ are the \emph{disagreement points} of the MNO and VO $1$ (the only VO), which are equal to their payoffs when no agreement is reached. Through setting $b_1 = 0$ and $p_1 = 0$ in (\ref{MNOpayoff}) and (\ref{VOpayoff}), we obtain $U\left( {0,0} \right)=0$ and ${V_1}\left( {0,0} \right)=0$, respectively. The NBS essentially maximizes the product of the MNO and VO $1$'s payoff gains over their disagreement points. Intuitively, with a higher disagreement point, the MNO (or the VO) can obtain a larger payoff under the NBS.

We further define ${\pi _1} \triangleq {V_1}\left( {{b_1},{p_1}} \right)$ as the payoff of VO $1$. This enables us to rewrite problem (\ref{formulation1}) with respect to ${\pi _1}$ and $\Psi \left( {{b_1}} \right)$:
\begin{align}
\begin{split}
&{\max}{\rm{~~~~~~~}} \left( {\Psi \left( {{b_1}} \right) -{\pi _1}} \right) \cdot {\pi _1}\\
&{\rm s.t.~~~~~~~} {\Psi \left( {{b_1}} \right) - {\pi _1}} \ge 0, {\pi _1} \ge 0,\\
&{\rm var.~~~~~~~~~}{b_1} \in \left\{ {0,1} \right\},{\pi _1} \in {\mathbb R}.\label{formulation2}
\end{split}
\end{align}
Problems (\ref{formulation1}) and (\ref{formulation2}) are equivalent, in the sense that given any bargaining solution in terms of $\left( {{b_1},{\pi_1}} \right)$, we can compute the equivalent bargaining solution in terms of $\left( {{b_1},{p_1}} \right)$ as $\left( {b_1,p_1} \right) = \left( {b_1,\pi_1-b_1{R_1}} \right)$ based on (\ref{VOpayoff}).

We show the closed-form optimal solution to (\ref{formulation2}) in the following proposition.{\footnote{The detailed proofs of the propositions and theorems in this paper are given in the appendix.}} 
\begin{proposition}\label{proposition:one2one}
The optimal solution to problem (\ref{formulation2}) is
\begin{align}
\left( {b_1^*,\pi _1^*} \right) = \left\{ {\begin{array}{*{20}{l}}
{\left( {1,\frac{1}{2}\Psi \left( 1 \right)} \right),}&{{\rm if~}\Psi \left( 1 \right) \ge 0},\\
{\left( {0,0} \right),}&{\rm otherwise},
\end{array}} \right.\label{choice1}
\end{align}
where $\Psi \left( 1 \right) = \sum_{t=1}^{T} {f_t\left( {{X_1^t}} \right)} + {Q_1}$ is defined in (\ref{socialwelfare}).
\end{proposition}

Proposition \ref{proposition:one2one} indicates that if reaching an agreement increases the social welfare, \emph{i.e.,} $\Psi \left( 1 \right) \ge \Psi \left( 0 \right) = 0$, the MNO will deploy Wi-Fi at venue $1$ and equally share the generated social welfare with VO $1$; otherwise no Wi-Fi will be deployed, and both the MNO and VO $1$ will obtain zero payoff.

\section{One-to-Many Bargaining with \\Exogenous Sequence}\label{Section4}
In this section, we study the case where the MNO bargains with $N$ VOs sequentially under a fixed sequence. We illustrate the bargaining protocol in Figure \ref{figure1}. At each step, the MNO bargains with one VO $n\in\cal{N}$ on $\left( {{b_n},{p_n}} \right)$.
\begin{figure}[t]
  \centering
  \includegraphics[scale=0.3]{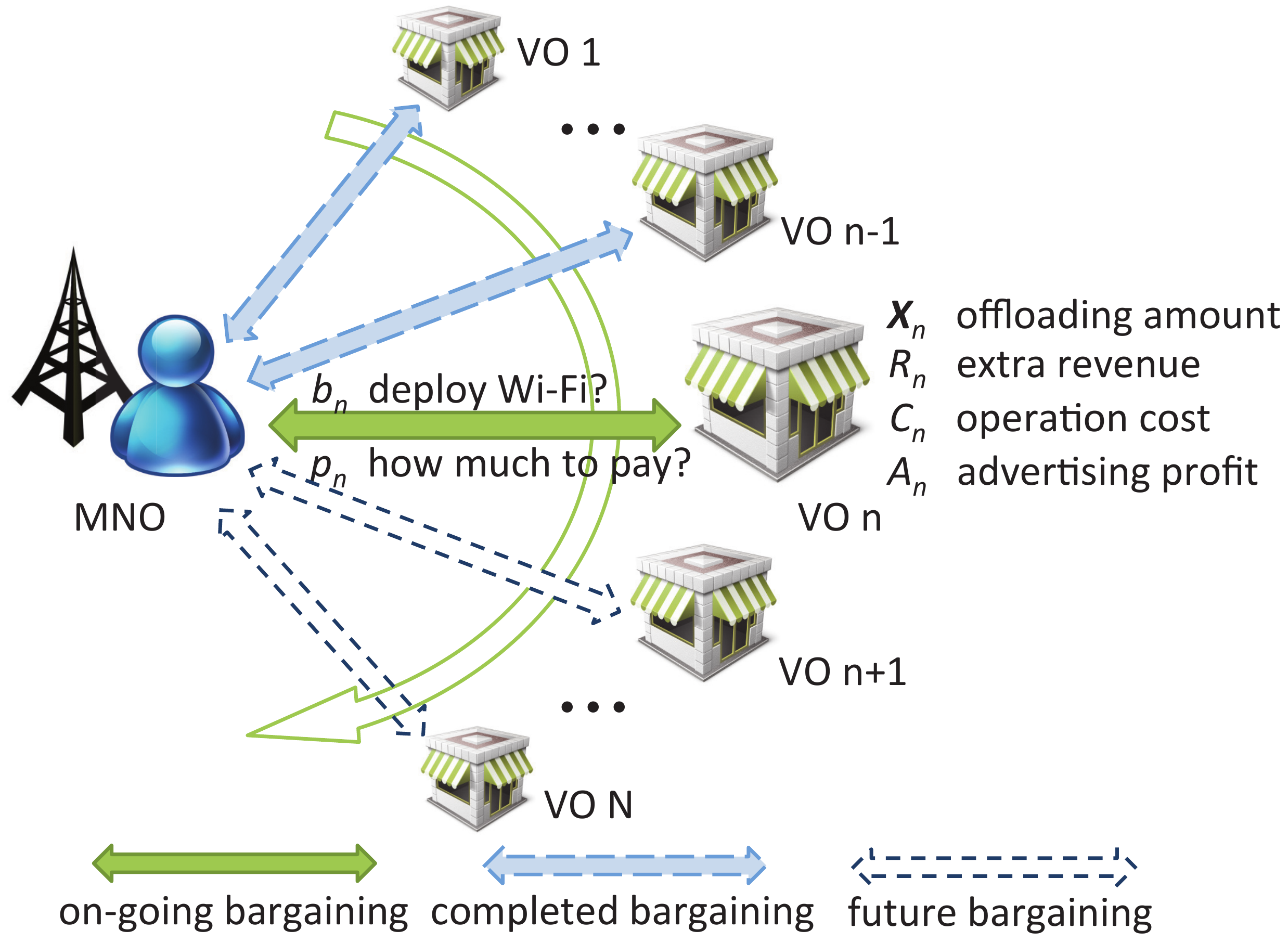}
  \centering
  \caption{Bargaining Protocol.}
  \label{figure1}
\end{figure}

We define ${\pi _n}$ as VO $n\in\cal{N}$'s payoff. As we have discussed in Section \ref{Section3}, bargaining on $\left( {{b_n},{p_n}} \right)$ and bargaining on $\left( {{b_n},{\pi_n}} \right)$ are equivalent. Therefore, in Sections \ref{Section4} and \ref{Section5}, we present the NBS in the form of $\left( {{b_n},{\pi_n}} \right)$ to simplify the notations. Similar to ${{\bm{b}}_n}$ and ${{\bm{p}}_n}$, we define
\begin{align}
{{\bm{\pi}}_n}\triangleq \left(\pi_1,\pi_2,\ldots,\pi_n\right)
\end{align}
as the payoffs of the \emph{first} $n$ VOs.

Without loss of generality, we assume that the bargaining sequence follows $1,2,\ldots,N$, \emph{i.e.,} the MNO bargains with VO $n$ at step $n\in{\cal N}$. In Section \ref{subsec:general}, we formulate the bargaining problem for step $n$. In Section \ref{subsec:NBS}, we apply backward induction to compute the NBS for step $n$.
\subsection{Bargaining Problem for Step $n\in{\cal N}$}\label{subsec:general}
At step $n\in{\cal N}$, the MNO bargains with VO $n$. We define $U_n^0$ and $V_n^0$ as the MNO's and VO $n$'s disagreement points, respectively. Furthermore, when the MNO and VO $n$ agree on $\left( {b_n},{\pi _n} \right)$, we define their payoffs by $U_n^1$ and $V_n^1$, respectively.

Similar as (\ref{formulation1}), we formulate the Nash bargaining problem at step $n$ as
\begin{align}
\begin{split}
&{\max}{\rm ~~~~~~~} \left( {U_{n}^1 - U_{n}^0} \right)\cdot\left( {V_{n}^1 - V_{n}^0} \right)\\
&{\rm s.t.}{\rm ~~~~~~~}{U_{n}^1 - U_{n}^0}\ge 0,{V_{n}^1 - V_{n}^0} \ge 0,\\
&{\rm var.~~~~~~~~~~}{b_{n}} \in \left\{ {0,1} \right\},{\pi _{n}} \in {\mathbb R}.\label{formulation:general}
\end{split}
\end{align}
Because VO $n$ has a zero disagreement point if not reaching an agreement with the MNO, we have $V_n^0=0$. Moreover, based on the definition of $\pi_n$, we have $V_n^1=\pi_n$. However, the computation of $U_n^0$ and $U_n^1$ are challenging, as the MNO's payoff depends on the bargaining results of all the $N$ steps. In the next section, we compute $U_n^0$ and $U_n^1$ by backward induction, and solve problem (\ref{formulation:general}) to obtain the NBS for step $n$.

\subsection{NBS for Step $n\in{\cal N}$}\label{subsec:NBS}
We use backward induction to solve problem (\ref{formulation:general}) from step $n=N$ to step $n=1$.

\subsubsection{Step $N$}
Suppose that the MNO has already bargained with VO $1,\ldots,N-1$, and has reached ${{\bm b}_{N - 1}}$ and ${{\bm \pi}_{N - 1}}$. It now bargains with VO $N$.

The MNO's disagreement point is
\begin{align}
U_N^0 = \Psi \left( {{\bm b}_{N-1},0}\right) - \sum\limits_{m = 1}^{N - 1} {{\pi _m}}.\label{equ:MNO:dis:N}
\end{align}
Here, $\Psi \left( {{\bm b}_{N-1},0} \right)$ is the social welfare when the bargaining outcomes of all $N$ steps are given as $\left( {{\bm b}_{N-1},0}\right)$, \emph{i.e.}, assuming that no agreement is reached in step $N$. We obtain $U_N^0$ by subtracting the first $N-1$ VOs' payoffs from the social welfare.\footnote{Notice that when no agreement is reached in step $N$, we have $\pi_N=0$. Hence, we do not need to subtract $\pi_N$ from the social welfare in (\ref{equ:MNO:dis:N}).}

If the MNO reaches $\left( {{b_N},{\pi _N}} \right)$ with VO $N$ in step $N$, its payoff is
\begin{align}
U_N^1 = \Psi \left({{{\bm b}_{N - 1}},{b_N}} \right)- \sum\limits_{m = 1}^{N - 1} {{\pi _m}}-{\pi_N}.\label{equ:MNO:agr:N}
\end{align}
Here, $\Psi \left({{\bm b}_{N-1},b_N} \right)$ is the social welfare when the bargaining outcomes are given as $\left( {{\bm b}_{N-1},b_N}\right)$. We obtain $U_N^1$ by subtracting all VOs' payoffs from the social welfare.

Recall that $V_N^0=0$ and $V_N^1=\pi_N$. Based on $U_N^0$ in (\ref{equ:MNO:dis:N}) and $U_N^1$ in (\ref{equ:MNO:agr:N}), we solve problem (\ref{formulation:general}) for $n=N$ and obtain the NBS for step $N$:
\begin{align}
\nonumber
& \left( {b_N^*\left({{\bm b}_{N-1}}\right),\pi _N^*\left({{\bm b}_{N-1}}\right)} \right) = \\
& \left\{ {\begin{array}{*{20}{l}}
{\left( {1,\frac{1}{2}{\Delta _N}\left( {{\bm b}_{N-1}} \right)} \right),}&{{\rm if~} {\Delta_N}\left({{\bm b}_{N-1}}\right)\ge 0},\\
{\left( {0,0} \right),}&{\rm otherwise},
\end{array}} \right.\label{solutionN}
\end{align}
where we define
\begin{align}
{\Delta _N}\left( {{{\bm b}_{N-1}}} \right) \triangleq \Psi \left( {{{\bm b}_{N-1}},1} \right) - \Psi \left( {{{\bm b}_{N-1}},0} \right).\label{Njude}
\end{align}
Here, ${\Delta _N}\left( {{{\bm b}_{N-1}}} \right)$ can be understood as follows: if we treat the MNO and VO $N$ as a coalition, ${\Delta _N}\left( {{{\bm b}_{N-1}}} \right)$ describes the increase in the coalition's payoff by deploying Wi-Fi at venue $N$. If and only if such a value is non-negative, the MNO and VO $N$ will reach an agreement and equally share the generated revenue; otherwise no agreement is reached. This is similar as the one-to-one bargaining in Section \ref{Section3}.

We can also understand ${\Delta _N}\left( {{{\bm b}_{N-1}}} \right)$ as the increase in social welfare by deploying Wi-Fi at venue $N$. This is because VO $N$ is the last one that the MNO bargains with. For a general bargaining step $n\in{\cal N}$, we will later show that ${\Delta _n}\left( {{{\bm b}_{n-1}}} \right)$ is generally \emph{not} equal to the increase in social welfare by deploying Wi-Fi at venue $n$.

Based on (\ref{solutionN}), $\left( {b_N^*\left({{\bm b}_{N-1}}\right),\pi _N^*\left({{\bm b}_{N-1}}\right)} \right)$ depends on vector ${{\bm b}_{N-1}}$ but is independent of vector ${{\bm \pi}_{N-1}}$. This means that the NBS for step $N$ only depends on the first $N-1$ steps' bargaining outcomes, and not on the VOs' payoffs.

\vspace{-0.1cm}

\subsubsection{Step $N-1$}

Suppose that the MNO has already bargained with VO $1,\ldots,N-2$, and has reached ${{\bm b}_{N - 2}}$ and ${{\bm \pi}_{N - 2}}$. It now bargains with VO $N-1$.

The MNO's disagreement point is
\begin{align}
\nonumber
&U_{N - 1}^0=\Psi \left( {{{\bm b}_{N-2}},0,{b_N^*}\left( {{{\bm b}_{N- 2}},0} \right)} \right)\\
&{~~~~~~~~~~~~~}- \sum\limits_{m = 1}^{N -2} {{\pi _m}} - {\pi _N^*}\left({{{\bm b}_{N- 2}},0}\right).\label{equ:MNO:dis:N-1}
\end{align}
Here, $\Psi \left( {{{\bm b}_{N-2}},0,{b_N^*}\left( {{{\bm b}_{N- 2}},0} \right)} \right)$ is the social welfare when the MNO reaches ${\bm b}_{N-2}$ with the first $N-2$ VO, does not reach an agreement with VO $N-1$, and reaches ${b_N^*}\left( {{{\bm b}_{N - 2}},0} \right)$ with VO $N$. We obtain $U_{N-1}^0$ by subtracting VOs' payoffs from the social welfare. Notice that ${b_N^*}\left( {{{\bm b}_{N - 2}},0} \right)$ and ${\pi_N^*}\left( {{{\bm b}_{N - 2}},0} \right)$ together correspond to the NBS for step $N$ when the bargaining outcomes of the first $N-1$ steps are $\left( {{{\bm b}_{N - 2}},0}\right)$, as computed by (\ref{solutionN}).

If the MNO reaches $\left( {{b_{N-1}},{\pi _{N-1}}} \right)$ with VO $N-1$ in step $N-1$, its payoff is
\begin{align}
\nonumber
& U_{N - 1}^1=\Psi \left( {{{\bm b}_{N-2}},b_{N-1},{b_N^*}\left( {{{\bm b}_{N - 2}},b_{N-1}} \right)} \right)\\
&{~~~~~~~~~~}- \sum\limits_{m = 1}^{N - 2} {{\pi _m}}  - {\pi _{N - 1}}- {\pi _N^*}\left({{{\bm b}_{N - 2}},b_{N-1}}\right).\label{equ:MNO:agr:N-1}
\end{align}
Here ${b_N^*}\left( {{{\bm b}_{N-2}},b_{N-1}} \right)$ and ${{\pi}_N^*}\left( {{{\bm b}_{N-2}},b_{N-1}} \right)$ are also determined by (\ref{solutionN}).

Based on $U_{N-1}^0$ in (\ref{equ:MNO:dis:N-1}) and $U_{N-1}^1$ in (\ref{equ:MNO:agr:N-1}), we solve problem (\ref{formulation:general}) for $n=N-1$ and obtain the NBS for step $N-1$:
\begin{align}
\nonumber
&\left( {b_{N-1}^*\left( {{{\bm b}_{N-2}}} \right),\pi _{N-1}^*\left( {{{\bm b}_{N-2}}} \right)} \right) =\\
& \left\{ {\begin{array}{*{20}{l}}
{\left( {1,\frac{1}{2}{\Delta _{N-1}}\left( {{{\bm b}_{N-2}}} \right)} \right),}&{{\rm if~}{\Delta _{N-1}}\left( {{{\bm b}_{N-2}}} \right) \ge 0},\\
{\left( {0,0} \right),}&{\rm otherwise},
\end{array}} \right.\label{solutionN-1}
\end{align}
where we define
\begin{align}
\nonumber
&{\Delta _{N-1}}\left( {{{\bm b}_{N-2}}} \right)\triangleq \Psi \left( {{{\bm b}_{N - 2}},1,b_N^*\left( {{{\bm b}_{N - 2}},1} \right)} \right)-\pi _N^*\left( {{{\bm b}_{N - 2}},1} \right)\\
&{\rm~~~~~~~~~~}- \Psi \left( {{{\bm b}_{N - 2}},0,{b_N^*}\left( {{{\bm b}_{N - 2}},0} \right)} \right) + {\pi _N^*}\left( {{{\bm b}_{N - 2}},0} \right).\label{N-1jude}
\end{align}
If we treat the MNO and VO $N-1$ as a coalition, then ${\Delta _{N-1}}\left( {{{\bm b}_{N - 2}}} \right)$ describes the increase in the coalition's payoff by deploying Wi-Fi at venue $N-1$, taking into account VO $N$'s response.

\subsubsection{Step $k$, $k\in\left\{2,3,\ldots,N-2\right\}$}
Suppose that the MNO has bargained with VO $1,\ldots,k-1$, and has reached ${{\bm b}_{k\!-1}}$ and ${{\bm \pi}_{k-\!1}}$. It now bargains with VO $k$.

For ease of exposition, we define $B_{m}^{s}\left({\bm b}_s\right), m\ge s, m,s\in{\cal N}$, as
\begin{align}
\nonumber
& B_{m}^{s}\left({\bm b}_s\right)=\\
& \left\{ {\begin{array}{*{20}{l}}
{{\bm b}_s,}&{{\rm if~}m=s},\\
{\left( B_{m-1}^{s}\left({\bm b}_s\right), b_{m}^*\left( B_{m-1}^{s}\left({\bm b}_s\right) \right)\right),}&{{\rm if~}m=s+1,\ldots,N.}
\end{array}} \right.\label{equ:Bfunction}
\end{align}
Intuitively, $B_{m}^{s}\left({\bm b}_s\right)$ characterizes the bargaining outcomes of the first $m$ ($m \ge s$) steps when the MNO reaches ${\bm b}_s$ in the first $s$ steps.{\footnote{Notice that, $B_{m-1}^{s}\left({\bm b}_s\right)$ in (\ref{equ:Bfunction}) returns a vector with a length of $m-1$, and $b_{m}^*\left( B_{m-1}^{s}\left({\bm b}_s\right) \right)$ is the bargaining outcome computed in step $m$.}}

Based on $B_{m}^{s}\left({\bm b}_s\right)$, we can write the MNO's disagreement point at step $k$ as
\begin{align}
\nonumber
&U_{k}^0=\Psi \left( B_{N}^{k}\left({\bm b}_{k-1},0\right) \right)-\sum\limits_{m = 1}^{k-1} {{\pi _m}}\\
&{~~~~~~~~~~} -\sum\limits_{m = k+1}^{N} {{\pi_m^*}\left( B_{m-1}^{k}\left({\bm b}_{k-1},0\right) \right)}.\label{equ:MNO:dis:k}
\end{align}
Here, $B_{N}^{k}\left({\bm b}_{k-1},0\right)$ describes the bargaining outcomes of all the $N$ steps when the MNO reaches $\left({\bm b}_{k-1},0\right)$ with the first $k$ VOs. Based on (\ref{equ:Bfunction}), this is computed in a recursive manner. For example, from (\ref{equ:Bfunction}), we have $B_{N}^{k}\left({\bm b}_{k-1},0\right)=\left( B_{N-1}^{k}\left({\bm b}_{k-1},0\right), b_{N}^*\left( B_{N-1}^{k}\left({\bm b}_{k-1},0\right) \right)\right)$, where $B_{N-1}^{k}\left({\bm b}_{k-1},0\right)$ can be further obtained by using (\ref{equ:Bfunction}), and $b_{N}^*\left( B_{N-1}^{k}\left({\bm b}_{k-1},0\right) \right)$ is computed by (\ref{solutionN}). Term $\Psi \left( B_{N}^{k}\left({\bm b}_{k-1},0\right) \right)$ is the social welfare under the bargaining outcomes given by $B_{N}^{k}\left({\bm b}_{k-1},0\right)$. Furthermore, $\sum_{m = 1}^{k-1} {{\pi _m}}$ is the total payoff of the first $k-1$ VOs, and term $\sum_{m = k+1}^{N} {{\pi_m^*}\left( B_{m-1}^{k}\left({\bm b}_{k-1},0\right) \right)}$ is the total payoff of VOs $k+1,k+2,\ldots,N$. Notice that term ${\pi_m^*}\left( B_{m-1}^{k}\left({\bm b}_{k-1},0\right) \right),m=k+1,k+2,\ldots,N,$ denotes VO $m$'s payoff, and is a function of the bargaining outcomes of the first $m-1$ steps. In (\ref{equ:MNO:dis:k}), we compute $U_k^0$ by subtracting all VOs' payoffs from the social welfare.

If the MNO reaches $\left( {{b_{k}},{\pi _{k}}} \right)$ with VO $k$, its payoff is
\begin{align}
\nonumber
&U_{k}^1=\Psi \left( B_{N}^{k}\left({\bm b}_{k-1},b_k\right) \right)- \sum\limits_{m = 1}^{k-1} {{\pi _m}}-\pi_k\\
&{~~~~~~~~~~} -\sum\limits_{m = k+1}^{N} {{\pi_m^*}\left( B_{m-1}^{k}\left({\bm b}_{k-1},b_k\right) \right)}.\label{equ:MNO:agr:k}
\end{align}

Based on $U_k^0$ in (\ref{equ:MNO:dis:k}) and $U_k^1$ in (\ref{equ:MNO:agr:k}), we solve problem (\ref{formulation:general}) for $n=k$ and obtain the NBS for step $k$:
\begin{align}
\nonumber
&\left( {b_{k}^*\left( {{{\bm b}_{k-1}}} \right),\pi _{k}^*\left( {{{\bm b}_{k-1}}} \right)} \right)=\\
& \left\{ {\begin{array}{*{20}{l}}
{\left( {1,\frac{1}{2}{\Delta _{k}}\left( {{{\bm b}_{k - 1}}} \right)} \right),}&{{\rm if~}{\Delta _{k}}\left( {{{\bm b}_{k - 1}}} \right)\ge0},\\
{\left( {0,0} \right),}&{\rm otherwise},
\end{array}} \right.\label{solutionk}
\end{align}
where we define
\begin{align}
\nonumber
&{\Delta _{k}}\left( {{\bm b}_{k - 1}} \right) \triangleq  \Psi \left( B_{N}^{k}\left({\bm b}_{k-1},1\right) \right) -\!\!\sum\limits_{m = k+1}^{N} {{\pi_m^*}\left( B_{m-1}^{k}\left({\bm b}_{k-1},1\right) \right)}\\
& -\Psi \left( B_{N}^{k}\left({\bm b}_{k-1},0\right) \right)+\sum\limits_{m = k+1}^{N} {{\pi_m^*}\left( B_{m-1}^{k}\left({\bm b}_{k-1},0\right) \right)}.\label{kjude}
\end{align}
If we treat the MNO and VO $k$ as a coalition, ${\Delta _{k}}\left( {{{\bm b}_{k - 1}}} \right)$ characterizes the increase of the coalition's payoff by deploying Wi-Fi at venue $k$, considering the responses of VOs $k+1,\ldots,N$.

\subsubsection{Step $1$}
The analysis of step $1$ is similar to that of step $k$, $k=2,3,\ldots,N-2$, except that for step $1$, there is no prior bargaining outcome. To save space, we skip the computation of $U_1^0$ and $U_1^1$, and provide the NBS as follows:
\begin{align}
\left( {b_{1}^*,\pi _{1}^*} \right)= \left\{ {\begin{array}{*{20}{l}}
{\left( {1,\frac{1}{2}{\Delta _{1}}} \right),}&{{\rm if~}{\Delta _{1}}\ge0},\\
{\left( {0,0} \right),}&{\rm otherwise},
\end{array}} \right.\label{solution1}
\end{align}
where we define
\begin{align}
\nonumber
{\Delta _{1}} \triangleq &\Psi \left( B_{N}^{1}\left(1\right) \right) -\sum\limits_{m = 2}^{N} {{\pi_m^*}\left( B_{m-1}^{1}\left(1\right) \right)}\\
& -\Psi \left( B_{N}^{1}\left(0\right) \right)+\sum\limits_{m = 2}^{N} {{\pi_m^*}\left( B_{m-1}^{1}\left(0\right) \right)}.\label{equ:delta1}
\end{align}
\subsection{MNO's Payoff after Bargaining}
After applying backward induction to the analysis from step $N$ to $1$, we can eventually obtain the bargaining outcomes in all steps and all VOs' payoffs, and we denote them by ${\hat {\bm b}}_N=\left( {{{\hat b}_1}, \ldots ,{{\hat b}_N}} \right)$ and ${\hat {\bm \pi}}_N=\left( {{{\hat \pi}_1}, \ldots ,{{\hat \pi}_N}} \right)$. 
Based on ${\hat {\bm b}}_N$ and ${\hat {\bm \pi}}_N$, we can easily compute the MNO's eventual payoff as
\begin{align}
\nonumber
U_0 &=\Psi \left( {\hat{\bm b}}_N \right) - \sum\limits_{n = 1}^N {{{\hat \pi} _n}}\\
&=\sum\limits_{t=1}^T {f_t\left( {\sum\limits_{n = 1}^N {{{\hat b}_n}{X_n^t}} } \right)} + \sum\limits_{n = 1}^N {{{\hat b}_n}{Q_n}}  - \sum\limits_{n = 1}^N {{{\hat \pi }_n}}.\label{MNOpayoffsequential}
\end{align}

\subsection{Engineering Insights}\label{subsec:engineering}
Here we summarize the insights from the above analysis of the one-to-many bargaining under a fixed bargaining sequence.

First, we find that the NBS of a particular step depends on the Wi-Fi deployment decisions of all the prior bargaining steps. This is because the more Wi-Fi networks the MNO has already deployed, the less motivation it has to deploy a new Wi-Fi network. On the other hand, since such a negative externality is not related to the payments among the MNO and VOs, the NBS of a particular step is independent of the payments of all the prior bargaining steps.

Second, the MNO may cooperate with the VOs nonconsecutively. As we will discuss in Example 3 in Section \ref{optimalsequencingalgorithm}, under a particular bargaining sequence, the MNO does not cooperate with a VO in the middle, while reaching agreements with VOs before and after the middle VO.

\section{One-to-Many Bargaining with \\Endogenous Sequence}\label{Section5}
In this section, we study the one-to-many bargaining with {endogenous} sequence, where the bargaining sequence is selected by the MNO to maximize its payoff. In Section \ref{Section5Z}, we illustrate the influence of the bargaining sequence on the MNO's payoff through two examples. In Section \ref{Section5A}, we formulate the MNO's optimal bargaining sequencing problem. In Section \ref{optimalsequencingalgorithm}, we solve the problem through an \emph{Optimal VO Bargaining Sequencing} (OVBS) algorithm. In Sections \ref{SpecialA} and \ref{SpecialB}, we study two special cases, where we can explicitly determine the optimal bargaining sequence without running OVBS.
\subsection{Examples on the Influence of Bargaining Sequence}\label{Section5Z}
Based on the analysis in Section \ref{Section4}, we present two examples in Figure \ref{figureb} to illustrate that the bargaining sequence can significantly affect the bargaining solutions and the MNO's payoff.

\begin{figure}[t]
  \centering
  \includegraphics[scale=0.3]{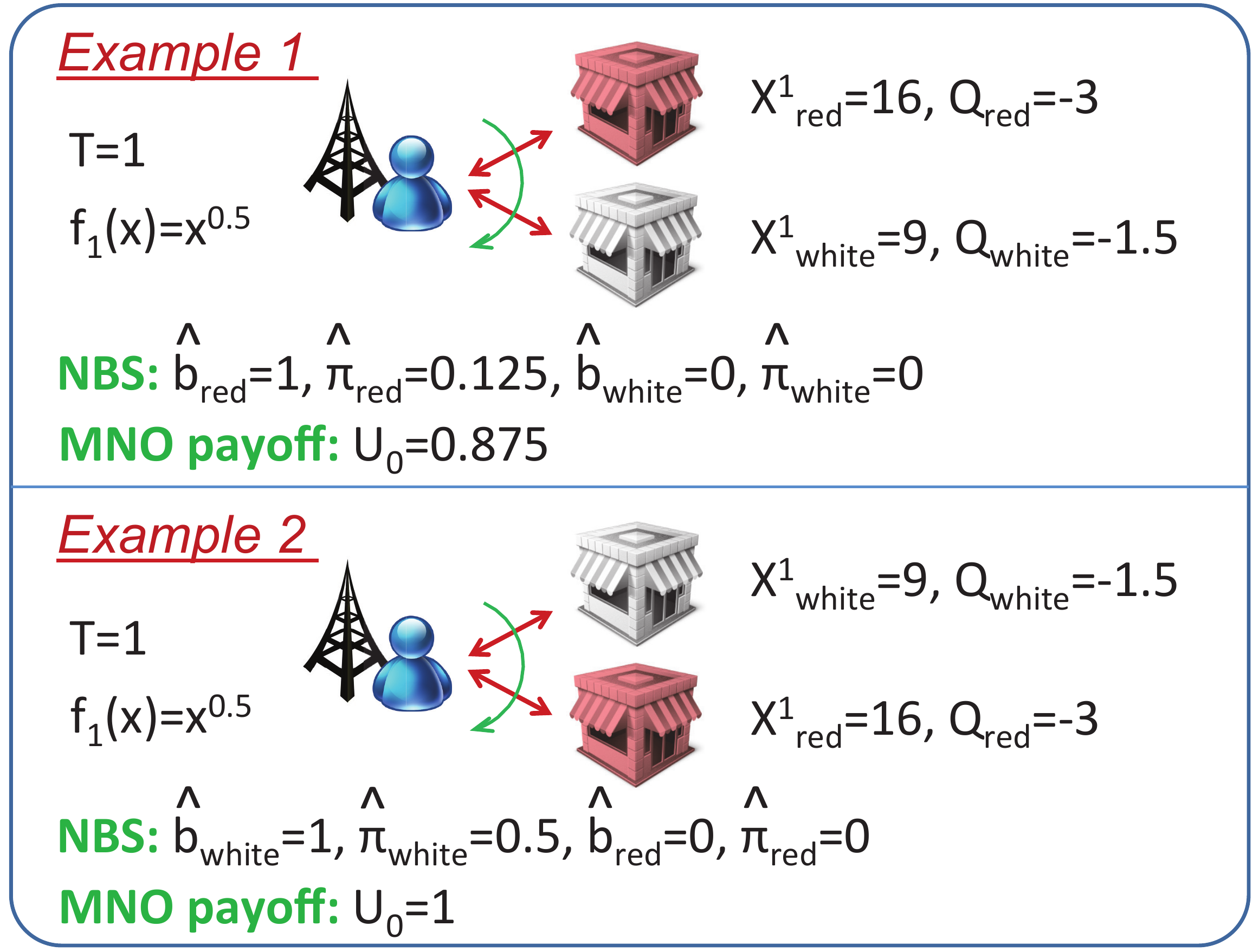}
  \centering
  \caption{Influence of Bargaining Sequence on MNO's Payoff.}
  \label{figureb}
\end{figure}

\begin{example}
The MNO first bargains with VO {red} and then bargains with VO {white}. We apply the backward induction and start the analysis from step 2. We first consider the case where the MNO reaches an agreement with VO {red} in step 1. By taking $N=2$ and $b_1=1$ in (\ref{Njude}), we have ${\Delta _2}\left( 1 \right) =  \Psi \left( {1,1} \right) - \Psi \left( {1,0} \right) =\sqrt {16+9}  - \sqrt 16  - 1.5 < 0$. Hence, we obtain from (\ref{solutionN}) that ${b_2^*}\left( 1 \right) = 0, {\pi_2^*}\left( 1 \right) = 0$, {i.e.}, the MNO does not cooperate with VO {white} in this case. We further consider the case where the MNO does not reach an agreement with VO {red} in step 1. By taking $N=2$ and $b_1=0$ in (\ref{Njude}), we have ${\Delta _2}\left( 0 \right) =  \Psi \left( {0,1} \right) - \Psi \left( {0,0} \right) =\sqrt {0+9}-\sqrt{0}  - 1.5=1.5 > 0$. Hence, we obtain from (\ref{solutionN}) that ${b_2^*}\left( 0 \right) = 1, {\pi_2^*}\left( 0 \right) = \frac{1}{2}{\Delta _2}\left( 0 \right)=0.75$, {i.e.}, the MNO cooperates with VO {white} in this case, and VO {white}'s payoff is $0.75$.

Next we come to the analysis of step 1, where the MNO bargains with VO {red}. Based on ${b_2^*}\left( 1 \right)$, ${\pi_2^*}\left( 1 \right)$, ${b_2^*}\left( 0 \right)$, and ${\pi_2^*}\left( 0 \right)$, we take $N=2$ in (\ref{equ:delta1}) and compute ${\Delta _1}$ as
\begin{align}
\nonumber
{\Delta _1}&=\Psi\left(1,b_2^*\left(1\right)\right)-\pi_2^*\left(1\right)-\Psi\left(0,b_2^*\left(0\right)\right)+\pi_2^*\left(0\right)\\
\nonumber
& =\left(\sqrt {16 } -3\right) -0 - \left(\sqrt {9} -1.5\right) + 0.75\\
&=0.25.
\end{align}
Since ${\Delta _1}>0$, based on (\ref{solution1}), we have ${b_1^*} = 1$ and ${\pi_1^*} =\frac{1}{2}\Delta _1= 0.125$. Therefore, the eventual bargaining outcome is ${{\hat b}_1} = 1$, ${{\hat b}_2} = 0$, ${{\hat \pi}_1} = 0.125$, and ${{\hat \pi}_2} = 0$. The MNO's eventual payoff is ${U_0} = \sqrt 16  - 3 - 0.125 = 0.875$.
\end{example}

\begin{example}
The MNO first bargains with VO {white} and then bargains with VO {red}. We start the analysis from step 2. We first consider the case that the MNO reaches an agreement with VO {white} in step 1, we have ${\Delta _2}\left( 1 \right) =  \Psi \left( {1,1} \right) - \Psi \left( {1,0} \right) =\sqrt {9+16}  - \sqrt 9  - 3 < 0$. Hence, we obtain ${b_2^*}\left( 1 \right) = 0, {\pi_2^*}\left( 1 \right) = 0$, {i.e.}, the MNO does not cooperate with VO {red} in this case. We further consider the case that the MNO does not reach an agreement with VO {white} in step 1, we have ${\Delta _2}\left( 0 \right) =   \Psi \left( {0,1} \right) - \Psi \left( {0,0} \right) = \sqrt {0+16}- \sqrt {0}  - 3 =1> 0$. Hence, we obtain ${b_2^*}\left( 0 \right) = 1, {\pi_2^*}\left( 0 \right) = \frac{1}{2}{\Delta _2}\left( 0 \right)=0.5$, {i.e.}, the MNO cooperates with VO {red} in this case, and VO {red}'s payoff is $0.5$.

Next we come to the analysis of step 1, where the MNO bargains with VO {white}. Based on ${b_2^*}\left( 1 \right)$, ${\pi_2^*}\left( 1 \right)$, ${b_2^*}\left( 0 \right)$, and ${\pi_2^*}\left( 0 \right)$, from (\ref{equ:delta1}), we can compute ${\Delta _1}$ as
\begin{align}
\nonumber
{\Delta _1}&=\Psi\left(1,b_2^*\left(1\right)\right)-\pi_2^*\left(1\right)-\Psi\left(0,b_2^*\left(0\right)\right)+\pi_2^*\left(0\right)\\
\nonumber
& =\left(\sqrt {9 } -1.5\right) -0 - \left(\sqrt {16} -3\right) + 0.5\\
&=1.
\end{align}
Since ${\Delta _1}>0$, we have ${b_1^*} = 1$ and ${\pi_1^*} =\frac{1}{2}\Delta _1= 0.5$. Therefore, the eventual bargaining outcome is ${{\hat b}_1} = 1$, ${{\hat b}_2} = 0$, ${{\hat \pi}_1} = 0.5$, and ${{\hat \pi}_2} = 0$. The MNO's eventual payoff is ${U_0} = \sqrt 9  - 1.5 - 0.5 = 1$.
\end{example}

Comparing Example 1 and Example 2, we find that the MNO obtains different payoffs under different bargaining sequences. Through exchanging the bargaining positions of the two VOs (\emph{red} and \emph{white}), the MNO's payoff ${U}_0$ improves from $0.875$ to $1$.
{{This is due to the cooperation cost and the externality between the two bargaining steps. In our problem, the cooperation cost is the cost of deploying and operating Wi-Fi, which is denoted by $C_n$ and has been included in $Q_n$ based on (\ref{Qn}). Because of the cooperation cost, the MNO may not choose to cooperate with all VOs.{\footnote{{{As we will discuss in Section \ref{SpecialA}, references \cite{gao2014bargaining} and \cite{moresi2008model} did not consider the cooperation cost, in which case the buyer's payoff is independent of the bargaining sequence.}}}} Moreover, the externality couples the analysis of the two bargaining steps, and makes the bargaining results dependent on the bargaining sequence.}}

\subsection{Optimal Sequencing Problem}\label{Section5A}
We use $\bm l=\left(l_1,l_2,\ldots,l_N\right)$ to denote the bargaining sequence, \emph{i.e.,} the MNO bargains with VO $l_n\in\cal N$ at step $n$. We further define $\cal L$ as the set of all possible bargaining sequences:
\[{\cal L} \triangleq \left\{ {{\bm l}\!:\!{l_i},{l_j}\! \in\! {\cal N} \!{\rm ~and~}\! {l_i} \ne {l_j},\forall i \ne j,i,j \in {\cal N}} \right\}.\]
We use $U_0^{\bm l}$ to denote the MNO's payoff in (\ref{MNOpayoffsequential}) under bargaining sequence ${\bm l}\in\cal L$. The MNO's optimal sequencing problem is

\begin{align}
\mathop {\max }\limits_{{\bm l}\in\cal L} U_0^{\bm l},\label{optimalsequencing}
\end{align}
\emph{i.e.,} choosing the optimal sequence ${\bm l}^*$ to maximize its payoff.

To solve (\ref{optimalsequencing}), we may apply the exhaustive search to compute the MNO's payoff for each ${\bm l}\in\cal L$ and determine $\bm l^*$ accordingly. Since $\left|{\cal L}\right|=N!$, the computational complexity of this method is high. In the next section, we prove an important structural property for the one-to-many bargaining, which allows us to design an \emph{Optimal VO Bargaining Sequencing} (OVBS) algorithm with a significantly lower complexity.

\subsection{Structural Property and OVBS Algorithm}\label{optimalsequencingalgorithm}
We categorize VOs into three types:
\begin{definition}\label{definition1}
VO $n\in\cal{N}$ belongs to

(i) {Type $1$}, if ${Q_n} \ge 0$;

(ii) {Type $2$}, if ${Q_n} < 0$ and $\sum_{t=1}^{T} {f_t\left( {{X_n^t}} \right)} + {Q_n} \ge 0$;

(iii) {Type $3$}, if $\sum_{t=1}^{T} {f_t\left( {{X_n^t}} \right)} + {Q_n} < 0$.{\footnote{Notice that since $\sum_{t=1}^{T} {f_t\left( {{X_n^t}} \right)}\ge0$, condition $\sum_{t=1}^{T} {f_t\left( {{X_n^t}} \right)} + {Q_n} < 0$ implies that $Q_n<0$ for type $3$ VOs.}}
\end{definition}

Recall that $Q_n$ is the net benefit of deploying Wi-Fi at venue $n$ without considering the data offloading benefit. Term $\sum_{t=1}^{T} {f_t\left( {{X_n^t}} \right)}$ is the offloading benefit brought by deploying Wi-Fi at venue $n$ when the MNO does not deploy Wi-Fi at other venues. Since function $f_t\left(\cdot\right)$ is concave for all $t=1,2,\ldots,T$, term $\sum_{t=1}^{T} {f_t\left( {{X_n^t}} \right)}$ can also be understood as the maximum possible offloading benefit brought by deploying Wi-Fi at venue $n$.

Based on the definition of the social welfare (\ref{socialwelfare}), the categorization in Definition \ref{definition1} can be understood as follows:
\begin{itemize}
\item For type $1$ VO $n$, its cooperation with the MNO does not decrease the social welfare, \emph{i.e.}, $\Psi \left(  { {{b_1}, \ldots ,{b_{n - 1}},1,{b_{n + 1}}, \ldots ,{b_N}} }  \right)  \ge  \Psi \left(  {b_1}, \ldots ,{b_{n  -  1}},0,{b_{n  +  1}}, \ldots ,{b_N} \right)$ for all $\left(b_1,\ldots,\!b_{n-1},\!\right.$\\$\left.b_{n+1},\ldots,b_N\right)$;
\item For type $2$ VO $n$, its cooperation with the MNO may or may not decrease the social welfare, which depends on other VOs' parameters and bargaining positions;
\item For type $3$ VO $n$, its cooperation with the MNO decreases the social welfare, \emph{i.e.}, $\Psi \left(  { {{b_1}, \ldots ,{b_{n - 1}},1,{b_{n + 1}}, \ldots ,{b_N}} }  \right)  <  \Psi \left(  {b_1}, \ldots ,{b_{n  -  1}}, \right.$\\$\left. 0,{b_{n  +  1}}, \ldots ,{b_N} \right)$ for all $\left(b_1,\ldots,b_{n-1},b_{n+1},\ldots,b_N\right)$.
\end{itemize}
We assume that the number of each type of VOs is ${N_1}$, $N_2$, and $N_3$, respectively, with ${N_1} + {N_2} + {N_3} = N$. We have the following propositions.
\begin{proposition}\label{proposition1}
The MNO will always cooperate with a type $1$ VO, regardless of such a VO's position in the bargaining sequence.
\end{proposition}
\begin{proposition}\label{proposition2}
The MNO will never cooperate with a type $3$ VO, regardless of such a VO's position in the bargaining sequence.
\end{proposition}
\begin{proposition}\label{proposition3}
If the bargaining sequence follows $1,2,\ldots,N$, and VO $k$ belongs to type $1$, where $k\in\left\{2,3,\ldots,N\right\}$, the MNO's payoff does not decrease after exchanging VOs $k-1$ and $k$'s bargaining positions.
\end{proposition}
\begin{proposition}\label{proposition4}
If the bargaining sequence follows $1,2,\ldots,N$, and VO $k$ belongs to type $3$, where $k\in\left\{2,3,\ldots,N\right\}$, the MNO's payoff does not change after exchanging VOs $k-1$ and $k$'s bargaining positions.
\end{proposition}

Now we are ready to state our main theorem, which describes the structural property of the optimal bargaining sequence.

\begin{theorem}\label{theoremA}
There exists a non-empty set of optimal bargaining sequences $\mathcal{L}^{\ast} \subseteq \mathcal{L}$, such that any $\boldsymbol{l} \in \mathcal{L}^\ast$ satisfies both of the following two conditions:{\footnote{Naturally, VO ${l_{N_1+1}},{l_{N_1+2}}, \ldots ,{l_{{N_1+N_2}}}$ are of type 2 when these two conditions are satisfied.}}

(i) VO ${l_1},{l_2}, \ldots ,{l_{{N_1}}}$ are of type 1;

(ii) VO ${l_{N_1+N_2+1}},{l_{N_1+N_2+2}},\ldots ,{l_{{N}}}$ are of type 3.

For any optimal sequence $\boldsymbol{l} \in \mathcal{L}^{\ast}$,

(i) if the MNO interchanges the bargaining positions of any two type 1 VOs, the MNO's payoff will not change;

(ii) if the MNO interchanges the bargaining positions of any two type 3 VOs, the MNO's payoff will not change.
\end{theorem}

Notice that there may exist some optimal bargaining sequences that are not in set $\mathcal{L}^\ast$. Since our focus is to maximize the MNO's payoff by a properly chosen sequence, we will focus on set $\mathcal{L}^{\ast}$ in the rest of this paper.

\begin{algorithm}[t]
\begin{algorithmic}[1]
\caption{\emph{Optimal VO Bargaining Sequencing} (OVBS)}\label{algorithm2}
\STATE {\underline{\textbf{Phase $1$}: Construct the reduced set ${\cal L}^{R\!E}$}}
\STATE Order all type $1$ VOs arbitrarily, and denote the sequence by a vector $\bm h^1=\left(h_1^1,h_2^1\ldots,h_{N_1}^1\right)$;
\STATE Order all type $3$ VOs arbitrarily, and denote the sequence by a vector $\bm h^3=\left(h_1^3,h_2^3\ldots,h_{N_3}^3\right)$;
\STATE Denote the set of all permutations of type $2$ VOs by set ${\cal H}^2$. Each permutation is denoted by a vector $\bm h^2=\left(h_1^2,h_2^2\ldots,h_{N_2}^2\right)\in{\cal H}^2$.
\STATE Pick every ${\bm h^2}\in{\cal H}^2$ and construct the corresponding total sequencing by ${\bm l}\!=\!\left({\bm h^1},\!{\bm h^2},\!{\bm h^3}\right)$. Denote the set of all such $\bm l$s as ${\cal L}^{R\!E}$.
\STATE {\underline{\textbf{Phase $2$}: Search the optimal sequence}}
\STATE Apply the backward induction and (\ref{MNOpayoffsequential}) in Section \ref{Section4} to compute $U_0^{\bm l}$ for each ${\bm l}\!\in\!{{\cal L}^{R\!E}}$ and return ${\bm l}^{R\!E}\!=\!\mathop {\arg\!\max }_{{\bm l} \in {{\cal L}^{R\!E}}}\! U_0^{\bm l}$.
\end{algorithmic}
\end{algorithm}

Based on Theorem \ref{theoremA}, we propose an \emph{Optimal VO Bargaining Sequencing} (OVBS) algorithm (\emph{i.e.}, Algorithm \ref{algorithm2}), which solves the optimal sequencing problem (\ref{optimalsequencing}) as follows.
\begin{theorem}\label{theoremadd}
The sequence ${\bm l}^{R\!E}$ obtained by OVBS lies in set ${\cal L}^*$. In other words, ${\bm l}^{R\!E}$ is one of the optimal bargaining sequences for problem (\ref{optimalsequencing}).
\end{theorem}

The basic idea of OVBS is to utilize Theorem \ref{theoremA} to reduce the searching space of ${\bm l}^*$ from set $\cal L$ to a new constructed set ${\cal L}^{R\!E}$. Since $\left| {\cal L} \right|=N!$ and $\left| {\cal L}^{R\!E} \right|={N_2}!$, the complexity of determining ${\bm l}^*$ is significantly reduced.

To summarize, the optimal sequence determined by OVBS has the following features: \textbf{(a)} The MNO bargains with the VOs sequentially in the order of type $1$, type $2$, and type $3$ (Theorem \ref{theoremA}); \textbf{(b)} The MNO will cooperate with all type $1$ VOs (Proposition \ref{proposition1}); \textbf{(c)} The MNO will not cooperate with any type $3$ VO (Proposition \ref{proposition2}); \textbf{(d)} Interchanging any two type $1$ VOs' positions will not change the MNO's payoff (Theorem \ref{theoremA}); \textbf{(e)} Interchanging any two type $3$ VOs' positions will not change the MNO's payoff (Theorem \ref{theoremA}).

We illustrate the optimal sequence's structure in Figure \ref{figure2}.
\begin{figure}[t]
  \centering
  \includegraphics[scale=0.32]{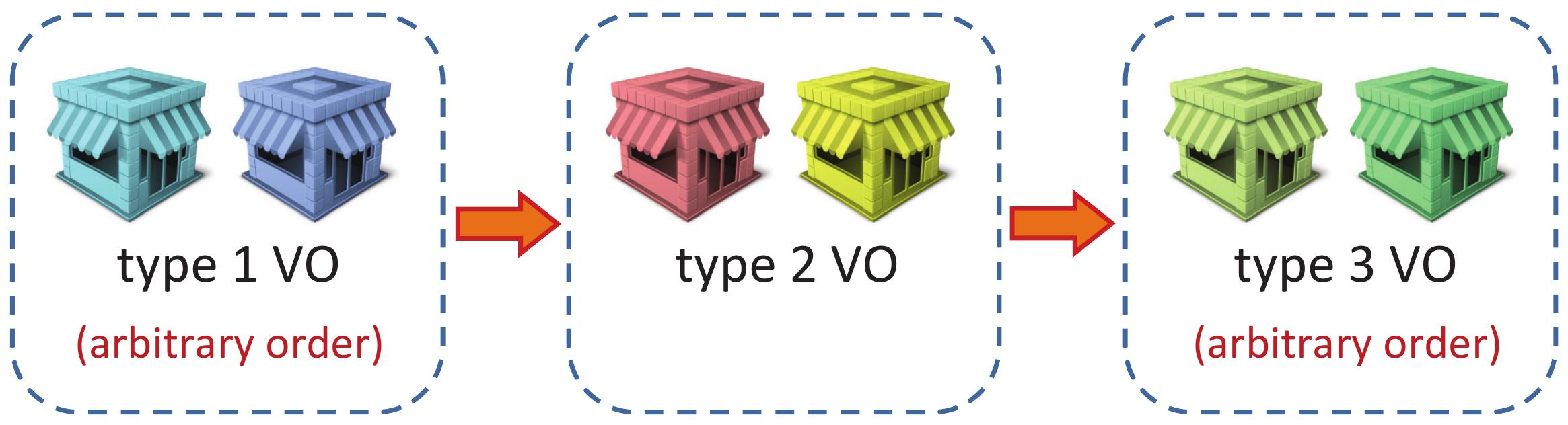}
  \centering
  \caption{Structure of The Optimal Bargaining Sequence under OVBS.}
  \label{figure2}
\end{figure}
\begin{figure}[t]
  \centering
  \includegraphics[scale=0.32]{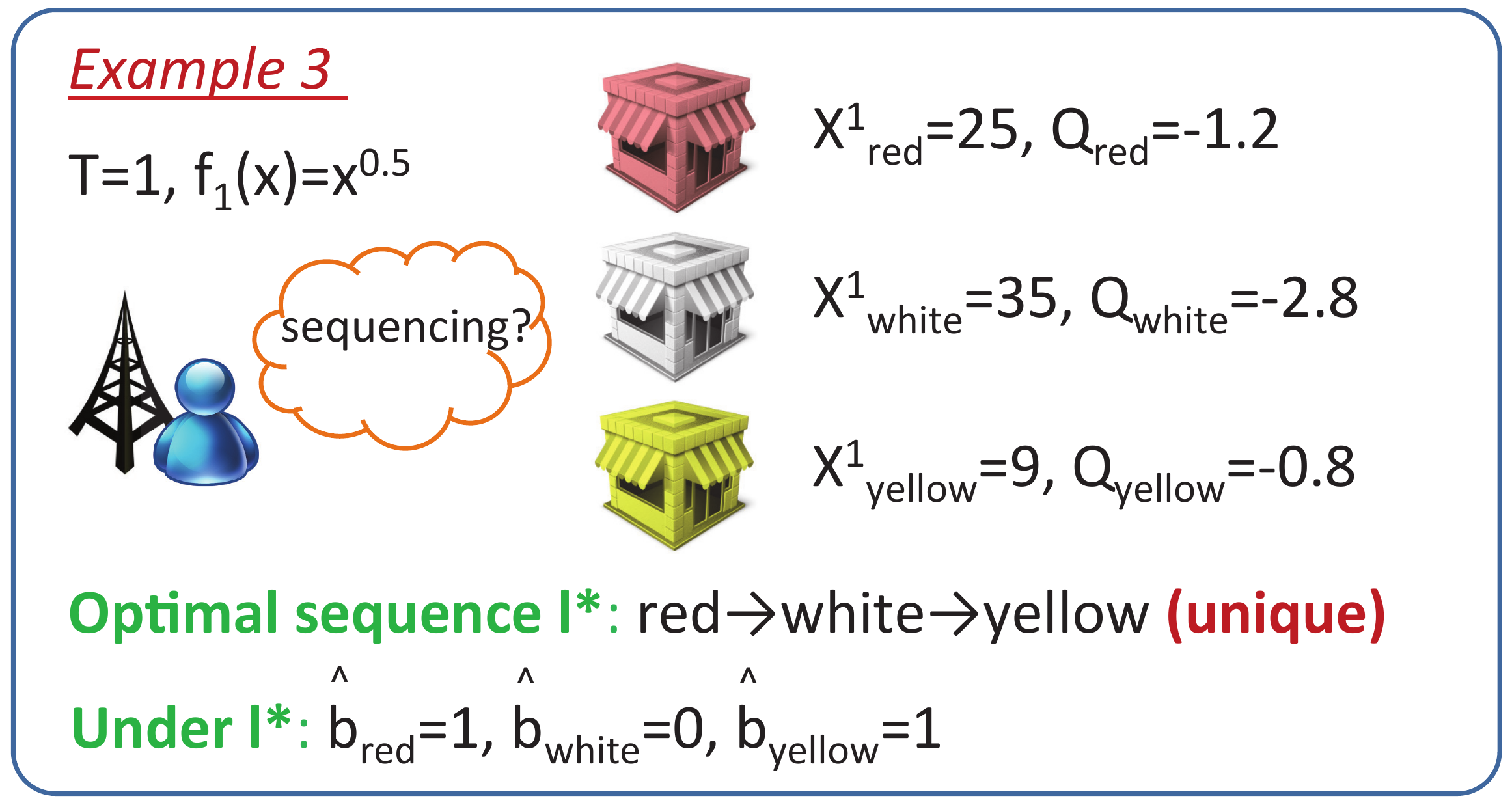}
  \centering
  \caption{Counter-Intuitive Sequencing for Type $2$ VOs.}
  \label{figurecounter}
\end{figure}

It is difficult to further reduce the searching space ${\cal L}^{RE}$, because the optimal sequencing problem involving type $2$ VOs is very complicated in general. To see this, we show a counter-intuitive result in the following proposition.\\

{\vspace{-0.3cm}}
\begin{proposition}\label{proposition:nonconsecutive}
If there are type $2$ VOs, {i.e.}, $N_2>0$, the MNO may cooperate with the VOs nonconsecutively under all the optimal bargaining sequences.
\end{proposition}
We show Example 3 in Figure \ref{figurecounter} to prove Proposition \ref{proposition:nonconsecutive}.{\footnote{In Section \ref{subsec:engineering}, we use Example 3 to show that under a given bargaining sequence, the MNO may cooperate with the VOs nonconsecutively. Here, we use Example 3 to show that this can still happen even if the bargaining sequence is the optimal one.}} In Example 3, all VOs are of type $2$, and the MNO has a unique optimal bargaining sequence, where it bargains with VOs \emph{red}, \emph{white}, and \emph{yellow} sequentially. We find that the MNO only cooperates with VOs \emph{red} and \emph{yellow} under this optimal bargaining sequence. In other words, it is optimal for the MNO in this example to bargain with someone (VO \emph{white}) that it will not cooperate with ahead of someone (VO \emph{yellow}) that it will cooperate with. {{The reason for this counter-intuitive result is that such strategy increases the MNO's disagreement point at the first bargaining step, and hence helps the MNO earn more profit from the cooperation with VO \emph{red}.}} 

Proposition \ref{proposition:nonconsecutive} implies that besides the structural property described in Theorem \ref{theoremA}, it is difficult to explore other structural properties to further reduce the complexity of OVBS.


\subsection{Special Case 1: Only Type $1$ VOs}\label{SpecialA}
We next study a special case where all VOs are of type $1$, \emph{i.e.}, $Q_n\ge0$ for all $n\in\cal N$. In this case, we not only know that any bargaining sequence is optimal (based on Theorem \ref{theoremA}), but also can obtain the closed-form solution of the MNO's payoff as follows.
\begin{theorem}\label{theoremB}
If all VOs are of type $1$, the MNO's payoff is independent of the bargaining sequence $\bm l$ and is given as:
\begin{equation}
{U_0} = \frac{1}{{{2^N}}} \sum\limits_{{{\bm b}_N}\in\cal B} {\Psi \left( {{\bm b}_N} \right)},\label{flipcoin}
\end{equation}
where ${\cal B} \triangleq \left\{ {\left( {{b_1},{b_2}, \ldots ,{b_N}} \right):{b_n} \in \left\{ {0,1} \right\},\forall n \in {\cal N}} \right\}$.
\end{theorem}

Mathematically, the MNO's payoff in (\ref{flipcoin}) can be viewed as the expected social welfare under such a scenario, where the MNO cooperates with each VO with a probability of $0.5$. This observation is consistent with \cite{gao2014bargaining},\cite{moresi2008model}. In fact, \cite{gao2014bargaining},\cite{moresi2008model} studied the one-to-many bargaining without cooperation cost. Hence, the buyer would definitely cooperate with all sellers. That corresponds to the special case that we study in this subsection, \emph{i.e.,} all VOs are of type $1$. In this case, the bargaining sequence does not affect the buyer's payoff, so \cite{gao2014bargaining},\cite{moresi2008model} only studied the one-to-many bargaining with {exogenous} sequence. Our work in Sections \ref{Section4} and \ref{Section5} considers a more general case, where the buyer (\emph{i.e.}, the MNO) may not necessarily cooperate with sellers (\emph{i.e.}, the VOs), and provides a deeper understanding on the one-to-many bargaining with both {exogenous} and {endogenous} sequences.

\vspace{-0.5cm}
\subsection{Special Case 2: Sortable VOs}\label{SpecialB}
In this subsection, we study another special case where all VOs are \emph{sortable}, which is defined in the following.

\begin{definition}\label{definition2}
A set $\cal N$ of VOs is sortable if for any pair of VOs $i,j\in\cal{N}$, we have either (i) $Q_i\ge Q_j$ and $X_i^t\ge X_j^t$ for all $t=1,2,\ldots,T$, or (ii) $Q_i\le Q_j$ and $X_i^t\le X_j^t$ for all $t=1,2,\ldots,T$.
\end{definition}

When a set of VOs are sortable, we can sort them based on $Q_n$ and ${\bm X}_n$. The following theorem shows that this simple sorting generates the optimal bargaining sequence.

\begin{theorem}\label{theorem:sortable}
If all the VOs are sortable, we can construct a sequence $\bm l$ such that for all $n\in\left\{1,2,\ldots,N-1\right\}$, we have ${Q_{l_n}} \ge {Q_{l_{n + 1}}}$ and ${X_{l_n}^t} \ge {X_{l_{n + 1}}^t}$ for all $t=1,2,\ldots,T$. Furthermore:

(i) $\bm l$ is the optimal bargaining sequence of problem (\ref{optimalsequencing});

(ii) Under $\bm l$, the MNO will and only will cooperate with the first $k$ VOs, {i.e.}, VO $l_1,l_2,\ldots,l_{k}$, where $k\in{\left\{ 0 \right\} \cup \cal N}$ is the unique index that satisfies both of the following inequalities:
\begin{align}
& \sum\limits_{t=1}^{T} {f_t\left( {\sum\limits_{n = {l_1}}^{{l_{k-1}}} {{X_{n}^t}}  + {X_{l_k}^t}} \right)} - \sum\limits_{t=1}^{T} {f_t\left( {\sum\limits_{n = l_1}^{{l_{k-1}}} {{X_{n}^t}} } \right)} + {Q_{l_k}} \ge 0,\\
& \sum\limits_{t=1}^{T} {f_t\left( {\sum\limits_{n = l_1}^{l_k} {{X_n^t}}  + {X_{l_{k + 1}}^t}} \right)}- \sum\limits_{t=1}^{T} {f_t\left( {\sum\limits_{n = l_1}^{l_k} {{X_n^t}} } \right)} + {Q_{l_{k + 1}}} < 0.
\end{align}
\end{theorem}

That is to say, when all VOs are sortable, we can explicitly determine the optimal bargaining sequence and identify those VOs that the MNO will cooperate with.

\vspace{-0.2cm}
\section{Influence of Bargaining Sequence on VOs' Payoffs}\label{sec:VOpayoff}
In this section, we study the influence of the bargaining sequence on VOs' payoffs. When VOs are homogenous, we prove that it is always no worse for a particular VO to bargain with the MNO at an earlier position. When VOs are heterogenous, we use an example to show that such ``the earlier the better'' feature is no longer true in general.
\vspace{-0.2cm}
\subsection{Homogenous VOs}
We assume $Q_n=Q$ and $X_n^t=X^t$ for all $n\in\cal N$ and $t=1,2,\ldots,T$, and state the following theorem.
\begin{theorem}\label{homogethe}
If all VOs are homogenous, then for any bargaining sequence $\bm l\in\cal L$, we have ${\hat \pi _{{l_i}}} \ge {\hat \pi _{{l_j}}}$ for any $i<j,i,j\in\cal N$.
\end{theorem}

Theorem \ref{homogethe} shows that the payoff of a VO with an earlier bargaining position is no smaller than the payoff of a VO with a later bargaining position. Since all VOs are homogenous, we conclude that it is always better for a particular VO to bargain with the MNO at an earlier position.

Notice that when VOs are homogenous, they are sortable based on Definition \ref{definition2}. Therefore, we can apply the conclusions in Theorem \ref{theorem:sortable} and obtain the following corollary.
\begin{corollary}\label{corollary:homogenous}
If all VOs are homogenous, then for any bargaining sequence $\bm l\in\cal L$, we have (i) ${\hat \pi _{{l_i}}} \ge {\hat \pi _{{l_j}}} \ge 0$ for any $i<j\le k,i,j\in\cal N$, and (ii) ${\hat \pi _{{l_m}}}=0$ for any $m> k,m\in\cal N$, where $k\in{\left\{ 0 \right\} \cup \cal N}$ is the unique index that satisfies both of the following inequalities:
\begin{align}
& \sum\limits_{t=1}^{T} {f_t\left( k X^t \right)} - \sum\limits_{t=1}^{T} {f_t\left( \left(k-1\right) X^t\right)} + {Q} \ge 0,\\
& \sum\limits_{t=1}^{T} {f_t\left( \left(k+1\right) X^t \right)}- \sum\limits_{t=1}^{T} {f_t\left( kX^t \right)} + {Q} < 0.
\end{align}
\end{corollary}
Corollary \ref{corollary:homogenous} shows that the MNO only cooperates with the first $k$ VOs, and the remaining $N-k$ VOs obtain zero payoffs.
\subsection{Heterogenous VOs}
\begin{figure}[t]
  \centering
  \includegraphics[scale=0.36]{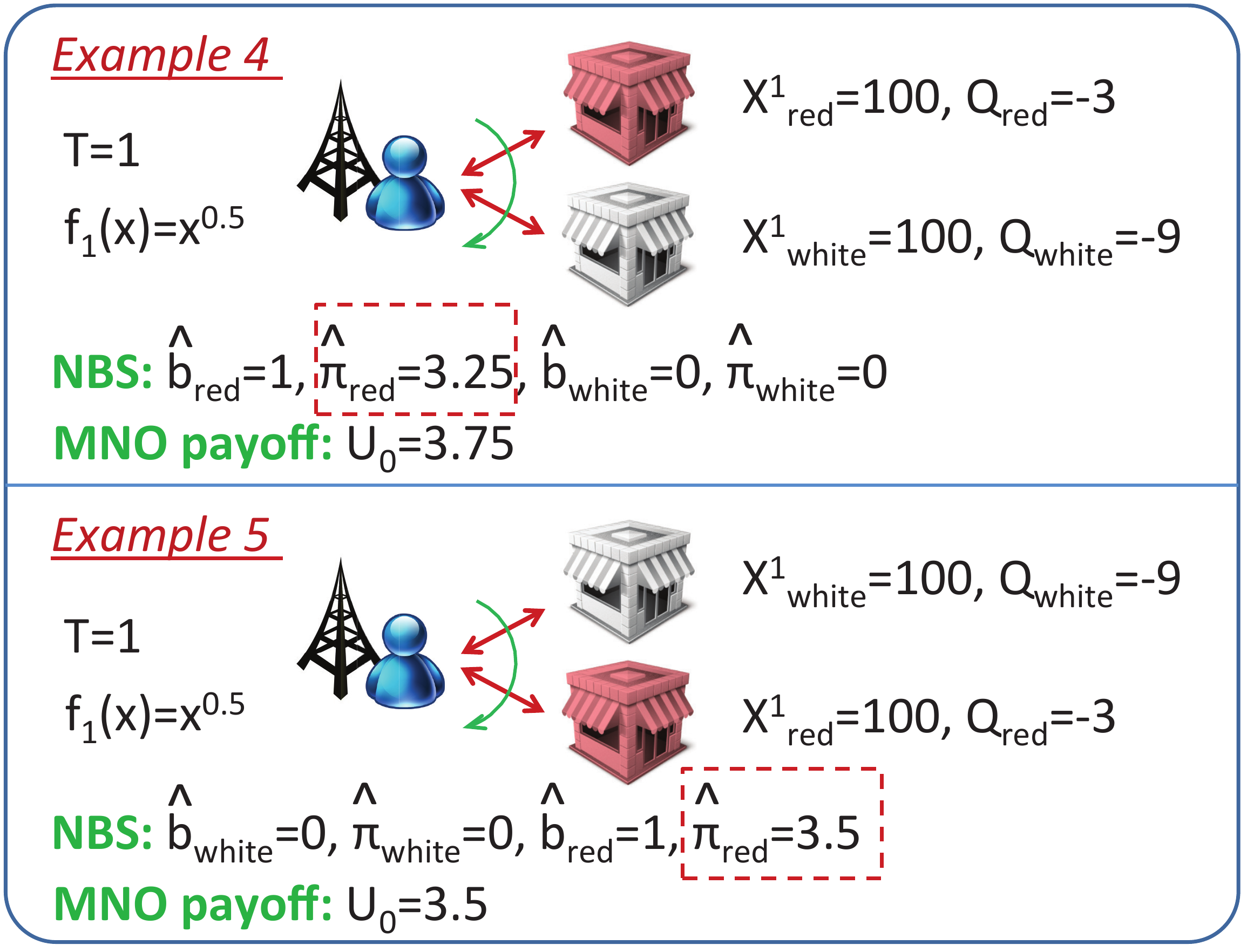}
  \centering
  \caption{Influence of Bargaining Sequence on Heterogenous VOs' Payoffs.}
  \label{figureextra}
\end{figure}

In Figure \ref{figureextra}, we illustrate Examples 4 and 5, where there are two VOs and they are heterogenous in $Q_n$.{\footnote{Similar examples where VOs are heterogenous in ${\bm X}_n$ are given in the appendix.}} 
{{We observe that, the \emph{red} VO's payoff under the later bargaining position is higher than that under the earlier bargaining position.}} 
Intuitively, this can be understood as follows. The MNO only cooperates with the \emph{red} VO in both cases. However, in the first case, the existence of the white VO serves as the ``backup plan'' for the MNO and allows the MNO to obtain a non-zero revenue even if the MNO fails to cooperate with the \emph{red} VO. This increases the MNO's disagreement point in the first bargaining step, and allows the MNO to extract more revenue from its cooperation with the \emph{red} VO. As a result, compared with the second case, the \emph{red} VO receives a lower payoff in the first case.

Examples 4 and 5 imply that when VOs are heterogenous, bargaining with the MNO at an earlier position may decrease the VO's payoff. This conclusion is very interesting, since it contrasts with literature \cite{gao2014bargaining}, which studies the one-to-many bargaining without cooperation cost and concludes that bargaining with the buyer earlier does not decrease the seller's payoff. In our problem, we show that this is not true when considering the cooperation cost.

\begin{figure*}[t]
\begin{minipage}[t]{0.47\linewidth}
\centering
  \includegraphics[scale=0.42]{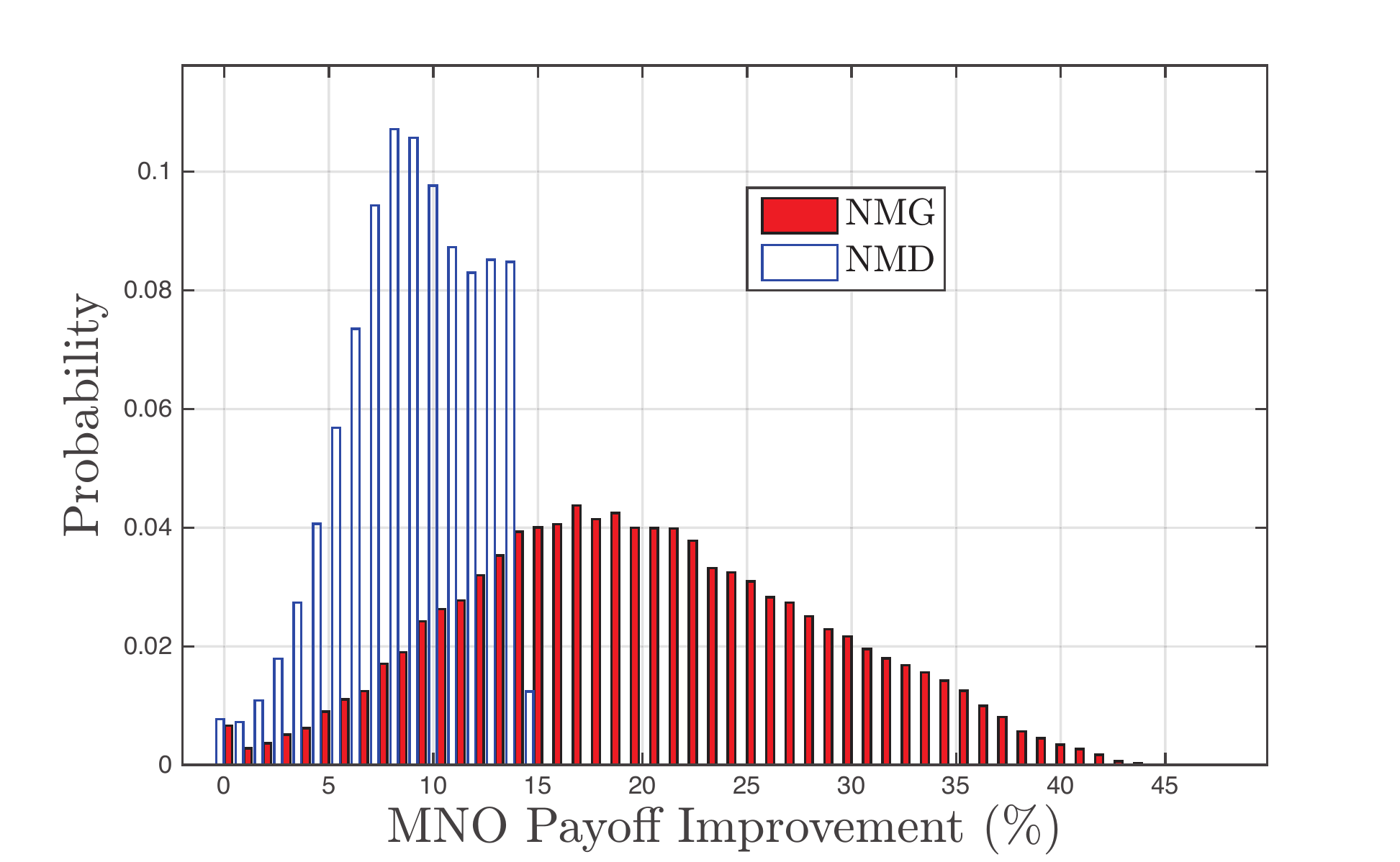}
  \centering
  \caption{Distributions of NMG and NMD (Truncated Normal Distribution).}
  \label{fig:1A}
\end{minipage}
\begin{minipage}[t]{.47\linewidth}
\centering
  \includegraphics[scale=0.42]{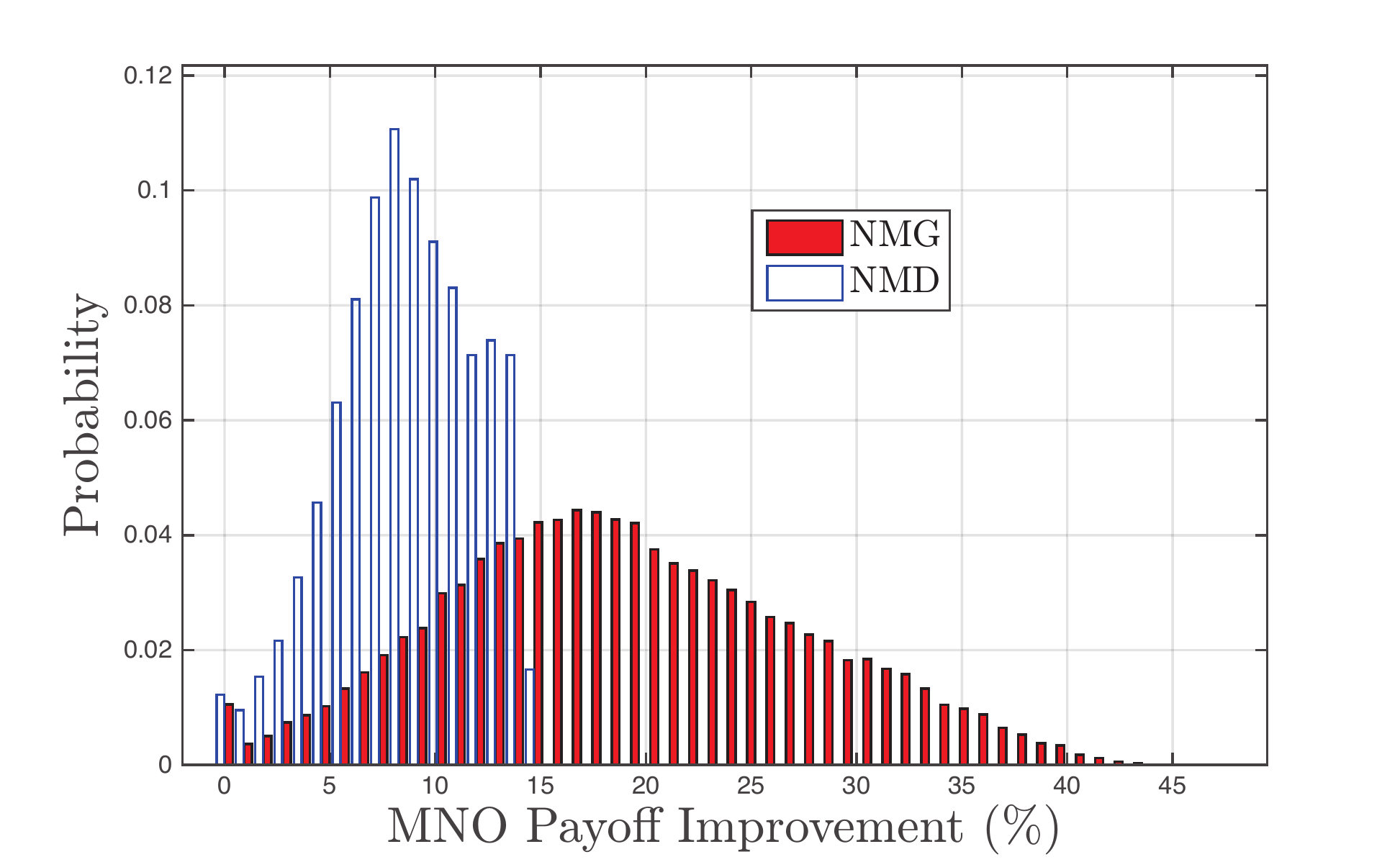}
  \centering
  \caption{Distributions of NMG and NMD (Uniform Distribution).}
  \label{fig:1B}
\end{minipage}
\end{figure*}

\vspace{-0.3cm}
\section{Numerical Results}\label{sec:numerical}
In this section, we evaluate the performance of the optimal sequencing and study the impact of system parameters on the bargaining.
\subsection{Performance of Optimal Sequencing}
First we define the criteria for evaluating the performance gap between different sequencing strategies. For a set $\cal N$ of VOs and the corresponding set $\cal L$ of bargaining sequences, we define the MNO's maximum, minimum, and average payoff as follows:
\begin{align}
\nonumber
{~~~}U_0^{\max } \triangleq \mathop {\max }\limits_{{\bm l} \in {\cal L}} U_0^{\bm l},{~}U_0^{\min } \triangleq \mathop {\min }\limits_{{\bm l} \in {\cal L}} U_0^{\bm l},{~}U_0^{\rm {ave}} \triangleq \frac{1}{{\left| {\cal L} \right|}}\sum\limits_{{\bm l} \in {\cal L}} {U_0^{\bm l}}.
\end{align}
Hence, $U_0^{\max }$, $U_0^{\min }$, and $U_0^{\rm {ave} }$ measure the MNO's payoff under the optimal sequence, worst sequence, and random sequence, respectively. Then we define the normalized maximum gap (NMG) and the normalized maximum deviation (NMD):
\begin{align}
\nonumber
{\rm NMG} \triangleq \frac{{U_0^{\max } - U_0^{\min }}}{{U_0^{\min }}},{\rm{~}} {\rm NMD} \triangleq \frac{{U_0^{\max } - U_0^{\rm ave}}}{{U_0^{\rm ave}}}.
\end{align}
NMG and NMD capture the performance improvement of the optimal sequence over the worst sequence and the random sequence, respectively.

\subsubsection{Distributions of NMG and NMD}\label{simulationtext1}

We choose $\left|\cal N\right|=5$, $T=2$, and $f_t\left(x\right)=x^{0.3}$ for $t=1,2$, and study the probability distributions of NMG and NMD.

First, we assume that $X_n^t$ and $Q_n$ follow the truncated normal distributions. Specifically, we obtain the distribution of $X_n^t,n\in{\cal N},t=1,2$, by truncating the normal distribution ${\cal N}\left(90,900\right)$ to interval $\left[60,120\right]$. Moreover, we obtain the distribution of $Q_n,n\in{\cal N}$, by truncating the normal distribution ${\cal N}\left(-6,9\right)$ to interval $\left[-9,-3\right]$. We run the experiment 30,000 times, and record the probability mass functions of NMG and NMD in Figure \ref{fig:1A}.
We conclude that, (i) compared with the worst sequence, the optimal sequence improves the MNO's payoff by 19.8\% on average and by 45.3\% in the extreme case; (ii) compared with the random sequence, the optimal sequence improves the MNO's payoff by 9.2\% on average and by 14.8\% in the extreme case.

Second, we consider the uniform distribution, and assume that $X_n^t\sim U\left[60,120\right]$ for all $n,t$, and $Q_n\sim U\left[-9,-3\right]$ for all $n$. We illustrate the corresponding probability mass functions of NMG and NMD in Figure \ref{fig:1B}. We can see that the results are similar to those in Figure \ref{fig:1A}, which shows that the simulation results on NMG and NMD are robust to the assumption on probability distributions of the system parameters. To save space, we only simulate the truncated normal distributions for the system parameters in the rest of this section.

We summarize the observations in Figures \ref{fig:1A} and \ref{fig:1B} as follows.
\begin{observation}
For both the truncated normal distribution and the uniform distribution, the optimal bargaining sequence improves the MNO's payoff over the random and worst bargaining sequences by more than 9\% and 19\% on average, respectively.
\end{observation}
\subsubsection{Influences of ${\mathbb E}\left\{X_n^t\right\}$ and ${\mathbb E}\left\{Q_n\right\}$} We investigate the influences of the means of $X_n^t$ and $Q_n$ on the performance of the optimal sequencing. The settings of $\left|{\cal N}\right|$, $T$, and $f_t\left(x\right)$ are the same as those in Section \ref{simulationtext1}.

\begin{figure*}[t]
\begin{minipage}[t]{0.47\linewidth}
\centering
  \includegraphics[scale=0.42]{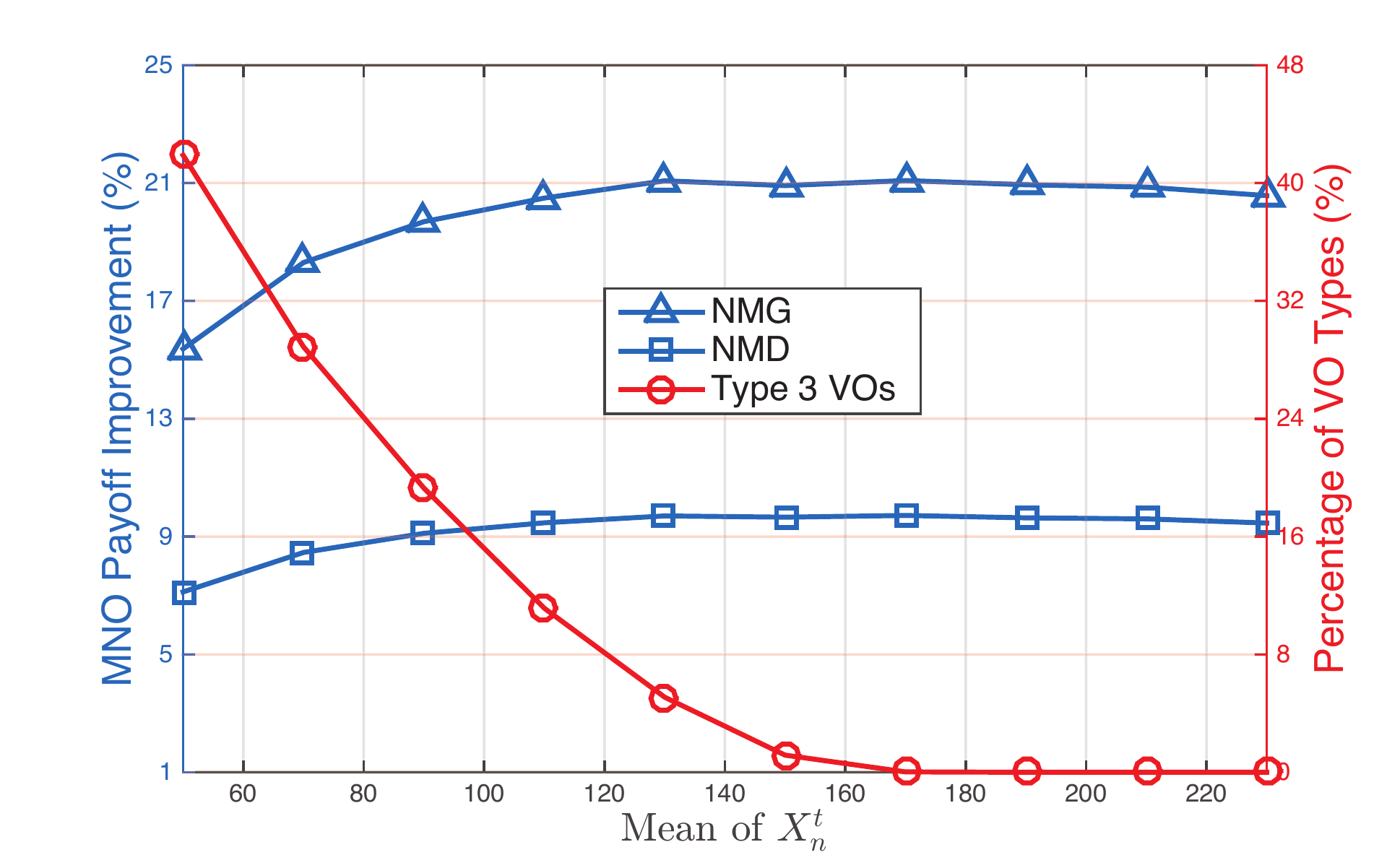}
  \caption{Influence of ${\mathbb{E}}\left\{X_n^t\right\}$ on NMG and NMD.}
  \label{fig:2A}
\end{minipage}
\begin{minipage}[t]{.47\linewidth}
\centering
  \includegraphics[scale=0.42]{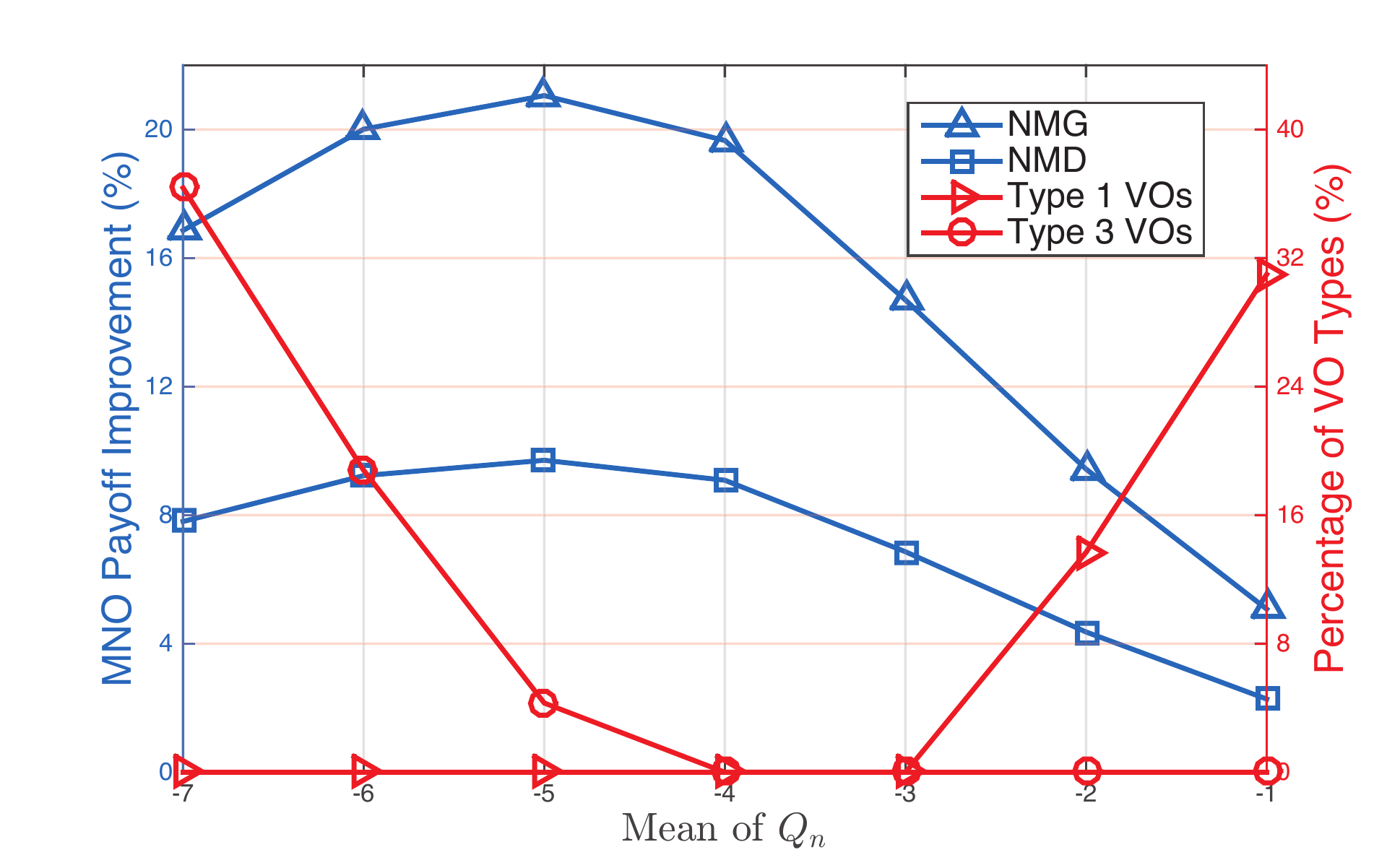}
  \centering
  \caption{Influence of ${\mathbb{E}}\left\{Q_n\right\}$ on NMG and NMD.}
  \label{fig:2B}
\end{minipage}
\end{figure*}

\begin{figure*}[t]
\begin{minipage}[t]{0.47\linewidth}
\centering
  \includegraphics[scale=0.42]{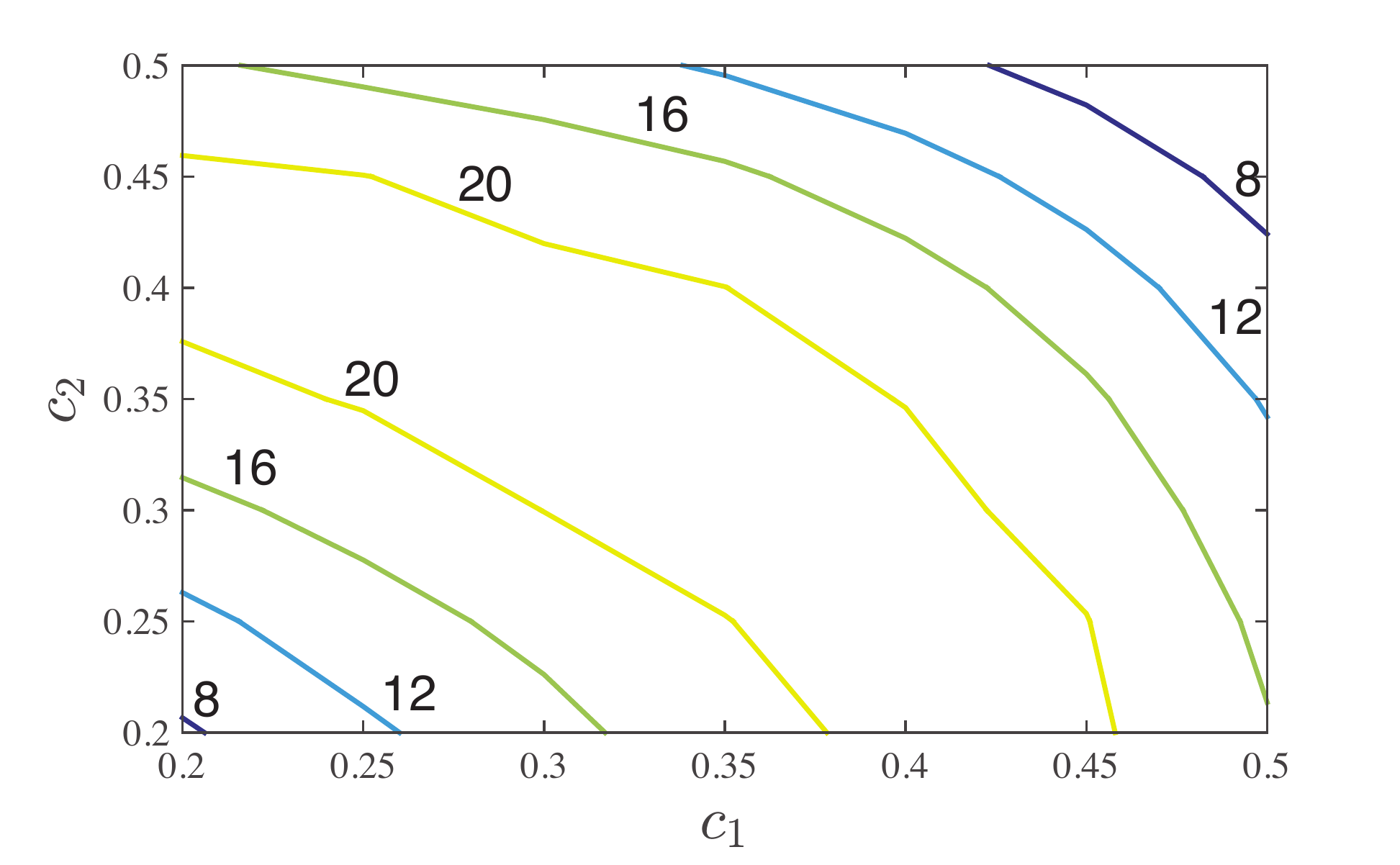}
  \caption{Expected NMG under Different $f_t\left(\cdot\right)$ (\%).}
  \label{fig:4A}
\end{minipage}
\begin{minipage}[t]{.47\linewidth}
\centering
  \includegraphics[scale=0.42]{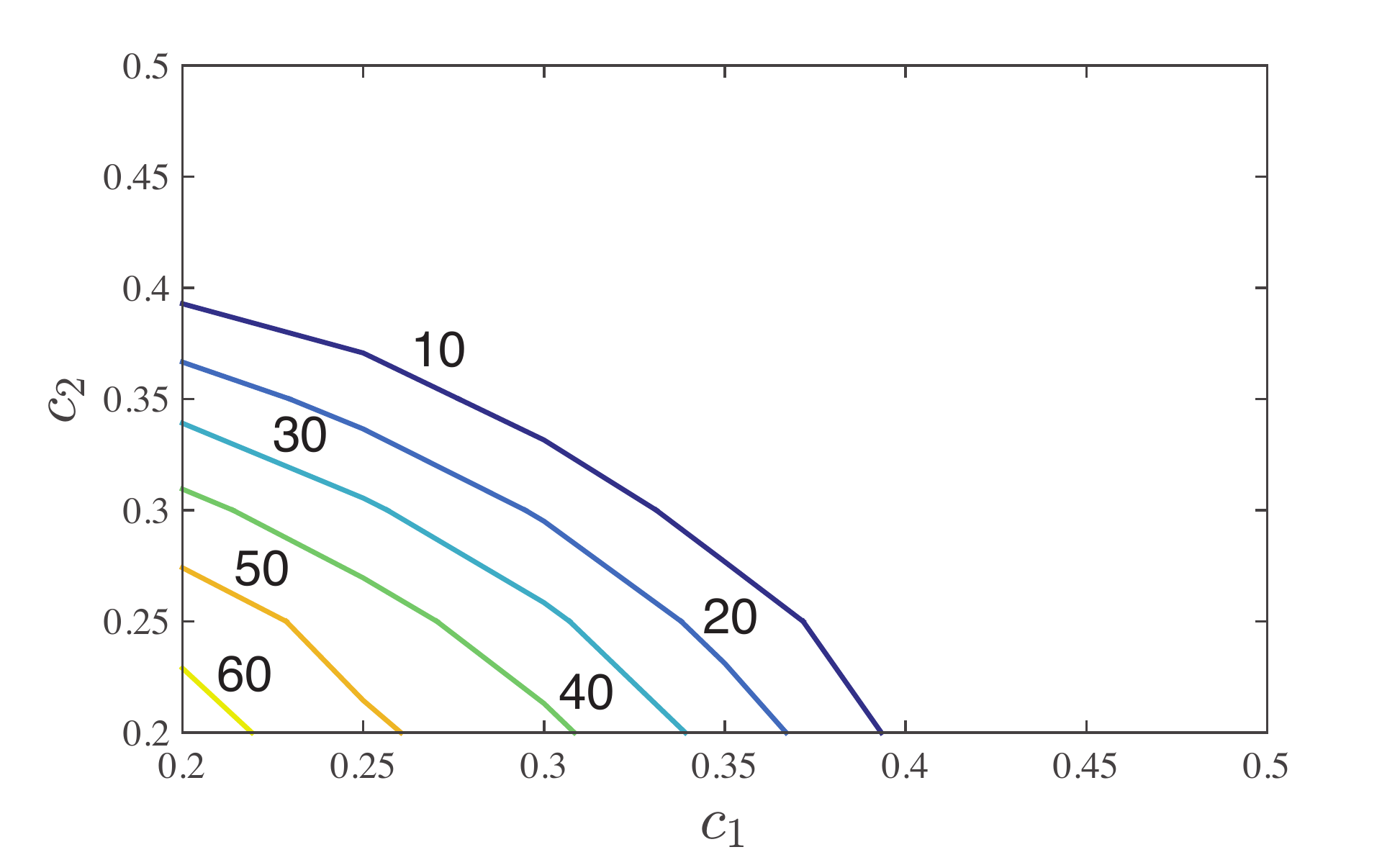}
  \caption{Percentage of Type $3$ VOs under Different $f_t\left(\cdot\right)$ (\%).}
  \label{fig:4B}
\end{minipage}
\end{figure*}

First, we study the influence of ${\mathbb{E}}\left\{X_n^t\right\}$ in Figure \ref{fig:2A}. We assume that $Q_n,n\in{\cal N},$ follows the same distribution as that in Figure \ref{fig:1A}. Moreover, we generate the distribution of $X_n^t,n\in{\cal N},t=1,2$, by truncating the normal distribution ${\cal N}\left({\mathbb{E}}\left\{X_n^t\right\},900\right)$ to interval $\left[{\mathbb{E}}\left\{X_n^t\right\}-30,{\mathbb{E}}\left\{X_n^t\right\}+30\right]$, where ${\mathbb{E}}\left\{X_n^t\right\}$ changes from $50$ to $230$. For each value of ${\mathbb{E}}\left\{X_n^t\right\}$, we run the experiments $10,000$ times, and compute the expected values of NMG and NMD. We plot the expected values of NMG and NMD against ${\mathbb{E}}\left\{X_n^t\right\}$ in Figure \ref{fig:2A}. Since the percentage of type $3$ VOs changes according to $X_n^t$ based on Definition \ref{definition1}, we also plot the expected percentage of type $3$ VOs against ${\mathbb{E}}\left\{X_n^t\right\}$.

In Figure \ref{fig:2A}, we observe that both NMG and NMD slightly increase when ${\mathbb{E}}\left\{X_n^t\right\}$ increases from $50$ to $130$. This is because when ${\mathbb{E}}\left\{X_n^t\right\}$ is small, the percentage of type $3$ VOs is large. Based on Proposition \ref{proposition2}, the MNO never cooperates with these type $3$ VOs. Hence, for a small ${\mathbb{E}}\left\{X_n^t\right\}$, the influence of the bargaining sequence on the MNO's payoff is small, and the benefit of the optimal sequencing is small as well. When ${\mathbb{E}}\left\{X_n^t\right\}$ increases from $130$ to $230$, the percentage of type $3$ VOs decreases to zero, and there are no significant changes in NMG and NMD.

Second, we investigate the influence of ${\mathbb{E}}\left\{Q_n\right\}$ in Figure \ref{fig:2B}. We assume that $X_n^t,n\in{\cal N},t=1,2,$ follows the same distribution as that in Figure \ref{fig:1A}. Furthermore, we obtain the distribution of $Q_n,n\in{\cal N}$, by truncating the normal distribution ${\cal N}\left({\mathbb{E}}\left\{Q_n\right\},9\right)$ to interval $\left[{\mathbb{E}}\left\{Q_n\right\}-3,{\mathbb{E}}\left\{Q_n\right\}+3\right]$, where ${\mathbb{E}}\left\{Q_n\right\}$ changes from $-7$ to $-1$. For each value of ${\mathbb{E}}\left\{Q_n\right\}$, we run the experiments $10,000$ times, and obtain the expected values of NMG and NMD. We plot the expected values of NMG and NMD against ${\mathbb{E}}\left\{Q_n\right\}$ in Figure \ref{fig:2B}. Based on Definition \ref{definition1}, $Q_n$ influences the percentages of both type $1$ and type $3$ VOs. Hence, we also plot the expected percentages of type $1$ and type $3$ VOs against ${\mathbb{E}}\left\{Q_n\right\}$.

In Figure \ref{fig:2B}, we observe that both NMG and NMD first increase and then decrease. The reason is that under a small ${\mathbb{E}}\left\{Q_n\right\}$, there are many type $3$ VOs, which the MNO never cooperates with based on Proposition \ref{proposition2}. Furthermore, under a large ${\mathbb{E}}\left\{Q_n\right\}$, there are many type $1$ VOs, which the MNO always cooperates with based on Proposition \ref{proposition1}. Only under a medium ${\mathbb{E}}\left\{Q_n\right\}$, the bargaining sequence has a large impact on the MNO's payoff, and both NMG and NMD become large. Compared with ${\mathbb{E}}\left\{X_n^t\right\}$ in Figure \ref{fig:2A}, we find that the change in ${\mathbb{E}}\left\{Q_n\right\}$ results in more significant changes of NMG and NMD.

We summarize the observations in Figures \ref{fig:2A} and \ref{fig:2B} as follows.
\begin{observation}
The change in ${\mathbb{E}}\left\{Q_n\right\}$ has a larger impact on the performance of the optimal sequencing than that of ${\mathbb{E}}\left\{X_n^t\right\}$. The benefit of the optimal sequencing is most significant for a medium ${\mathbb{E}}\left\{Q_n\right\}$.
\end{observation}

We further investigate the influence of the concavity of function $f_t\left(\cdot\right)$ on the performance of the optimal sequencing in Figures \ref{fig:4A} and \ref{fig:4B}. We choose the same settings on $\left|{\cal N}\right|$, $T$, and the distributions of $X_n^t$ and $Q_n$ as Figure \ref{fig:1A}. Furthermore, we assume $f_t\left(x\right)=x^{c_t},t=1,2,$ and choose $c_1$ and $c_2$ from $0.2$ to $0.5$, respectively. Note that a smaller $c_t$ means a more concave function $f_t\left(\cdot\right)$. For each pair of $\left(c_1,c_2\right)$, we run the experiment 3,000 times and compute the expected NMG and the percentage of type $3$ VOs, as shown in Figures \ref{fig:4A} and \ref{fig:4B}, respectively.

\begin{figure*}[t]
  \centering
  \begin{minipage}[t]{.48\linewidth}
  \includegraphics[scale=0.42]{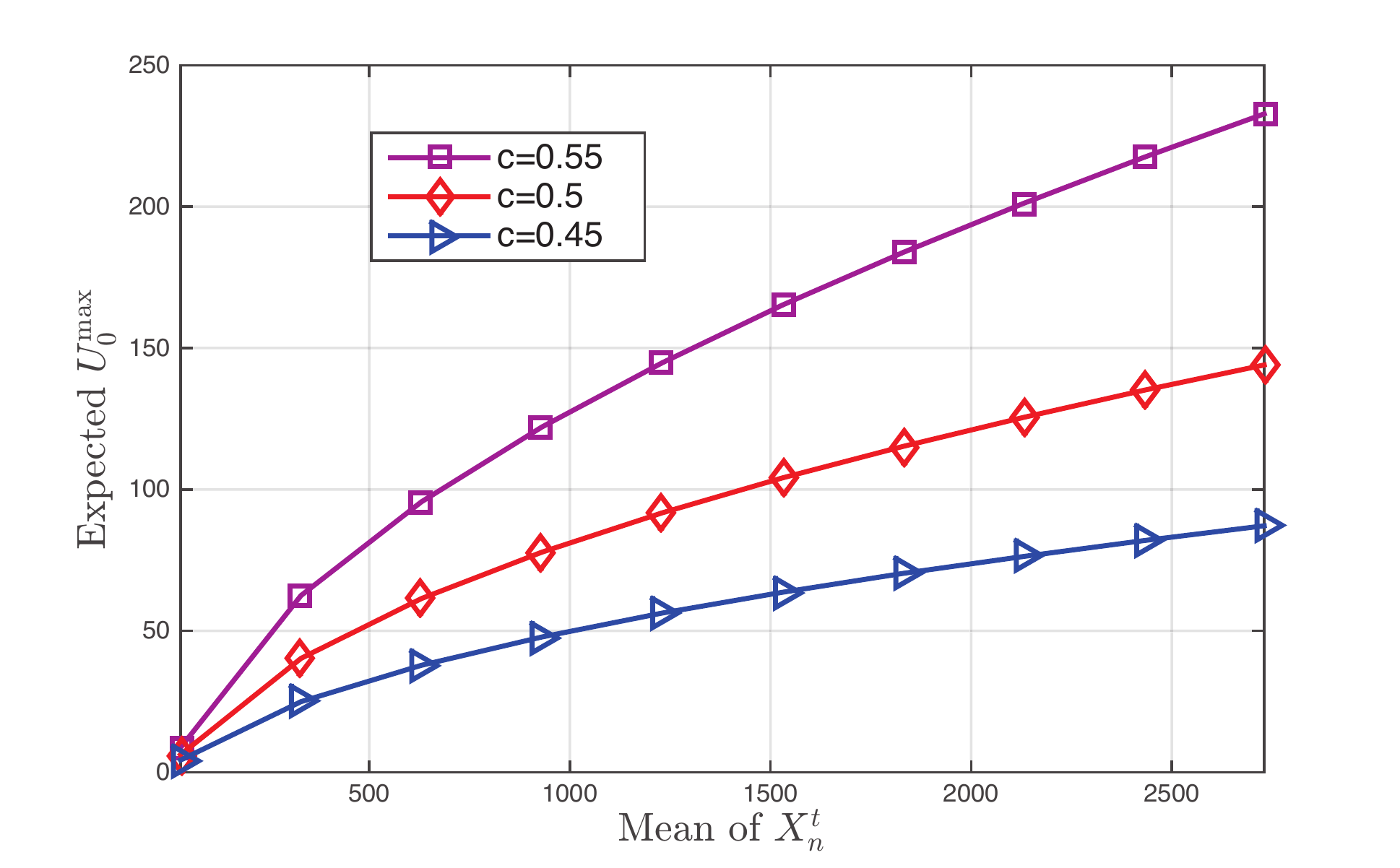}
  \caption{Influence of ${\mathbb{E}}\left\{X_n^t\right\}$ on $U_0^{\max}$.}
  \label{fig:7A}
  \end{minipage}
  \begin{minipage}[t]{.48\linewidth}
  \centering
  \includegraphics[scale=0.42]{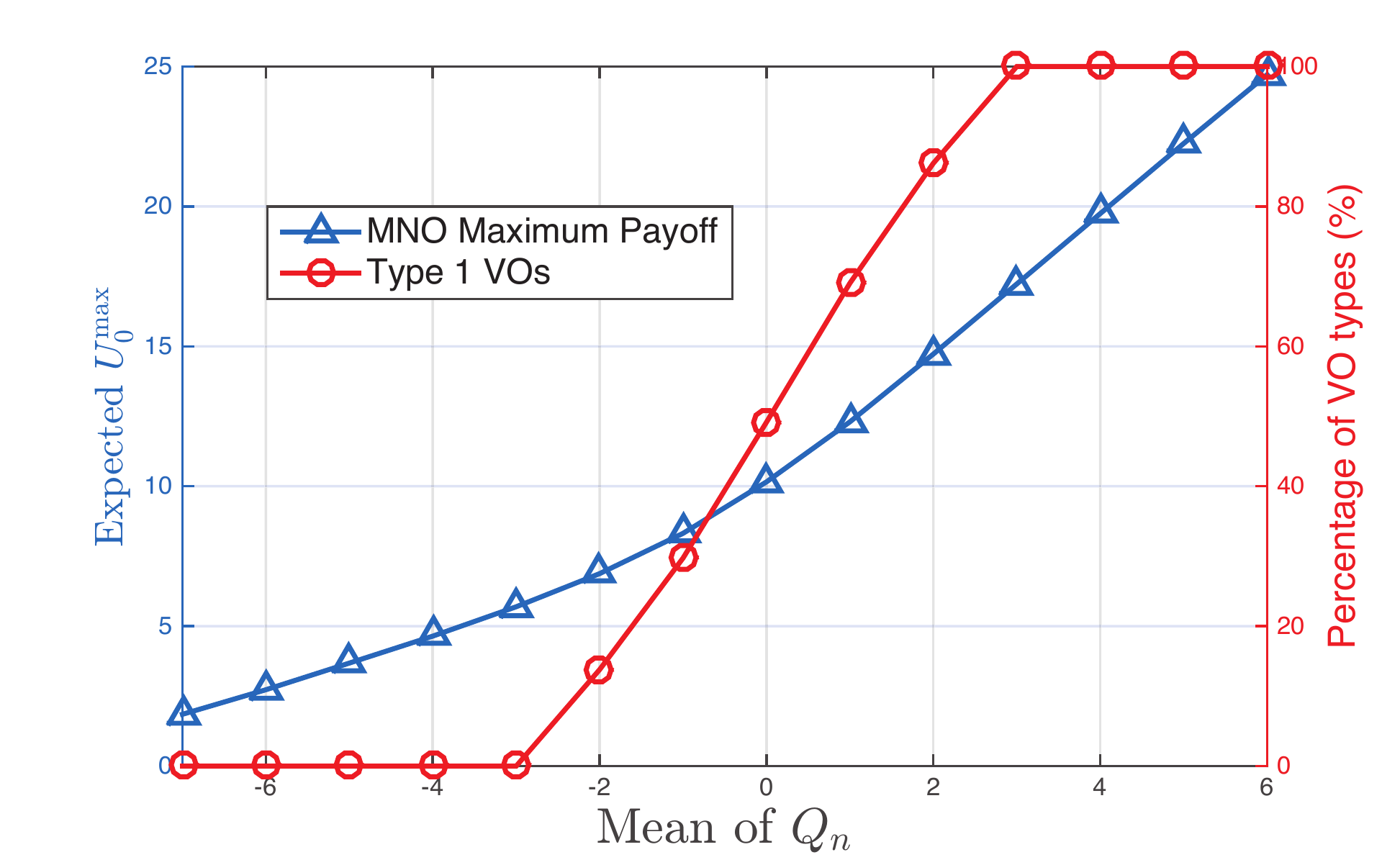}
  \caption{Influence of ${\mathbb{E}}\left\{Q_n\right\}$ on $U_0^{\max}$.}
  \label{fig:7B}
  \end{minipage}
\end{figure*}

In Figure \ref{fig:4A}, we observe that the expected NMG reaches its peak value for medium $c_1$ and $c_2$. This is because when both $c_1$ and $c_2$ are small, the offloading benefit for the MNO is small and most VOs are of type $3$ as shown in Figure \ref{fig:4B}. Recall that the MNO never cooperates with these type $3$ VOs. Hence, for small $c_1$ and $c_2$, the optimal sequencing does not significantly improve the MNO's payoff. {{When both $c_1$ and $c_2$ are large, functions $f_1\left(\cdot\right)$ and $f_2\left(\cdot\right)$ become less concave. In this case, given the same number of deployed Wi-Fi networks, the MNO is more willing to deploy new Wi-Fi networks. That is to say, the externalities among different steps of bargaining become weaker. As a result, different bargaining steps are less tightly coupled, and the bargaining sequence has a smaller impact on the MNO's payoff. Therefore, the advantage of the optimal sequencing reduces and the expected NMG decreases.}}

We summarize the following observation for Figures \ref{fig:4A} and \ref{fig:4B}.
\begin{observation}
The benefit of the optimal sequencing is most significant when the offloading benefit function $f_t\left(\cdot\right)$ has a medium concavity.
\end{observation}

\subsection{MNO's Payoff}
In Figures \ref{fig:7A} and \ref{fig:7B}, we study the impact of different parameters on the MNO's maximum payoff, \emph{i.e.}, $U_0^{\max}$.
\subsubsection{Influence of ${\mathbb E}\left\{X_n^t\right\}$}

We apply the same simulation settings on $\left|{\cal N}\right|$, $T$, and the distributions of $X_n^t$ and $Q_n$ as Figure \ref{fig:2A}. {{We assume that $f_t\left(x\right)=x^c$ for $t=1,2,$ and choose $c=0.45,0.5$, and $0.55$. For each $c$, we change ${\mathbb{E}}\left\{ X_n^t\right\}$ from $30$ to $2730$, and illustrate the corresponding expected $U_0^{\max}$ in Figure \ref{fig:7A}.}}
We observe that $U_0^{\max}$ concavely increases with ${\mathbb{E}}\left\{ X_n^t\right\}$. Based on (\ref{MNOpayoff}), such a concavity is due to the concave offloading benefit function $f_t\left(\cdot\right)$. Since a larger $c$ corresponds to a less concave function $f_t\left(\cdot\right)$, we observe in Figure \ref{fig:7A} that an increase of $c$ leads to a decrease of the concavity of the $U_0^{\max}$ curve.

\subsubsection{Influence of ${\mathbb{E}}\left\{Q_n\right\}$}

We use the same simulation settings on $\left|{\cal N}\right|$, $T$, $f_t\left(\cdot\right)$, and the distributions of $X_n^t$ and $Q_n$ as Figure \ref{fig:2B}. We change ${\mathbb{E}}\left\{Q_n\right\}$ from $-7$ to $6$ and illustrate the corresponding $U_0^{\max}$ in Figure \ref{fig:7B}.
We find that $U_0^{\max}$ increases with ${\mathbb{E}}\left\{ Q_n \right\}$, because a large ${\mathbb{E}}\left\{ Q_n \right\}$ implies a large benefit (or a small cost) of deploying Wi-Fi network, and the MNO can earn more profit from the cooperative Wi-Fi deployment. Furthermore, we find that $U_0^{\max}$ eventually linearly increases when ${\mathbb{E}}\left\{Q_n \right\}\ge3$. To explain this, we also show the percentage of type $1$ VOs in Figure \ref{fig:7B}. As ${\mathbb{E}}\left\{Q_n\right\}$ increases, the percentage of type $1$ VOs approaches 100\%. Based on (\ref{socialwelfare}) and (\ref{flipcoin}), when all VOs are of type $1$, we have
\begin{align}
{U_0^{\max }} = \frac{1}{{{2^N}}}\sum\limits_{{\bm b}_N \in {\cal B}}\sum\limits_{t=1}^T {f_t\left( {\sum\limits_{n = 1}^N {{b_n}{X_n^t}} } \right)}  + \frac{1}{2}\sum\limits_{n = 1}^N {{Q_n}},
\end{align}
where ${\cal B}$ is defined in Theorem \ref{theoremB}. Hence, $U_0^{\max}$ linearly increases with ${\mathbb{E}}\left\{ Q_n \right\}$, and the slope of the curve is $N/{2}$.

We conclude the following observations for Figures \ref{fig:7A} and \ref{fig:7B}.
\begin{observation}
The MNO's maximum payoff concavely increases with ${\mathbb{E}}\left\{X_n^t\right\}$, and the concavity of the curve increases with the concavity of function $f_t\left(\cdot\right)$. Moreover, the MNO's maximum payoff increases with ${\mathbb{E}}\left\{Q_n\right\}$. In particular, it linearly increases with ${\mathbb{E}}\left\{Q_n\right\}$ when all VOs are of type $1$.
\end{observation}

\section{Conclusion}\label{sec:conclusion}
In this paper, we investigated the economic interactions among the MNO and VOs in the cooperative Wi-Fi deployment. We analyzed the problem under the one-to-many bargaining framework, with both exogenous and endogenous sequences. For the exogenous case, we applied backward induction to compute the bargaining results in terms of the cooperation decisions and payments for a given bargaining sequence. For the endogenous case, we proposed the OVBS algorithm that searches for the optimal bargaining sequence by leveraging the structural property. Furthermore, we studied the influence of the bargaining sequence on VOs, and found that when VOs are homogenous, the earlier bargaining positions are always no worse for the VOs.
Numerical results showed that the optimal bargaining sequence significantly improves the MNO's payoff as compared with the random and worst bargaining sequences. 
We illustrated that the optimal sequencing is most beneficial when the offloading benefit functions have medium concavities.

In our future work, we will further consider the incomplete information scenario, where the MNO and the VO have limited information of the remaining VOs for each step of the bargaining. Moreover, we are interested in studying the MNO competition, where multiple MNOs compete for the VOs' cooperation.

\bibliographystyle{IEEEtran}
\bibliography{IEEEabrv,bare_conf}

\begin{thebibliography}{10}
\providecommand{\url}[1]{#1}
\csname url@samestyle\endcsname
\providecommand{\newblock}{\relax}
\providecommand{\bibinfo}[2]{#2}
\providecommand{\BIBentrySTDinterwordspacing}{\spaceskip=0pt\relax}
\providecommand{\BIBentryALTinterwordstretchfactor}{4}
\providecommand{\BIBentryALTinterwordspacing}{\spaceskip=\fontdimen2\font plus
\BIBentryALTinterwordstretchfactor\fontdimen3\font minus
  \fontdimen4\font\relax}
\providecommand{\BIBforeignlanguage}[2]{{%
\expandafter\ifx\csname l@#1\endcsname\relax
\typeout{** WARNING: IEEEtran.bst: No hyphenation pattern has been}%
\typeout{** loaded for the language `#1'. Using the pattern for}%
\typeout{** the default language instead.}%
\else
\language=\csname l@#1\endcsname
\fi
#2}}
\providecommand{\BIBdecl}{\relax}
\BIBdecl

\bibitem{yu2015cooperative}
H.~Yu, M.~H. Cheung, and J.~Huang, ``Cooperative {Wi-Fi} deployment: A
  one-to-many bargaining framework,'' in \emph{Proc. of IEEE WiOpt}, Mumbai,
  India, May 2015, pp. 347--354.

\bibitem{4Ginteg}
{4G Americas}, ``Integration of cellular and {Wi-Fi} networks,'' \emph{White
  Paper}, September 2013.

\bibitem{lee2013mobile}
K.~Lee, J.~Lee, Y.~Yi, I.~Rhee, and S.~Chong, ``Mobile data offloading: How
  much can {WiFi} deliver?'' \emph{IEEE/ACM Transactions on Networking},
  vol.~21, no.~2, pp. 536--550, April 2013.

\bibitem{WBA}
{Wireless Broadband Alliance}, ``Global trends in public {Wi-Fi},'' Tech. Rep.,
  November 2013.

\bibitem{Cisco2013}
Cisco, ``{Wi-Fi}: New business models create real value for service
  providers,'' Tech. Rep., June 2013.

\bibitem{ATTWiFismall}
{AT\&T}, ``{AT\&T} {Wi-Fi} small site,'' \emph{Product Brochure}, March 2015.

\bibitem{napel2002bilateral}
S.~Napel, \emph{Bilateral bargaining: {T}heory and applications}.\hskip 1em
  plus 0.5em minus 0.4em\relax Springer, 2002, vol. 518.

\bibitem{gao2014bargaining}
L.~Gao, G.~Iosifidis, J.~Huang, L.~Tassiulas, and D.~Li, ``Bargaining-based
  mobile data offloading,'' \emph{IEEE Journal on Selected Areas in
  Communications}, vol.~32, no.~6, pp. 1114--1125, June 2014.

\bibitem{nash1950bargaining}
J.~F. Nash, ``The bargaining problem,'' \emph{Econometrica: Journal of the
  Econometric Society}, pp. 155--162, 1950.

\bibitem{zheng2012sparse}
Z.~Zheng, P.~Sinha, and S.~Kumar, ``Sparse {WiFi} deployment for vehicular
  internet access with bounded interconnection gap,'' \emph{IEEE/ACM
  Transactions on Networking}, vol.~20, no.~3, pp. 956--969, June 2012.

\bibitem{wang2013exploiting}
T.~Wang, W.~Jia, G.~Xing, and M.~Li, ``Exploiting statistical mobility models
  for efficient {Wi-Fi} deployment,'' \emph{IEEE Transactions on Vehicular
  Technology}, vol.~62, no.~1, pp. 360--373, January 2013.

\bibitem{bulut2013wifi}
E.~Bulut and B.~K. Szymanski, ``{WiFi} access point deployment for efficient
  mobile data offloading,'' \emph{ACM SIGMOBILE Mobile Computing and
  Communications Review}, vol.~17, no.~1, pp. 71--78, January 2013.

\bibitem{liao2011two}
L.~Liao, W.~Chen, C.~Zhang, L.~Zhang, D.~Xuan, and W.~Jia, ``Two birds with one
  stone: Wireless access point deployment for both coverage and localization,''
  \emph{IEEE Transactions on Vehicular Technology}, vol.~60, no.~5, pp.
  2239--2252, June 2011.

\bibitem{george2016mobihoc}
K.~Poularakis, G.~Iosifidis, and L.~Tassiulas, ``Deploying carrier-grade
  {WiFi}: Offload traffic, not money,'' in \emph{Proc. of ACM MobiHoc},
  Paderborn, Germany, July 2016.

\bibitem{iosifidis2015double}
G.~Iosifidis, L.~Gao, J.~Huang, and L.~Tassiulas, ``A double-auction mechanism
  for mobile data-offloading markets,'' \emph{IEEE/ACM Transactions on
  Networking}, vol.~23, no.~5, pp. 1634--1647, October 2015.

\bibitem{paris2015efficient}
S.~Paris, F.~Martignon, I.~Filippini, and L.~Chen, ``An efficient auction-based
  mechanism for mobile data offloading,'' \emph{IEEE Transactions on Mobile
  Computing}, vol.~14, no.~8, pp. 1573--1586, August 2015.

\bibitem{dong2014ideal}
W.~Dong, S.~Rallapalli, R.~Jana, L.~Qiu, K.~Ramakrishnan, L.~Razoumov,
  Y.~Zhang, and T.~W. Cho, ``{iDEAL}: Incentivized dynamic cellular offloading
  via auctions,'' \emph{IEEE/ACM Transactions on Networking}, vol.~22, no.~4,
  pp. 1271--1284, August 2014.

\bibitem{lu2014easybid}
Z.~Lu, P.~Sinha, and R.~Srikant, ``Easybid: Enabling cellular offloading via
  small players,'' in \emph{Proc. of IEEE INFOCOM}, Toronto, Canada, April
  2014, pp. 691--699.

\bibitem{yu2016wifiad}
H.~Yu, M.~H. Cheung, L.~Gao, and J.~Huang, ``Economics of public {Wi-Fi}
  monetization and advertising,'' in \emph{Proc. of IEEE INFOCOM}, San
  Francisco, CA, April 2016.

\bibitem{moresi2008model}
S.~Moresi, S.~C. Salop, and Y.~Sarafidis, ``A model of ordered bargaining with
  applications,'' {W}orking paper, 2008.

\bibitem{li2010one}
D.~Li, ``One-to-many bargaining with endogenous protocol,'' {W}orking paper,
  2010.

\bibitem{cai2000delay}
H.~Cai, ``Delay in multilateral bargaining under complete information,''
  \emph{Journal of Economic Theory}, vol.~93, no.~2, pp. 260--276, 2000.

\bibitem{cai2003inefficient}
------, ``Inefficient {Markov} perfect equilibria in multilateral bargaining,''
  \emph{Economic Theory}, vol.~22, no.~3, pp. 583--606, 2003.

\bibitem{wang2015understanding}
H.~Wang, F.~Xu, Y.~Li, P.~Zhang, and D.~Jin, ``Understanding mobile traffic
  patterns of large scale cellular towers in urban environment,'' in
  \emph{Proc. of ACM IMC}, Tokyo, Japan, October 2015, pp. 225--238.

\bibitem{kim2006extracting}
M.~Kim, D.~Kotz, and S.~Kim, ``Extracting a mobility model from real user
  traces,'' in \emph{Proc. of INFOCOM}, Barcelona, Spain, April 2006, pp.
  1--13.

\bibitem{SenzaEco}
{Senza Fili}, ``The economics of small cells and {Wi-Fi} offload,'' \emph{White
  Paper}, 2012.

\end{thebibliography}
\begin{IEEEbiography}
[{\includegraphics[width=1in,height=1.25in,clip,keepaspectratio]{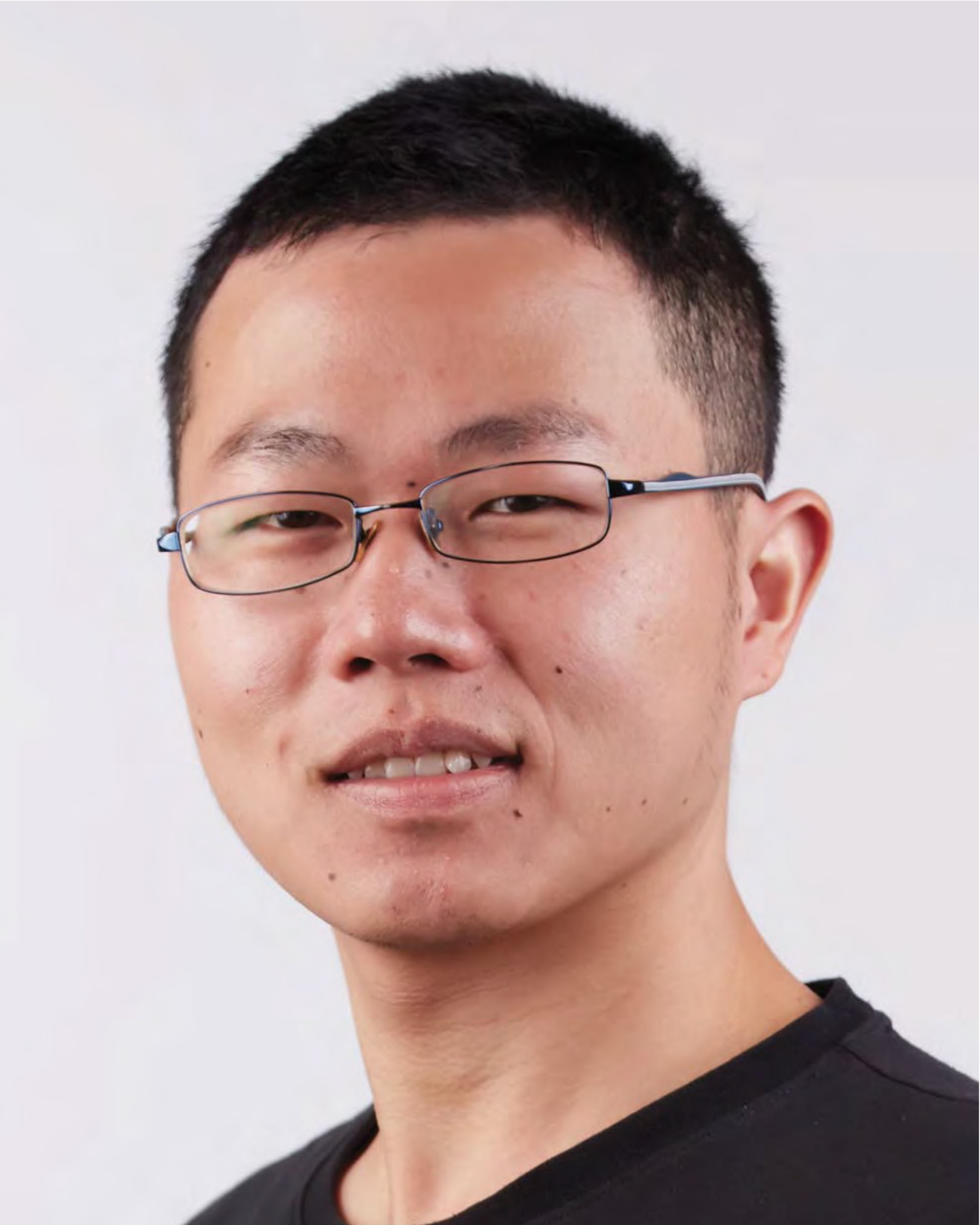}}]
{Haoran Yu} (S'14) is a Ph.D. student in the Department of Information Engineering at the Chinese University of Hong Kong (CUHK). He was a visiting student in the Yale Institute for Network Science (YINS) and the Department of Electrical Engineering at Yale University during 2015-2016. His research interests lie in the field of wireless communications and network economics, with current emphasis on mobile data offloading, cellular/Wi-Fi integration, LTE in unlicensed spectrum, and economics of public Wi-Fi networks. 
He was awarded the Global Scholarship Programme for Research Excellence by CUHK. 
His paper in {\it IEEE INFOCOM 2016} was selected as a Best Paper Award finalist and one of top 5 papers from 1600+ submissions. 
\end{IEEEbiography}


\begin{IEEEbiography}
[{\includegraphics[width=1in,height=1.25in,clip,keepaspectratio]{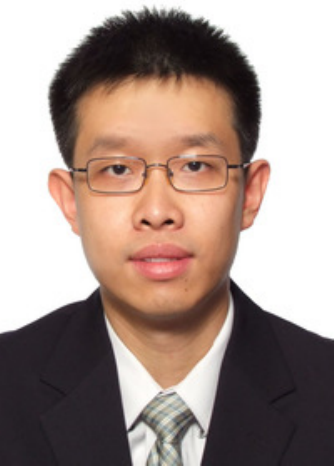}}]
{Man Hon Cheung} received the B.Eng. and M.Phil. degrees in Information Engineering from the Chinese University of Hong Kong (CUHK) in 2005 and 2007, respectively, and the Ph.D. degree in Electrical and Computer Engineering from the University of British Columbia (UBC) in 2012.
 Currently, he is a postdoctoral fellow in the Department of Information Engineering in CUHK.
 He received the IEEE Student Travel Grant for attending {\it IEEE ICC 2009}. He was awarded the Graduate Student International Research Mobility Award by UBC, and the Global Scholarship Programme for Research Excellence by CUHK.
 He serves as a Technical Program Committee member in {\it IEEE ICC}, {\it Globecom}, and {\it WCNC}.
 His research interests include the design and analysis of wireless network protocols using optimization theory, game theory, and dynamic programming, with current focus on mobile data offloading, mobile crowd sensing, and network economics.
\end{IEEEbiography}

\begin{IEEEbiography}[{\includegraphics[width=1in,height=1.25in,clip,keepaspectratio]{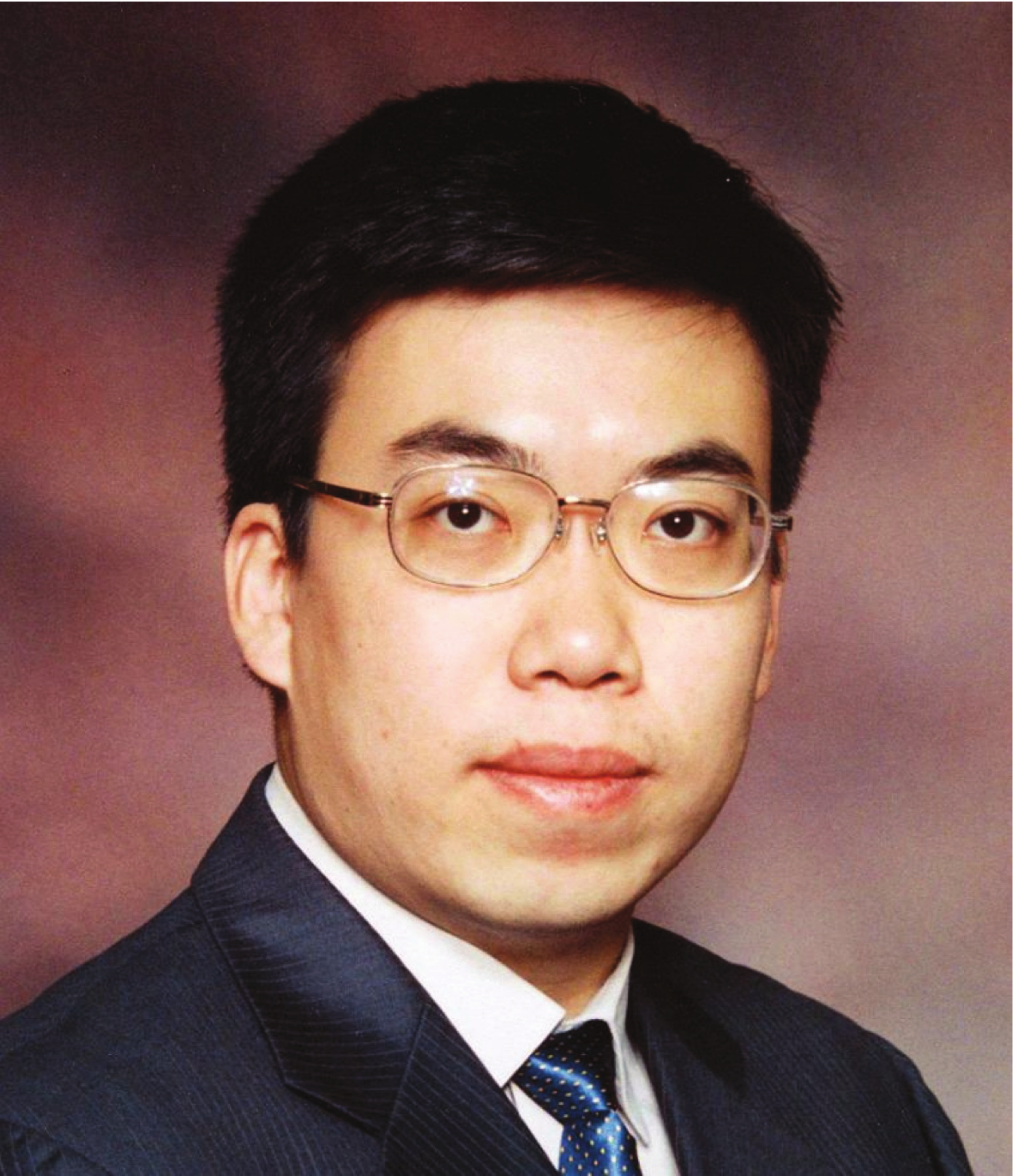}}]
{Jianwei Huang} (S'01-M'06-SM'11-F'16) is an Associate Professor and Director of the Network Communications and Economics Lab (ncel.ie.cuhk.edu.hk), in the Department of Information Engineering at the Chinese University of Hong Kong. He received the Ph.D. degree from Northwestern University in 2005, and worked as a Postdoc Research Associate in Princeton during 2005-2007. He is the co-recipient of 8 international Best Paper Awards, including IEEE Marconi Prize Paper Award in Wireless Communications in 2011. He has co-authored five books: ``Wireless Network Pricing," ``Monotonic Optimization in Communication and Networking Systems,"  ``Cognitive Mobile Virtual Network Operator Games,'' ``Social Cognitive Radio Networks," and ``Economics of Database-Assisted Spectrum Sharing''. He has served as an Associate Editor of IEEE Transactions on Cognitive Communications and Networking, IEEE Transactions on Wireless Communications, and IEEE Journal on Selected Areas in Communications - Cognitive Radio Series. He is the Vice Chair of IEEE ComSoc Cognitive Network Technical Committee and the Past Chair of IEEE ComSoc Multimedia Communications Technical Committee. He is a Fellow of IEEE (Class of 2016) and a Distinguished Lecturer of IEEE Communications Society.
\end{IEEEbiography}

\newpage

\appendices

\section{Proof of Proposition \ref{proposition:one2one}}
\begin{proof}
We study $\Psi \left( {{1}} \right)\ge0$ and $\Psi \left( {{1}} \right)<0$ separately.

(Case 1) $\Psi \left( {{1}} \right)\ge0$: we first consider the optimal $\pi_1^*$ given $b_1=0$. Based on (\ref{socialwelfare}), we have $\Psi\left(0\right)=0$. Hence, with $b_1=0$, the constraints in problem (\ref{formulation2}) become $0-\pi_1\ge0$ and $\pi_1\ge0$. Therefore, the only feasible $\pi_1$ is zero, and the corresponding value of the objective function of problem (\ref{formulation2}) is zero.

Next we consider the optimal $\pi_1^*$ given $b_1=1$. We obtain $\pi_1^*$ by optimizing $\left( {\Psi \left( {{1}} \right) -{\pi _1}} \right) \cdot {\pi _1}$ over $\pi_1\in\left[0,\Psi\left(1\right)\right]$. It is easy to find that $\pi_1^*=\frac{1}{2}\Psi \left( 1 \right)$ and the corresponding value of the objective function of problem (\ref{formulation2}) is $\frac{1}{4}\Psi^2\left(1\right)$.

Finally, we compare the optimal values of the objective function of problem (\ref{formulation2}) under $b_1=0$ and $b_1=1$. Since $\Psi \left( {{1}} \right)\ge0$, we have $\frac{1}{4}\Psi^2 \left( {{1}} \right)\ge0$. Therefore, the optimal solution to problem (\ref{formulation2}) is $\left( {b_1^*,\pi _1^*} \right)=\left( {1,\frac{1}{2}\Psi \left( 1 \right)} \right)$.

(Case 2) $\Psi \left( {{1}} \right)<0$: given $b_1=0$, the analysis of the optimal $\pi_1^*$ is the same as that in Case 1. The only feasible $\pi_1$ is zero, and the corresponding objective function's value is zero.
Given $b_1=1$, the constraints in problem (\ref{formulation2}) become $\Psi\left(1\right)-\pi_1\ge0$ and $\pi_1\ge0$. Since $\Psi \left( {{1}} \right)<0$, there is no feasible solution for $\pi_1$.
Considering the feasibilities of $\pi_1$ under $b_1=0$ and $b_1=1$, we conclude that the optimal solution to problem (\ref{formulation2}) is $\left( {b_1^*,\pi _1^*} \right)=\left( {0,0} \right)$.

Combining Case 1 and Case 2 completes the proof.
\end{proof}

\section{Preliminary Lemmas I}\label{appendix:section:pro1}
In this section, we prove a series of lemmas, which are useful to show the propositions and theorems in the paper.

We first analyze the NBS and the MNO's payoff under a particular bargaining sequence. Without loss of generality, we assume the bargaining sequence follows $1,2,...,N$, \emph{i.e.}, at step $n$, the MNO bargains with VO $n$.

For vector ${\bm b}_k=\left(b_1,b_2,\ldots,b_k\right)$, where $b_1,b_2,\ldots,b_k\in\left\{0,1\right\}$, we define ${W_{k}}\left( {{{\bm b}_{k}}} \right)$, $k\in{\cal N}$, as follows:
\begin{align}
{W_{k}}\left( {{{\bm b}_{k}}} \right) \triangleq \Psi\left(B_N^k\left({\bm b}_{k}\right)\right) - \sum\limits_{n = 1}^{k} {{b_n}{Q_n}}-\sum\limits_{n = k+1}^{N} {{\pi_n^*}\left(B_{n-1}^k\left({\bm b}_k\right)\right)}.\label{complicated}
\end{align}
In particular, we define $W_0$ as
\begin{align}
W_0 \triangleq U_0 = \Psi \left( {\hat{\bm b}}_N \right) - \sum\limits_{n = 1}^N {{{\hat \pi} _n}}.
\end{align}

Next we show $W_{k}\left( {{{\bm b}_{k}}} \right)$ have the following properties.
\begin{lemma}\label{lemma1}
For $k\in\left\{1,\!\ldots,N-1\right\}$ and any ${\bm b}_k$, we have
\begin{align}
\nonumber
& {W_k}\left( {{{\bm b}_k}} \right) = \frac{1}{2}{W_{k+1}}\left( {{{\bm b}_{k}},0}\right)\\
& +\frac{1}{2}\max \left\{ {{W_{k+1}}\left( {{\bm{b}}_k,0} \right),{W_{k+1}}\left( {{{\bm b}_{k}},1} \right) + {Q_{k+1}}} \right\}.\label{appendix:equ:pro1a}
\end{align}
In particular, we have
\begin{align}
{W_0} = \frac{1}{2}{W_{1}}\left( {0} \right) +\frac{1}{2}\max \left\{ {{W_{1}}\left(0\right),{W_{1}}\left( 1 \right) + {Q_{1}}} \right\}.\label{appendix:equ:pro1b}
\end{align}
\end{lemma}

\begin{proof}
We first prove (\ref{appendix:equ:pro1a}). Recall the NBS for step $k+1$:
\begin{align}
& \left( {b_{k+1}^*\left( {{{\bm b}_{k}}} \right),\pi _{k+1}^*\left( {{{\bm b}_{k}}} \right)} \right)=\\
& \left\{ {\begin{array}{*{20}{l}}
{\left( {1,\frac{1}{2}{\Delta _{k+1}}\left( {{{\bm b}_{k}}} \right)} \right),}&{{\rm if~}{\Delta _{k+1}}\left( {{{\bm b}_{k}}} \right)\ge0},\\
{\left( {0,0} \right),}&{\rm otherwise},
\end{array}} \right.\label{night}
\end{align}
where we define
\begin{align}
\nonumber
&{\Delta _{k+1}}\left( {{\bm b}_{k}} \right) \triangleq \\
\nonumber
& \Psi \left( B_{N}^{k+1}\left({\bm b}_{k},1\right) \right) -\sum\limits_{n = k+2}^{N} {{\pi_n^*}\left( B_{n-1}^{k+1}\left({\bm b}_{k},1\right) \right)}\\
& -\Psi \left( B_{N}^{k+1}\left({\bm b}_{k},0\right) \right)+\sum\limits_{n = k+2}^{N} {{\pi_n^*}\left( B_{n-1}^{k+1}\left({\bm b}_{k},0\right) \right)}.
\end{align}
By checking the definition of $W_{k}\left( {{{\bm b}_{k}}} \right)$, we find:
\begin{eqnarray}
{\Delta _{k+1}}\left( {{\bm b}_{k}} \right)=W_{k+1}\left( {{{\bm b}_{k}}},1\right)-W_{k+1}\left( {{{\bm b}_{k}}},0\right) +Q_{k+1}.\label{deltaandW}
\end{eqnarray}

Now we study the following two cases.

\textbf{Case 1}: ${\Delta _{k+1}}\left( {{\bm b}_{k}} \right)<0$.

Based on (\ref{night}), we have
\begin{eqnarray}
{b_{k+1}^*}\left( {\bm b}_{k} \right)=0 {\rm{~and~}} {\pi_{k+1}^*}\left( {\bm b}_{k} \right)=0.
\end{eqnarray}
According to the definition of $W_{k}\left( {{{\bm b}_{k}}} \right)$, we obtain
\begin{eqnarray}
{W_k}\left( {{{\bm b}_k}} \right) ={W_{k+1}}\left( {{{\bm b}_{k}},0} \right).\label{mark1}
\end{eqnarray}

\textbf{Case 2}: ${\Delta _{k+1}}\left( {{\bm b}_{k}} \right)\ge0$.

Based on (\ref{night}), we have
\begin{align}
&{b_{k+1}^*}\left( {\bm b}_{k} \right)=1,\\
&{\pi_{k+1}^*}\left( {\bm b}_{k} \right)=\frac{1}{2}\left({W_{k+1}\left({{{\bm b}_{k}}},1\right) -W_{k+1}\left( {{{\bm b}_{k}}},0\right)+Q_{k+1}}\right).
\end{align}
According to the definition of $W_{k}\left( {{{\bm b}_{k}}} \right)$, we obtain
\begin{align}
\nonumber
&{W_k}\left( {{{\bm b}_k}} \right)={W_{k+1}}\left( {{{\bm b}_{k}},1} \right)+Q_{k+1}-\pi_{k+1}^*\left({\bm b}_k\right)\\
&=\frac{1}{2}\left({W_{k+1}\left( {{{\bm b}_{k}}},1\right)+W_{k+1}\left( {{{\bm b}_{k}}},0\right)+Q_{k+1}}\right).\label{mark2}
\end{align}

Combining \textbf{Case 1} and \textbf{Case 2} completes the proof of equation (\ref{appendix:equ:pro1a}). We can use the similar approach to prove equation (\ref{appendix:equ:pro1b}), and the details are omitted.
\end{proof}

Let ${\bm b}_{k}^{1} \triangleq \left( {b_1^1, \ldots b_k^1} \right)$ and ${\bm b}_{k}^{2} \triangleq \left( {b_1^2, \ldots b_k^2} \right)$ for all $k\in{\cal N}$. We state the following lemmas.

\begin{lemma}\label{lemma2}
If for a particular $k\in {\cal N}$, we have $\sum_{n = 1}^k {b_n^1{X_n^t}}  = \sum_{n = 1}^k {b_n^2{X_n^t}}$, for all $t=1,2,\ldots,T$, then we have
\begin{eqnarray}
{W_{k}}\left( {{\bm b}_k^1} \right) = {W_{k}}\left( {{\bm b}_k^2} \right).\label{pro2}
\end{eqnarray}
\end{lemma}
\begin{proof}
We prove it by mathematical induction.

\textbf{Part A}: It's easy to show Lemma \ref{lemma2} is true for $k=N$.

\textbf{Part B}: We assume that Lemma \ref{lemma2} holds for a particular $k\in\left\{2,3,\ldots,N\right\}$, \emph{i.e.}, if for vectors ${\bm b}_k^1$ and ${\bm b}_k^2$, we have $\sum_{n = 1}^k {b_n^1{X_n^t}}  = \sum_{n = 1}^k {b_n^2{X_n^t}}$ for all $t$, we can obtain ${W_{k}}\left( {{\bm b}_k^1} \right) = {W_{k}}\left( {{\bm b}_k^2} \right)$. Now we check whether Lemma \ref{lemma2} also holds for $k-1$.

We assume that, for vectors ${\bm b}_{k-1}^1$ and ${\bm b}_{k-1}^2$, we have $\sum_{n = 1}^{k-1} {b_n^1{X_n^t}}  = \sum_{n = 1}^{k-1} {b_n^2{X_n^t}}$ for all $t$. Naturally, we get:
\begin{align}
\sum\limits_{n = 1}^{k-1} {b_n^1{X_n^t}} + 0\cdot X_k^t  = \sum\limits_{n = 1}^{k-1} {b_n^2{X_n^t}} + 0\cdot X_k^t, \forall t,\\
\sum\limits_{n = 1}^{k-1} {b_n^1{X_n^t}} + 1\cdot X_k^t  = \sum\limits_{n = 1}^{k-1} {b_n^2{X_n^t}} + 1\cdot X_k^t, \forall t.
\end{align}
Since Lemma \ref{lemma2} holds for $k$, we have:
\begin{align}
{W_{k}}\left({{\bm b}_{k-1}^1},0 \right) = {W_{k}}\left( {{\bm b}_{k-1}^2},0 \right),\label{morning1}\\
{W_{k}}\left( {{\bm b}_{k-1}^1},1 \right) = {W_{k}}\left( {{\bm b}_{k-1}^2},1 \right).\label{morning2}
\end{align}
According to (\ref{morning1}), (\ref{morning2}), and Lemma \ref{lemma1}, we conclude
\begin{eqnarray}
{W_{k-1}}\left( {{\bm b}_{k-1}^1} \right) = {W_{k-1}}\left( {{\bm b}_{k-1}^2} \right).
\end{eqnarray}
Therefore, we prove that Lemma \ref{lemma2} also holds for $k-1$.

Combining \textbf{Part A} and \textbf{Part B} completes the proof.
\end{proof}

\begin{lemma}\label{lemma3}
If for a particular $k\in \cal N$, we have $\sum_{n = 1}^k {b_n^1{X_n^t}}  \ge \sum_{n = 1}^k {b_n^2{X_n^t}}$ for all $t=1,2,\ldots,T$, then we have
\begin{eqnarray}
{W_{k}}\left( {{\bm b}_k^1} \right) \ge {W_{k}}\left( {{\bm b}_k^2} \right).\label{pro3}
\end{eqnarray}
\end{lemma}
\begin{proof}
We prove it by mathematical induction.

\textbf{Part A}: It is easy to show Lemma \ref{lemma3} is true for $k=N$.

\textbf{Part B}: We assume that Lemma \ref{lemma3} holds for a particular $k\in\left\{2,3,\ldots,N\right\}$, \emph{i.e.}, if for vectors ${\bm b}_k^1$ and ${\bm b}_k^2$, we have $\sum_{n = 1}^k {b_n^1{X_n^t}}  \ge \sum_{n = 1}^k {b_n^2{X_n^t}}$ for all $t$, we can obtain ${W_{k}}\left( {{\bm b}_k^1} \right) \ge {W_{k}}\left( {{\bm b}_k^2} \right)$. Now we check whether Lemma \ref{lemma3} also holds for $k-1$.

We assume that, for vectors ${\bm b}_{k-1}^1$ and ${\bm b}_{k-1}^2$, we have $\sum_{n = 1}^{k-1} {b_n^1{X_n^t}}  \ge \sum_{n = 1}^{k-1} {b_n^2{X_n^t}}$ for all $t$. Naturally, we get:
\begin{align}
\sum\limits_{n = 1}^{k-1} {b_n^1{X_n^t}} + 0\cdot X_k^t  \ge \sum\limits_{n = 1}^{k-1} {b_n^2{X_n^t}} + 0\cdot X_k^t,\forall t\\
\sum\limits_{n = 1}^{k-1} {b_n^1{X_n^t}} + 1\cdot X_k^t  \ge \sum\limits_{n = 1}^{k-1} {b_n^2{X_n^t}} + 1\cdot X_k^t,\forall t.
\end{align}
Since Lemma \ref{lemma3} holds for $k$, we have:
\begin{align}
{W_{k}}\left({{\bm b}_{k-1}^1},0 \right) \ge {W_{k}}\left( {{\bm b}_{k-1}^2},0 \right),\label{morning3}\\
{W_{k}}\left( {{\bm b}_{k-1}^1},1\right) \ge {W_{k}}\left( {{\bm b}_{k-1}^2},1 \right).\label{morning4}
\end{align}
According to (\ref{morning3}), (\ref{morning4}), and Lemma \ref{lemma1}, we conclude
\begin{eqnarray}
{W_{k-1}}\left( {{\bm b}_{k-1}^1} \right) \ge {W_{k-1}}\left( {{\bm b}_{k-1}^2} \right).
\end{eqnarray}
Therefore, we prove that Lemma \ref{lemma3} also holds for $k-1$.

Combining \textbf{Part A} and \textbf{Part B} completes the proof.
\end{proof}

\begin{lemma}\label{lemma4}
If the following two inequalities hold:
\begin{eqnarray}
{W_{k + 1}}\left( {{\bm b}_k^1,0} \right) + \delta  > {W_{k + 1}}\left(  {{\bm b}_k^2,0} \right),\\
{W_{k + 1}}\left( {{\bm b}_k^1,1} \right) + \delta  > {W_{k + 1}} \left( {{\bm b}_k^2,1} \right),
\end{eqnarray}
where $k\in\left\{1,\ldots,N-1\right\}$ and $\delta>0$, we have:
\begin{eqnarray}
{W_{k}}\left( {{\bm b}_k^1} \right) + \delta  > {W_{k}}\left( {{\bm b}_k^2} \right).
\end{eqnarray}
\end{lemma}
\begin{proof}
This is obvious by checking Lemma \ref{lemma1}.
\end{proof}

\section{Proof of Proposition \ref{proposition1}}
\begin{proof}
Without loss of generality, we assume that VO $k\in\cal{N}$ is of type $1$, \emph{i.e.}, $Q_k\ge0$. Furthermore, we assume that the MNO and the first $k-1$ VOs reached ${{\bm b}_{k-1}}$ in the first $k-1$ steps of bargaining. By checking the definition of ${\Delta _{k}}\left( {{\bm b}_{k-1}} \right)$, we have:
\begin{eqnarray}
{\Delta _{k}}\left( {{\bm b}_{k-1}} \right)=W_{k}\left( {{{\bm b}_{k-1}}},1 \right)-W_{k}\left( {{{\bm b}_{k-1}}},0 \right)+Q_{k}.\label{waha}
\end{eqnarray}
Since $\sum_{n = 1}^{k-1} {b_n {X_n^t}}+X_k^t\ge \sum_{n = 1}^{k-1} {b_n {X_n^t}}$ for all $t$, according to Lemma \ref{lemma3}, we have:
\begin{eqnarray}
W_{k}\left({{{\bm b}_{k-1}}},1\right)\ge W_{k}\left( {{{\bm b}_{k-1}}},0\right).
\end{eqnarray}
Together with $Q_k\ge0$, we conclude that ${\Delta _{k}}\!\left( {{\bm b}_{k\!-\!1}} \right)\ge0$. Based on (\ref{solutionk}), ${\Delta _{k}}\left( {{\bm b}_{k-1}} \right)\ge0$ implies $b_k^*\left( {{\bm b}_{k-1}} \right)=1$. Therefore, the MNO definitely cooperates with VO $k$.
\end{proof}

\section{Proof of Proposition \ref{proposition2}}
\begin{proof}
Without loss of generality, we assume that VO $k\in\cal{N}$ is of type $3$, \emph{i.e.}, $Q_k<0$ and $\sum_{t=1}^{T} {f_t\left( {{X_k^t}} \right)} + {Q_k} < 0$. Furthermore, we assume that the MNO and the first $k-1$ VOs reached ${{\bm b}_{k-1}}$ in the first $k-1$ steps of bargaining.

First of all, we consider the following two functions: ${W_{N}}\left( {{\bm b}_{k-1},0,{\bm 0}_{N-k}} \right)$ and ${W_{N}}\left( {{\bm b}_{k-1},1,{\bm 0}_{N-k}} \right)$, where we define ${\bm 0}_{N-k}$ as the vector that has $N-k$ zeros as its entries. Based on the definition,
\begin{align}
&{W_{N}}\left( {{\bm b}_{k-1},0,{\bm 0}_{N-k}} \right)=\sum\limits_{t=1}^T f_t\left(\sum\limits_{n = 1}^{k-1} {b_n{X_n^t}}\right),\label{threeA}\\
&{W_{N}}\left( {{\bm b}_{k-1},1,{\bm 0}_{N-k}} \right)=\sum\limits_{t=1}^T f_t\left(\sum\limits_{n = 1}^{k-1} {b_n{X_n^t}}+X_k^t\right).\label{threeB}
\end{align}
Due to the concavity of function $f_t\left(\cdot\right),t=1,2\ldots,T$, we have:
\begin{eqnarray}
f_t\left(\sum\limits_{n = 1}^{k-1} {b_n{X_n^t}}+X_k^t\right)-f_t\left(\sum\limits_{n = 1}^{k-1} {b_n{X_n^t}}\right)\le f_t\left(X_k^t\right), \forall t,
\end{eqnarray}
Since $\sum_{t=1}^{T} {f_t\left( {{X_k^t}} \right)} + {Q_k} < 0$, we conclude that
\begin{eqnarray}
\sum\limits_{t=1}^T f_t\left(\sum\limits_{n = 1}^{k-1} {b_n{X_n^t}}+X_k^t\right)-\sum\limits_{t=1}^T f_t\left(\sum\limits_{n = 1}^{k-1} {b_n{X_n^t}}\right)+Q_k<0.\label{threeC}
\end{eqnarray}
Based on (\ref{threeA}), (\ref{threeB}), and (\ref{threeC}), we conclude
\begin{align}
{W_{N}}\left( {{\bm b}_{k-1},0,{\bm 0}_{N-k}} \right)-Q_k>{W_{N}}\left( {{\bm b}_{k-1},1,{\bm 0}_{N-k}} \right).\label{108A}
\end{align}
Similarly, we next consider the following two functions: ${W_{N}}\left( {{\bm b}_{k-1},0,{\bm 0}_{N-k-1},1} \right)$ and ${W_{N}}\left( {{\bm b}_{k-1},1,{\bm 0}_{N-k-1},1} \right)$. We can also prove
\begin{align}
{W_{N}}\left( {{\bm b}_{k-1},0,{\bm 0}_{N-k-1},1} \right)-Q_k>{W_{N}}\left( {{\bm b}_{k-1},1,{\bm 0}_{N-k-1},1} \right).\label{108B}
\end{align}
Recall that $-Q_k>0$, based on (\ref{108A}), (\ref{108B}), and Lemma \ref{lemma4}, we conclude,
\begin{align}
{W_{N-1}}\left( {{\bm b}_{k-1},0,{\bm 0}_{N-1-k}} \right)-Q_k> {W_{N-1}}\left( {{\bm b}_{k-1},1,{\bm 0}_{N-1-k}} \right).\label{secondA}
\end{align}
Repeating these processes, we can eventually conclude
\begin{align}
{W_{k}}\left( {{\bm b}_{k-1},0} \right)-Q_k>{W_{k}}\left( {{\bm b}_{k-1},1} \right).
\end{align}
Recall (\ref{waha}), where we express ${\Delta _{k}}\left( {{\bm b}_{k-1}} \right)$ as
\begin{eqnarray}
{\Delta _{k}}\left( {{\bm b}_{k-1}} \right)=W_{k}\left( {{{\bm b}_{k-1}}},1 \right)-W_{k}\left( {{{\bm b}_{k-1}}},0 \right)+Q_{k}.
\end{eqnarray}
We conclude that ${\Delta _{k}}\!\left( {{\bm b}_{k\!-\!1}} \right)<0$. Based on (\ref{solutionk}), $b_k^*\left( {{\bm b}_{k-1}} \right)=0$, \emph{i.e.}, the MNO does not cooperate with VO $k$.
\end{proof}

\section{Preliminary Lemmas II}\label{appendix:section:pro2}
In this section, we continue to prove lemmas, which are useful to show the propositions and theorems in the paper.

Same as Appendix \ref{appendix:section:pro1}, we first assume the bargaining sequence follows $1,2,\ldots,N$, and define $W_k\left({\bm b}_k\right)$ as (\ref{complicated}). Then we interchange the bargaining positions of VO $k+1$ and VO $k+2$, where $k\in\left\{0,1,\ldots,N-2\right\}$. That is, the MNO bargains with VO $k+2$ at step $k+1$, and bargains with VO $k+1$ at step $k+2$. For the new sequence, we define
\begin{align}
{{\tilde W}_k}\left( {{{\bm b}_{k}}} \right) \triangleq  \Psi\left({{\tilde B}_N^k}\left({\bm b}_{k}\right)\right) - \sum\limits_{n = 1}^{k} {{b_n}{Q_n}}-\sum\limits_{n = k+1}^{N} {{{\tilde \pi}_n^*}\left({\tilde B}_{n-1}^k\left({\bm b}_k\right)\right)}.
\label{complicatedprime}
\end{align}
Here, functions ${{\tilde B}_N^k}\left({\bm b}_{k}\right)$ and ${{{\tilde \pi}_n^*}\left({\tilde B}_{n-1}^k\left({\bm b}_k\right)\right)}$ are defined for the new bargaining sequence, and are generally not equal to  ${{ B}_N^k}\left({\bm b}_{k}\right)$ and ${{{ \pi}_n^*}\left({ B}_{n-1}^k\left({\bm b}_k\right)\right)}$ in (\ref{complicated}).

Next we state the following lemmas.
\begin{lemma}\label{lemma5}
If $Q_{k+2}\ge0$, we have ${{\tilde W}_k}\left( {{{\bm b}_{k}}} \right) \ge {W_k}\left( {{{\bm b}_{k}}} \right)$ for any ${\bm b}_{k}$.
\end{lemma}
\begin{proof}
First, we study $W_k\left({\bm b}_k\right)$. We define:
\begin{align}
A \triangleq W_{k+2}\left({\bm b}_k,0,0\right),\\
B \triangleq W_{k+2}\left({\bm b}_k,0,1\right),\\
C \triangleq W_{k+2}\left({\bm b}_k,1,0\right),\\
D \triangleq W_{k+2}\left({\bm b}_k,1,1\right).
\end{align}
According to Lemma \ref{lemma1},
\begin{align}
&W_{k+1}\left({\bm b}_k,0\right)=\frac{1}{2}A+\frac{1}{2}\max\left\{A,B+Q_{k+2}\right\},\label{W1k}\\
&W_{k+1}\left({\bm b}_k,1\right)=\frac{1}{2}C+\frac{1}{2}\max\left\{C,D+Q_{k+2}\right\},\\
\nonumber
&W_k\left({\bm b}_k\right)=\frac{1}{2}W_{k+1}\left({\bm b}_k,0\right)+{~~~~~~~~~~~~~~~~~~~~~}\\
&\frac{1}{2}\max\left\{W_{k+1}\left({\bm b}_k,0\right),W_{k+1}\left({\bm b}_k,1\right)+Q_{k+1}\right\}.\label{Wk}
\end{align}
Based on Lemma \ref{lemma3}, we conclude
\begin{eqnarray}
B\ge A {\rm{~and~}}D\ge C.
\end{eqnarray}
Since $Q_{k+2}\ge0$, we further have
\begin{eqnarray}
B+Q_{k+2}\ge A {\rm{~and~}}D+Q_{k+2}\ge C.
\end{eqnarray}
Therefore, we rewrite (\ref{W1k})-(\ref{Wk}) as:
\begin{align}
&W_{k+1}\left({\bm b}_k,0\right)=\frac{1}{2}A+\frac{1}{2}B+\frac{1}{2} Q_{k+2},\\
&W_{k+1}\left({\bm b}_k,1\right)=\frac{1}{2}C+\frac{1}{2}D+\frac{1}{2}Q_{k+2},\\
\nonumber
&W_k\left({\bm b}_k\right)=\frac{1}{4}A+\frac{1}{4}B+\frac{1}{4} Q_{k+2}{~~~~~~~~~}\\
\nonumber
&{~~~~~~~~~~~~~~}+\max\left\{\frac{1}{4}A+\frac{1}{4}B+\frac{1}{4} Q_{k+2},\right.\\
&{~~~~~~~~~~~~~~~}\left.\frac{1}{4}C+\frac{1}{4}D+\frac{1}{4}Q_{k+2}+\frac{1}{2}Q_{k+1}\right\}.\label{needA}
\end{align}
Next we study ${\tilde W}_k\left({\bm b}_k\right)$. We define:
\begin{align}
{\tilde A} \triangleq {\tilde W}_{k+2}\left({\bm b}_k,0,0\right),\\
{\tilde B} \triangleq {\tilde W}_{k+2}\left({\bm b}_k,0,1\right),\\
{\tilde C} \triangleq {\tilde W}_{k+2}\left({\bm b}_k,1,0\right),\\
{\tilde D} \triangleq {\tilde W}_{k+2}\left({\bm b}_k,1,1\right).
\end{align}

According to Lemma \ref{lemma1},
\begin{align}
&{\tilde W}_{k+1}\left({\bm b}_k,0\right)=\frac{1}{2}{\tilde A}+\frac{1}{2}\max\left\{{\tilde A},{\tilde B}+Q_{k+1}\right\},\\
&{\tilde W}_{k+1}\left({\bm b}_k,1\right)=\frac{1}{2}{\tilde C}+\frac{1}{2}\max\left\{{\tilde C},{\tilde D}+Q_{k+1}\right\},\\
\nonumber
&{\tilde W}_k\left({\bm b}_k\right)=\frac{1}{2}{\tilde W}_{k+1}\left({\bm b}_k,0\right)+{~~~~~~~~~~~~~~~~~~~~~}\\
&\frac{1}{2}\max\left\{{\tilde W}_{k+1}\left({\bm b}_k,0\right),{\tilde W}_{k+1}\left({\bm b}_k,1\right)+Q_{k+2}\right\}.
\end{align}
Based on Lemma \ref{lemma3}, we conclude
\begin{eqnarray}
{\tilde W}_{k+1}\left({\bm b}_k,1\right)\ge {\tilde W}_{k+1}\left({\bm b}_k,0\right).
\end{eqnarray}
Since $Q_{k+2}\ge0$, we further have
\begin{eqnarray}
{\tilde W}_{k+1}\left({\bm b}_k,1\right)+Q_{k+2}\ge {\tilde W}_{k+1}\left({\bm b}_k,0\right).
\end{eqnarray}
Therefore, we have
\begin{align}
\nonumber
&{\tilde W}_k\left({\bm b}_k\right)=\frac{1}{4}{\tilde A}+\frac{1}{4}\max\left\{{\tilde A},{\tilde B}+Q_{k+1}\right\}\\
&{~~}+\frac{1}{4}{\tilde C}+\frac{1}{4}\max\left\{{\tilde C},{\tilde D}+Q_{k+1}\right\}+\frac{1}{2}{Q_{k+2}}.\label{markkk}
\end{align}

Next we compare $W_k\left({\bm b}_k\right)$ and ${\tilde W}_k\left({\bm b}_k\right)$. Based on Lemma \ref{lemma2}, we have:
\begin{eqnarray}
A={\tilde A}{\rm{~and~}}D={\tilde D},\label{uu1}\\
B={\tilde C}{\rm{~and~}}C={\tilde B}.\label{uu2}
\end{eqnarray}
Based on (\ref{uu1}) and (\ref{uu2}), we rewrite (\ref{markkk}) as
\begin{align}
\nonumber
&{\tilde W}_k\left({\bm b}_k\right)=\frac{1}{4}A+\frac{1}{4}\max\left\{A,C+Q_{k+1}\right\}\\
&{~~}+\frac{1}{4}B+\frac{1}{4}\max\left\{B,D+Q_{k+1}\right\}+\frac{1}{2}{Q_{k+2}}.\label{needB}
\end{align}
By (\ref{needA}) and (\ref{needB}), we obtain
\begin{align}
\nonumber
&{\tilde W}_k\left({\bm b}_k\right)-W_k\left({\bm b}_k\right)=\\
\nonumber
&\frac{1}{4}\max\left\{A,C+Q_{k+1}\right\}+\frac{1}{4}\max\left\{B,D+Q_{k+1}\right\}\\
&-\max\left\{\frac{1}{4}A+\frac{1}{4}B,\frac{1}{4}C+\frac{1}{4}D+\frac{1}{2}Q_{k+1}\right\}.
\end{align}
It is easy to check that ${\tilde W}_k\left({\bm b}_k\right)-W_k\left({\bm b}_k\right)\ge0$. Here we complete the proof.
\end{proof}

\begin{lemma}\label{lemma:add1}
If VO $k+2$ is of type $3$, we have ${{\tilde W}_k}\left( {{{\bm b}_{k}}} \right) = {W_k}\left( {{{\bm b}_{k}}} \right)$ for any ${\bm b}_{k}$.
\end{lemma}
\begin{proof}
First, we study $W_k\left({\bm b}_k\right)$. We define:
\begin{align}
A \triangleq W_{k+2}\left({\bm b}_k,0,0\right),\\
B \triangleq W_{k+2}\left({\bm b}_k,0,1\right),\\
C \triangleq W_{k+2}\left({\bm b}_k,1,0\right),\\
D \triangleq W_{k+2}\left({\bm b}_k,1,1\right).
\end{align}
According to Lemma \ref{lemma1},
\begin{align}
&W_{k+1}\left({\bm b}_k,0\right)=\frac{1}{2}A+\frac{1}{2}\max\left\{A,B+Q_{k+2}\right\},\label{W1k:add}\\
&W_{k+1}\left({\bm b}_k,1\right)=\frac{1}{2}C+\frac{1}{2}\max\left\{C,D+Q_{k+2}\right\},\\
\nonumber
&W_k\left({\bm b}_k\right)=\frac{1}{2}W_{k+1}\left({\bm b}_k,0\right)+{~~~~~~~~~~~~~~~~~~~~~}\\
&\frac{1}{2}\max\left\{W_{k+1}\left({\bm b}_k,0\right),W_{k+1}\left({\bm b}_k,1\right)+Q_{k+1}\right\}.\label{Wk:add}
\end{align}
Based on the proof of Proposition \ref{proposition2}, since VO $k+2$ is of type $3$, we have
\begin{align}
A>B+Q_{k+2} {\rm~and~} C>D+Q_{k+2}.
\end{align}
Therefore, we rewrite (\ref{W1k:add})-(\ref{Wk:add}) as:
\begin{align}
&W_{k+1}\left({\bm b}_k,0\right)=A,\\
&W_{k+1}\left({\bm b}_k,1\right)=C,\\
&W_k\left({\bm b}_k\right)=\frac{1}{2}A+\frac{1}{2}\max\left\{A,C+Q_{k+1}\right\}.\label{uu1:add}
\end{align}
Next we study ${\tilde W}_k\left({\bm b}_k\right)$. We define:
\begin{align}
{\tilde A} \triangleq {\tilde W}_{k+2}\left({\bm b}_k,0,0\right),\\
{\tilde B} \triangleq {\tilde W}_{k+2}\left({\bm b}_k,0,1\right),\\
{\tilde C} \triangleq {\tilde W}_{k+2}\left({\bm b}_k,1,0\right),\\
{\tilde D} \triangleq {\tilde W}_{k+2}\left({\bm b}_k,1,1\right).
\end{align}

According to Lemma \ref{lemma1},
\begin{align}
&{\tilde W}_{k+1}\left({\bm b}_k,0\right)=\frac{1}{2}{\tilde A}+\frac{1}{2}\max\left\{{\tilde A},{\tilde B}+Q_{k+1}\right\},\\
&{\tilde W}_{k+1}\left({\bm b}_k,1\right)=\frac{1}{2}{\tilde C}+\frac{1}{2}\max\left\{{\tilde C},{\tilde D}+Q_{k+1}\right\},\\
\nonumber
&{\tilde W}_k\left({\bm b}_k\right)=\frac{1}{2}{\tilde W}_{k+1}\left({\bm b}_k,0\right)+{~~~~~~~~~~~~~~~~~~~~~}\\
&\frac{1}{2}\max\left\{{\tilde W}_{k+1}\left({\bm b}_k,0\right),{\tilde W}_{k+1}\left({\bm b}_k,1\right)+Q_{k+2}\right\}.
\end{align}
Based on the proof of Proposition \ref{proposition2}, since VO $k+2$ is of type $3$, we have
\begin{align}
{\tilde W}_{k+1}\left({\bm b}_k,0\right)>{\tilde W}_{k+1}\left({\bm b}_k,1\right)+Q_{k+2}.
\end{align}
Therefore, we have
\begin{align}
{\tilde W}_k\left({\bm b}_k\right)=\frac{1}{2}{\tilde A}+\frac{1}{2}\max\left\{{\tilde A},{\tilde B}+Q_{k+1}\right\}.\label{uu2:add}
\end{align}

Next we compare $W_k\left({\bm b}_k\right)$ and ${\tilde W}_k\left({\bm b}_k\right)$. Based on Lemma \ref{lemma2}, we have:
\begin{eqnarray}
A={\tilde A}{\rm{~and~}}D={\tilde D},\\
B={\tilde C}{\rm{~and~}}C={\tilde B}.
\end{eqnarray}
Based on (\ref{uu1:add}) and (\ref{uu2:add}), we obtain that ${\tilde W}_k\left({\bm b}_k\right)=W_k\left({\bm b}_k\right)$, where we complete the proof.
\end{proof}

\begin{lemma}\label{lemma6}
If ${{\tilde W}_k}\left( {{{\bm b}_{k}}} \right) \ge {W_k}\left( {{{\bm b}_{k}}} \right)$ for any ${\bm b}_{k}$, the MNO's payoff does not decrease after exchanging VO $k+1$ and VO $k+2$'s bargaining positions.
\end{lemma}
\begin{proof}
Similar as ${\tilde W}_k\left({\bm b}_k\right)$, we define ${\tilde W}_{k-1}\left({\bm b}_{k-1}\right)$ for the new sequence after the position exchange. Based on Lemma \ref{lemma1}, we have the following equalities for the sequences before and after the position exchange.
\begin{align}
\nonumber
& {W_{k-1}}\left( {{{\bm b}_{k-1}}} \right) = \frac{1}{2}{W_{k}}\left( {{{\bm b}_{k-1}},0} \right)\\
& +\frac{1}{2}\max \left\{ {{W_{k}}\left( {{\bm{b}}_{k-1},0} \right),{W_{k}}\left( {{{\bm b}_{k-1}},1} \right)+ {Q_{k}}} \right\},\label{goneA}\\
\nonumber
& {{\tilde W}_{k-1}}\left( {{{\bm b}_{k-1}}} \right) = \frac{1}{2}{{\tilde W}_{k}}\left({{{\bm b}_{k-1}},0} \right)\\
& +\frac{1}{2}\max \left\{ {{{\tilde W}_{k}}\left({{\bm{b}}_{k-1},0}\right),{{\tilde W}_{k}}\left( {{{\bm b}_{k-1}},1}\right) + {Q_{k}}} \right\}\label{goneB}.
\end{align}
Since ${{\tilde W}_k}\left( {{{\bm b}_{k}}} \right) \ge {W_k}\left( {{{\bm b}_{k}}} \right)$ for any ${\bm b}_{k}$, from (\ref{goneA}) and (\ref{goneB}), we conclude that ${{\tilde W}_{k-1}}\left( {{{\bm b}_{k-1}}} \right) \ge {W_{k-1}}\left( {{{\bm b}_{k-1}}} \right)$ for any ${\bm b}_{k-1}$.

Similarly, we define ${\tilde W}_{k-2}\left({\bm b}_{k-2}\right)$ for the new sequence and conclude that ${{\tilde W}_{k-2}}\left( {{{\bm b}_{k-2}}} \right) \ge {W_{k-2}}\left( {{{\bm b}_{k-2}}} \right)$ for any ${\bm b}_{k-2}$.

Repeating the process, we can eventually conclude that ${{\tilde W}_{0}} \ge {W_{0}}$. According to the definitions of $W$ and $\tilde W$, we have $W_0=U_0$ and ${\tilde W}_0={\tilde U}_0$, where $U_0$ and ${\tilde U}_0$ are the MNO's payoffs before and after the position exchange. Hence, we conclude that
\begin{align}
{\tilde U}_0\ge U_0.
\end{align}
In other words, if ${{\tilde W}_k}\left( {{{\bm b}_{k}}} \right) \ge {W_k}\left( {{{\bm b}_{k}}} \right)$ for any ${\bm b}_{k}$, the MNO's payoff does not decrease after exchanging VO $k+1$ and VO $k+2$'s positions. Here we complete the proof.
\end{proof}

\section{Proof of Proposition \ref{proposition3}}
It is easy to prove Proposition \ref{proposition3} by directly combing Lemma \ref{lemma5} and Lemma \ref{lemma6} introduced in the last section.

\section{Proof of Proposition \ref{proposition4}}
From Lemma \ref{lemma:add1}, we know that if VO $k+2$ is of type $3$, we have ${{\tilde W}_k}\left( {{{\bm b}_{k}}} \right) = {W_k}\left( {{{\bm b}_{k}}} \right)$ for any ${\bm b}_{k}$. From Lemma \ref{lemma6}, we conclude that if ${{\tilde W}_k}\left( {{{\bm b}_{k}}} \right) = {W_k}\left( {{{\bm b}_{k}}} \right)$ for any ${\bm b}_{k}$, the MNO's payoff does not change after exchanging VO $k+1$ and VO $k+2$'s bargaining positions. By combining these two statements, we can easily prove Proposition \ref{proposition4}.



\section{Proof of Theorem \ref{theoremA}}
\begin{proof}
We first prove the existence of set $\mathcal{L}^{\ast}$. We assume that $\bm l^1=\left(l_1^1,l_2^1,\ldots,l_N^1\right)$ is one of the optimal bargaining sequences. Next we show that we can rearrange the VOs' bargaining positions in sequence $\bm l^1$ and obtain a new optimal bargaining sequence that lies in set $\mathcal{L}^{\ast}$.

We first find out the type $1$ VO with the earliest bargaining position in sequence $\bm l^1$. We assume that the index of this type $1$ VO is $l_n^1$, \emph{i.e.}, its bargaining position is $n$. If $n\ne 1$, we move VO $l_n^1$ to the first bargaining position, and obtain a new sequence ${\bm l}^2=\left(l_1^2,l_2^2,\ldots,l_N^2\right)$. Mathematically,
\begin{align}
&l_1^2=l_n^1,\\
&l_i^2=l_{i-1}^1,\forall i\in\left\{2,3,\ldots,n\right\},\\
&l_i^2=l_i^1,\forall i\in\left\{n+1,n+2,\ldots,N\right\}.
\end{align}

For sequence $\bm l^2$, we find the type $1$ VO with the $2$nd earliest bargaining position, move it to the $2$nd bargaining position, and generate a new sequence $\bm l^3$. Repeating the process $N_1$ times (recall that $N_1$ is the number of type $1$ VOs), we obtain a sequence $\bm l^{N_1}$. Apparently, the first $N_1$ VOs in sequence $\bm l^{N_1}$ are of type $1$. Based on Proposition \ref{proposition3} and the optimality of $\bm l^1$, it is easy to conclude that sequences ${\bm l}^1,{\bm l}^2,\ldots,{\bm l}^{N_1}$ generate the same MNO's payoff and all of them are optimal sequences.

Then we apply the similar rule to move all type $3$ VOs in sequence $\bm l^{N_1}$ to the last $N_3$ bargaining positions. We denote the resulting sequences as ${\bm l}^{N_1+1},{\bm l}^{N_1+2},\ldots,{\bm l}^{N_1+N_3}$. Based on Proposition \ref{proposition4}, sequences ${\bm l}^{N_1},{\bm l}^{N_1+1},\ldots,{\bm l}^{N_1+N_3}$ generate the same MNO's payoff. In other words, ${\bm l}^{N_1+N_3}$ is also one of the optimal bargaining sequences.

The first $N_1$ VOs in sequence $\bm l^{N_1+N_3}$ are of type $1$, and the last $N_3$ VOs in sequence $\bm l^{N_1+N_3}$ are of type $3$. This means that the optimal bargaining sequence $\bm l^{N_1+N_3}$ lies in set ${\cal L}^\ast$, which shows the existence of the non-empty set ${\cal L}^\ast$.

For any $\boldsymbol{l} \in \mathcal{L}^{\ast}$, there are only type $1$ VOs between any two non-adjacent type $1$ VOs. Based on Proposition \ref{proposition3}, it is easy to show that if the MNO interchanges the bargaining positions of any two type 1 VOs in $\boldsymbol{l} \in \mathcal{L}^{\ast}$, the MNO's payoff will not change. Similarly, for any $\boldsymbol{l} \in \mathcal{L}^{\ast}$, there are only type $3$ VOs between any two non-adjacent type $3$ VOs. Based on Proposition \ref{proposition4}, it is easy to show that if the MNO interchanges the bargaining positions of any two type 3 VOs in $\boldsymbol{l} \in \mathcal{L}^{\ast}$, the MNO's payoff will not change. Here we complete the proof.
\end{proof}



\vspace{-0.4cm}
\section{Proof of Theorem \ref{theoremadd}}
\begin{proof}
Based on Algorithm \ref{algorithm2}, the first $N_1$ and last $N_3$ VOs in sequence ${\bm l}^{RE}=\left(l_1^{RE},l_2^{RE},\ldots,l_N^{RE}\right)$ are of type $1$ and type $3$, respectively. Therefore, to show that ${\bm l}^{R\!E}$ lies in set ${\cal L}^*$, we only need to prove that bargaining sequence ${\bm l}^{RE}$ optimizes the MNO's payoff.

We first show that there exists at least one optimal bargaining sequence in set ${\cal L}^{RE}$. Based on Theorem \ref{theoremA}, ${\cal L}^\ast$ is non-empty and we pick a sequence ${\bm l}^A=\left(l_1^{A},l_2^{A},\ldots,l_N^{A}\right)$ from set $\in{\cal L}^\ast$. For sequence ${{\bm l}^A}$, the first $N_1$ VOs are of type $1$, the last $N_3$ VOs are of type $3$, and the remaining $N_2$ VOs in the middle are of type $2$. Based on Theorem \ref{theoremA}, if we interchange the bargaining positions of any two type $1$ VOs or any two type $3$ VOs in ${{\bm l}^A}$, the MNO's payoff will not change. Therefore, by interchanging the bargaining positions of type $1$ or type $3$ VOs in ${{\bm l}^A}$, we can obtain a new optimal bargaining sequence ${{\bm l}^B}=\left(l_1^{B},l_2^{B},\ldots,l_N^{B}\right)$, where
\begin{align}
& l_i^{B}=l_i^{RE}, \forall i\in\left\{1,2,\ldots,N_1\right\},\\
& l_i^{B}=l_i^{A}, \forall i\in\left\{N_1+1,N_1+2,\ldots,N_1+N_2\right\},\\
& l_i^{B}=l_i^{RE}, \forall i\in\left\{N_1+N_2+1,N_1+N_2+2,\ldots,N\right\}.
\end{align}
In other words, the first $N_1$ VOs and the last $N_3$ VOs in sequence ${{\bm l}^B}$ are the same as those in sequence ${{\bm l}^{RE}}$, and the remaining $N_2$ VOs in the middle of sequence ${{\bm l}^B}$ are the same as those in sequence ${{\bm l}^{A}}$.
Apparently, ${{\bm l}^B}$ lies in set ${\cal L}^{RE}$. Since sequence ${{\bm l}^B}$ is optimal, there exists at least one optimal bargaining sequence in set ${\cal L}^{RE}$.

From Algorithm \ref{algorithm2}, we have ${\bm l}^{RE}=\mathop {\arg\!\max}_{{\bm l} \in {{\cal L}^{R\!E}}} U_0^{\bm l}$. Hence, it is easy to conclude that ${\bm l}^{RE}$ is also an optimal bargaining sequence for the MNO, and it lies in set ${\cal L}^*$. Here we complete the proof.
\end{proof}


\section{Proof of Theorem \ref{theoremB}}

\begin{proof}
Based on Theorem \ref{theoremA}, when all VOs are of type $1$, all bargaining sequences generate the same MNO's payoff. Without loss of generality, we consider sequence ${\bm l}=\left(1,2,\ldots,N\right)$, \emph{i.e.}, the MNO bargains with VO $n$ at step $n\in{\cal N}$.

To facilitate the proof, we define $W_0$ as
\begin{align}
W_0 \triangleq U_0 = \Psi \left( {\hat{\bm b}}_N \right) - \sum\limits_{n = 1}^N {{{\hat \pi} _n}},\label{main:defineW0}
\end{align}
and define function ${W_{k}}\left( {{{\bm b}_{k}}} \right)$, $k\in{\cal N}$, as
\begin{align}
\nonumber
{W_{k}}\left( {{{\bm b}_{k}}} \right) \triangleq & \Psi\left(B_N^k\left({\bm b}_{k}\right)\right) - \sum\limits_{n = 1}^{k} {{b_n}{Q_n}}\\
&-\sum\limits_{n = k+1}^{N} {{\pi_n^*}\left(B_{n-1}^k\left({\bm b}_k\right)\right)}.\label{main:complicated}
\end{align}

From Proposition \ref{proposition1}, the MNO cooperates with all VOs. Hence, we have ${\hat b}_1=1$, and we can rewrite $W_0$ in (\ref{main:defineW0}) as
\begin{align}
W_0=\Psi\left(B_N^1\left(1\right)\right) -{\pi}_1^*- \sum\limits_{n = 2}^N {{\pi_n^*}\left(B_{n-1}^1\left(1\right)\right)}.\label{main:W0}
\end{align}
From (\ref{main:complicated}), we have
\begin{align}
{W_{1}}\left(1\right) = & \Psi\left(B_N^1\left(1\right)\right) - Q_1-\sum\limits_{n = 2}^{N} {{\pi_n^*}\left(B_{n-1}^1\left(1\right)\right)}.
\end{align}
Hence, we can further rewrite $W_0$ in (\ref{main:W0}) as
\begin{align}
W_0=W_1\left(1\right)+Q_1-\pi_1^*.\label{main:W0:2}
\end{align}
According to (\ref{solution1}) and the fact that the MNO cooperates with VO $1$, we have $\pi_1^*=\frac{1}{2}\Delta_1$. Furthermore, based on the definition of $\Delta_1$ in (\ref{equ:delta1}), and the definitions of $W_1\left(0\right)$ and $W_1\left(1\right)$ in (\ref{main:complicated}), we have
\begin{align}
\Delta _{1}=W_{1}\left(1\right)-W_{1}\left(0\right)+Q_{1}.
\end{align}
Therefore, we can rewrite $W_0$ in (\ref{main:W0:2}) as
\begin{align}
\nonumber
W_0& = W_1\left(1\right)+Q_1-\frac{1}{2}\left(W_{1}\left(1\right)-W_{1}\left(0\right)+Q_{1}\right)\\
& = \frac{1}{2}{W_{1}}\left( {0} \right) +\frac{1}{2} {W_{1}}\left( 1 \right) + \frac{1}{2}{Q_{1}}.\label{main:W0:3}
\end{align}

Based on the similar approach, we can show that for $k\in\left\{1,\!\ldots,N-1\right\}$ and any ${\bm b}_k$, we have
\begin{align}
{W_k}\left( {{{\bm b}_k}} \right) = \frac{1}{2}{W_{k+1}}\left( {{{\bm b}_{k}},0} \right)+\frac{1}{2} {W_{k+1}}\left( {{{\bm b}_{k}},1} \right) +\frac{1}{2}{Q_{k+1}}.\label{main:useful}
\end{align}

By choosing $k=1$ and $b_1=0$ or $b_1=1$ in (\ref{main:useful}) to further expand ${W_{1}}\left( {0} \right)$ and ${W_{1}}\left( {1}\right)$ in (\ref{main:W0:3}), we obtain
\begin{align}
\nonumber
{W_0} & = \frac{1}{4}{W_{2}}\left( {0,0}\right)+\frac{1}{4}{W_{2}} \left({0,1}\right)+\frac{1}{2}{Q_{2}}\\
&{~~~~}+\frac{1}{4}{W_{2}}\left( {1,0}\right)+\frac{1}{4}{W_{2}}\left( {1,1}\right)+\frac{1}{2}{Q_{1}}.
\end{align}
Repeating the process above, we eventually obtain the following equality:
\begin{align}
\nonumber
{W_0}& = \frac{1}{{{2^N}}} \sum\limits_{{{\bm b}_N}\in\cal B} {W_N \left( {{\bm b}_N} \right)}+\frac{1}{2}\sum\limits_{n\in\cal N} {Q_n}\\
&=\frac{1}{{{2^N}}} \sum\limits_{{{\bm b}_N}\in\cal B} {\Psi \left( {{\bm b}_N} \right)},
\end{align}
where ${\cal B} \triangleq \left\{ {\left( {{b_1},{b_2}, \ldots ,{b_N}} \right):{b_n} \in \left\{ {0,1} \right\},\forall n \in {\cal N}} \right\}$. Since $W_0=U_0$, and $U_0$ is the MNO's payoff, we complete the proof.
\end{proof}
\vspace{-2cm}

\section{Preliminary Lemmas III}
In this section, we introduce a lemma that helps us prove Theorem \ref{theorem:sortable}.

Same as Appendix \ref{appendix:section:pro1} and Appendix \ref{appendix:section:pro2}, we first assume that the bargaining sequence follows $1,2,\ldots,N$. Then we define $W_k\left({\bm b}_k\right)$ as (\ref{complicated}). Next we interchange the positions of VO $k+1$ and VO $k+2$, where $k\in\left\{0,1,\ldots,N-2\right\}$, and define ${\tilde W}_k\left({\bm b}_k\right)$ for the new sequence as (\ref{complicatedprime}). We introduce the following lemma.
\begin{lemma}\label{lemma7}
If $Q_{k+2}\ge Q_{k+1}$ and $X_{k+2}^t\ge X_{k+1}^t$ for all $t=1,2,\ldots,T$, we have ${{\tilde W}_k}\left( {{{\bm b}_{k}}} \right) \ge {W_k}\left( {{{\bm b}_{k}}} \right)\!$ for any ${\bm b}_{k}$.
\end{lemma}
\begin{proof}
The proof is similar to the proof of Lemma \ref{lemma5}.

First, we study $W_k\left({\bm b}_k\right)$. We define:
\begin{align}
A \triangleq W_{k+2}\left({\bm b}_k,0,0\right),\\
B \triangleq W_{k+2}\left({\bm b}_k,0,1\right),\\
C \triangleq W_{k+2}\left({\bm b}_k,1,0\right),\\
D \triangleq W_{k+2}\left({\bm b}_k,1,1\right).
\end{align}

According to Lemma \ref{lemma1},
\begin{align}
&W_{k+1}\left({\bm b}_k,0\right)=\frac{1}{2}A+\frac{1}{2}\max\left\{A,B+Q_{k+2}\right\},\\
&W_{k+1}\left({\bm b}_k,1\right)=\frac{1}{2}C+\frac{1}{2}\max\left\{C,D+Q_{k+2}\right\},\\
\nonumber
&W_k\left({\bm b}_k\right)=\frac{1}{2}W_{k+1}\left({\bm b}_k,0\right)+{~~~~~~~~~~~~~~~~~~~~~}\\
&\frac{1}{2}\max\left\{W_{k+1}\left({\bm b}_k,0\right),W_{k+1}\left({\bm b}_k,1\right)+Q_{k+1}\right\}.\label{labelA}
\end{align}

Then, we study ${\tilde W}_k\left({\bm b}_k\right)$. We define:
\begin{align}
{\tilde A} \triangleq {\tilde W}_{k+2}\left({\bm b}_k,0,0\right),\\
{\tilde B} \triangleq {\tilde W}_{k+2}\left({\bm b}_k,0,1\right),\\
{\tilde C} \triangleq {\tilde W}_{k+2}\left({\bm b}_k,1,0\right),\\
{\tilde D} \triangleq {\tilde W}_{k+2}\left({\bm b}_k,1,1\right).
\end{align}

According to Lemma \ref{lemma1},
\begin{align}
&{\tilde W}_{k+1}\left({\bm b}_k,0\right)=\frac{1}{2}{\tilde A}+\frac{1}{2}\max\left\{{\tilde A},{\tilde B}+Q_{k+1}\right\},\\
&{\tilde W}_{k+1}\left({\bm b}_k,1\right)=\frac{1}{2}{\tilde C}+\frac{1}{2}\max\left\{{\tilde C},{\tilde D}+Q_{k+1}\right\},\\
\nonumber
&{\tilde W}_k\left({\bm b}_k\right)=\frac{1}{2}{\tilde W}_{k+1}\left({\bm b}_k,0\right)+{~~~~~~~~~~~~~~~~~~~~~}\\
&\frac{1}{2}\max\left\{{\tilde W}_{k+1}\left({\bm b}_k,0\right),{\tilde W}_{k+1}\left({\bm b}_k,1\right)+Q_{k+2}\right\}.\label{labelB}
\end{align}

Now we compare $W_k\left({\bm b}_k\right)$ and ${\tilde W}_k\left({\bm b}_k\right)$.

Based on Lemma \ref{lemma2}, we have
\begin{align}
A={\tilde A},C={\tilde B},B={\tilde C},D={\tilde D}.\label{smallAA}
\end{align}

Based on $Q_{k+2}\ge Q_{k+1}$, $X_{k+2}\ge X_{k+1}$, and Lemma \ref{lemma3}, we have
\begin{align}
B\ge C.\label{smallBB}
\end{align}

We compare $W_k\left({\bm b}_k\right)$ and ${\tilde W}_k\left({\bm b}_k\right)$ under the following nine cases:
\vspace{-0.2cm}
\begin{itemize}
\item \textbf{Case 1}: $B+Q_{k+2}\ge C+Q_{k+1} \ge A$, $D\ge B-Q_{k+1}\ge C-Q_{k+2}$;
\item \textbf{Case 2}: $B+Q_{k+2}\ge C+Q_{k+1} \ge A$, $B-Q_{k+1}\ge D \ge C-Q_{k+2}$;
\item \textbf{Case 3}: $B+Q_{k+2}\ge C+Q_{k+1} \ge A$ , $B-Q_{k+1}\ge C-Q_{k+2}\ge D$;
\item \textbf{Case 4}: $B+Q_{k+2}\ge A \ge C+Q_{k+1}$ , $D\ge B-Q_{k+1}\ge C-Q_{k+2}$;
\item \textbf{Case 5}: $B+Q_{k+2}\ge A \ge C+Q_{k+1}$ , $B-Q_{k+1}\ge D \ge C-Q_{k+2}$;
\item \textbf{Case 6}: $B+Q_{k+2}\ge A \ge C+Q_{k+1}$ , $B-Q_{k+1}\ge C-Q_{k+2}\ge D$;
\item \textbf{Case 7}: $A\ge B+Q_{k+2}\ge C+Q_{k+1} $ , $D\ge B-Q_{k+1}\ge C-Q_{k+2}$;
\item \textbf{Case 8}: $A\ge B+Q_{k+2}\ge C+Q_{k+1} $ , $B-Q_{k+1}\ge D \ge C-Q_{k+2}$;
\item \textbf{Case 9}: $A\ge B+Q_{k+2}\ge C+Q_{k+1} $ , $B-Q_{k+1}\ge C-Q_{k+2}\ge D$.
\end{itemize}
\vspace{-0.1cm}

Here, we only provide the analysis of \textbf{Case 1}. Under this case, we can rewrite (\ref{labelA}) and (\ref{labelB}) as
\begin{align}
\nonumber
& W_k\left({\bm b}_k\right)=\frac{1}{4}A+\frac{1}{4}B+\frac{1}{2}Q_{k+2}+\\
&\max\left\{\frac{1}{4}A+\frac{1}{4}B,\frac{1}{4}C+\frac{1}{4}D+\frac{1}{2}Q_{k+1}\right\},\label{longA}\\
\nonumber
& {\tilde W}_k\left({\bm b}_k\right)=\frac{1}{4}A+\frac{1}{4}C+\frac{1}{2}Q_{k+1}+\\
&\max\left\{\frac{1}{4}A+\frac{1}{4}C,\frac{1}{4}B+\frac{1}{4}D+\frac{1}{2}Q_{k+2}\right\},\label{longB}
\end{align}
where we use (\ref{smallAA}). We further rewrite equalities (\ref{longA}) and (\ref{longB}) as
\begin{align}
\nonumber
& W_k\left({\bm b}_k\right)=\frac{1}{4}A+\max\left\{\frac{1}{4}A+\frac{1}{2}B+\frac{1}{2}Q_{k+2},\right.\\
&{~~~~~~~~~~~~~~}\left.\frac{1}{4}B+\frac{1}{4}C+\frac{1}{4}D+\frac{1}{2}Q_{k+1}+\frac{1}{2}Q_{k+2}\right\},\\
\nonumber
& {\tilde W}_k\left({\bm b}_k\right)=\frac{1}{4}A+\max\left\{\frac{1}{4}A+\frac{1}{2}C+\frac{1}{2}Q_{k+1},\right.\\
&{~~~~~~~~~~~~~~}\left.\frac{1}{4}B+\frac{1}{4}C+\frac{1}{4}D+\frac{1}{2}Q_{k+1}+\frac{1}{2}Q_{k+2}\right\}.
\end{align}
Using (\ref{smallBB}) and $Q_{k+2}\ge Q_{k+1}$, we can easily conclude ${\tilde W}_k\left({\bm b}_k\right)\ge W_k\left({\bm b}_k\right)$. We skip the analysis for the other eight cases. For all cases, we would obtain ${\tilde W}_k\left({\bm b}_k\right)\ge W_k\left({\bm b}_k\right)$, which completes the proof.
\end{proof}

\section{Proof of Theorem \ref{theorem:sortable}}
\begin{proof}
\textbf{Part A}: We first prove that bargaining sequence ${\bm l}=\left(l_1,l_2,\ldots,l_N\right)$ with ${Q_{l_n}} \ge {Q_{l_{n + 1}}}$ and ${X_{l_n}^t} \ge {X_{l_{n + 1}}^t}$ for all $n=1,2,\ldots,N-1,t=1,2,\ldots,T$, is optimal.

We assume that sequence ${\bm l}'=\left(l_1',l_2',\ldots,l_N'\right)$ is one of the optimal bargaining sequences, and consider the following two cases.

\textbf{Case 1}: ${\bm l}'$ satisfies ${Q_{l_n'}} \ge {Q_{l_{n + 1}'}}$ and ${X_{l_n'}^t} \ge {X_{l_{n + 1}'}^t}$ for all $n=1,2,\ldots,N-1,t=1,2,\ldots,T$.

It is easy to prove that $Q_{l_n}=Q_{l_n'}$ and $X_{l_n}^t=X_{l_n'}^t$ for all $n\in {\cal N},t=1,2,\ldots,T$. Therefore, ${\bm l}$ should generate the same bargaining solution and the MNO's payoff as ${\bm l}'$. In other words, ${\bm l}$ is also optimal.

\textbf{Case 2}: ${\bm l}'$ doesn't satisfy ${Q_{l_n'}} \ge {Q_{l_{n + 1}'}}$ and ${X_{l_n'}^t} \ge {X_{l_{n + 1}'}^t}$ for all $n=1,2,\ldots,N-1,t=1,2,\ldots,T$.

Based on Lemma \ref{lemma6} and Lemma \ref{lemma7}, we conclude that, for any bargaining sequence, if we exchange the bargaining positions of VO $k+1$ and VO $k+2$, and they satisfy $Q_{k+2}\ge Q_{k+1}$ and $X_{k+2}^t\ge X_{k+1}^t$ for all $t$, the MNO's payoff does not decrease. According to this and the fact that all VOs are sortable, we can rearrange sequence ${\bm l}'$ into a sequence ${\bm l}''$ that satisfies ${Q_{l_n''}} \ge {Q_{l_{n + 1}''}}$ and ${X_{l_n''}^t} \ge {X_{l_{n + 1}''}^t}$ for all $n=1,2,\ldots,N-1,t=1,2,\ldots,T$, and has $U_0^{{\bm l}''}\ge U_0^{{\bm l}'}$. Since sequence ${\bm l}'$ is one of the optimal bargaining sequences, we conclude that sequence ${\bm l}''$ is also optimal.

Since for sequence ${\bm l}$, we have ${Q_{l_n}} \ge {Q_{l_{n + 1}}}$ and ${X_{l_n}^t} \ge {X_{l_{n + 1}}^t}$ for all $n=1,2,\ldots,N-1,t=1,2,\ldots,T$. For sequence ${\bm l}''$, we also have ${Q_{l_n''}} \ge {Q_{l_{n + 1}''}}$ and ${X_{l_n''}^t} \ge {X_{l_{n + 1}''}^t}$ for all $n=1,2,\ldots,N-1,t=1,2,\ldots,T$. It is easy to show that $Q_{l_n}=Q_{l_n''}$ and $X_{l_n}^t=X_{l_n''}^t$ for all $n\in {\cal N},t=1,2,\ldots,T$. Therefore, ${\bm l}$ should generate the same bargaining solution and the MNO's payoff as ${\bm l}''$. In other words, ${\bm l}$ is also optimal.

Combining \textbf{Case 1} and \textbf{Case 2} completes the proof of \textbf{Part A}.

\textbf{Part B}: We next prove the existence of the cooperation threshold, \emph{i.e.}, under sequence $\bm l$, if the MNO does not cooperate with a particular VO ${l_n},n=1,2,\ldots,N-1$, it won't cooperate with VO ${l_{n+1}},{l_{n+2}},\ldots,{L_N}$.

To prove the existence of the cooperation threshold, we only need to prove that, under sequence $\bm l$, if the MNO does not cooperate with a particular VO ${l_n},n=1,2,\ldots,N-1$, it won't cooperate with VO ${l_{n+1}}$. We next show this by contradiction. We suppose that, the MNO does not cooperate with VO ${l_n}$, but cooperates with VO ${l_{n+1}}$.

We assume that the MNO reached ${\bm b}_{{{n-1}}}$ with the first ${n}-1$ VOs. Because the MNO does not cooperate with VO $l_n$, but cooperates with VO $l_{n+1}$, we have
\begin{align}
&\Delta_{{n}}\left({\bm b}_{{n-1}}\right)<0,\label{endA}\\
&\Delta_{{n+1}}\left({\bm b}_{{n-1}},0\right)\ge0.\label{endB}
\end{align}
By (\ref{deltaandW}), we can express $\Delta_{{n}}\left({\bm b}_{{n-1}}\right)$ and $\Delta_{{n+1}}\left({\bm b}_{{n-1}},0\right)$ as
\begin{align}
&{\Delta _{{n}}}\left( {{\bm b}_{{n}-1}} \right)=W_{{n}}\left( {{{\bm b}_{{n}-1}}},1 \right)-W_{{n}}\left( {{{\bm b}_{{n-1}}}},0\right)+Q_{{l_n}},\label{newweek}\\
\nonumber
&{\Delta _{{n}+1}}\left( {{\bm b}_{{n}-1},0} \right)=W_{{n}+1}\left( {{{\bm b}_{{n}-1}}},0,1\right)\\
&{~~~~~~~~~~~~~~~~~~~~~~~~~~~~}-W_{{n}+1}\left( {{{\bm b}_{{n-1}}}},0,0\right)+Q_{{l_{n+1}}}.\label{newweek2}
\end{align}
We define
\begin{align}
A \triangleq {W_{{n+1}}\left( {{{\bm b}_{{n}-1}}},0,0 \right)},\\
B \triangleq {W_{{n+1}}\left( {{{\bm b}_{{n}-1}}},0,1 \right)},\\
C \triangleq {W_{{n+1}}\left( {{{\bm b}_{{n}-1}}},1,0 \right)},\\
D \triangleq {W_{{n+1}}\left( {{{\bm b}_{{n}-1}}},1,1 \right)}.
\end{align}
Recall that, under sequence $\bm l$, we have $X_{l_n}^t\ge X_{l_{n+1}}^t$ for all $t$. Together with Lemma \ref{lemma3}, we obtain
\begin{align}
B\le C.
\end{align}
Based on Lemma \ref{lemma1}, we have
\begin{align}
& W_{{n}}\left( {{{\bm b}_{{n}-1}}},0\right) =\frac{1}{2}A+\frac{1}{2}\max\left\{A,{B+Q_{l_{n+1}}}\right\},\\
& W_{{n}}\left( {{{\bm b}_{{n}-1}}},1\right) =\frac{1}{2}C+\frac{1}{2}\max\left\{C,{D+Q_{l_{n+1}}}\right\}.
\end{align}
Now we rewrite (\ref{newweek}) and (\ref{newweek2}) as
\begin{align}
\nonumber
& \Delta_{{n}}\left({\bm b}_{{n-1}}\right)=\frac{1}{2}C+\frac{1}{2}\max\left\{C,{D+Q_{l_{n+1}}}\right\}\\
& -\frac{1}{2}A-\frac{1}{2}\max\left\{A,{B+Q_{l_{n+1}}}\right\}+Q_{l_n},\label{boring}\\
& {\Delta _{{n}+1}}\left( {{\bm b}_{{n}-1},0} \right)=B-A+Q_{l_{n+1}}.
\end{align}
According to (\ref{endB}), we have
\begin{align}
B+Q_{l_{n+1}}\ge A.
\end{align}
Hence, we can rewrite (\ref{boring}) as
\begin{align}
\nonumber
\Delta_{{n}}\left({\bm b}_{{n-1}}\right)=& \frac{1}{2}C+\frac{1}{2}\max\left\{C,{D+Q_{l_{n+1}}}\right\}\\
& -\frac{1}{2}A-\frac{1}{2}B-\frac{1}{2}{Q_{l_{n+1}}}+Q_{l_n}.
\end{align}
By checking the two cases $D+Q_{l_{n+1}}\ge C$ and $D+Q_{l_{n+1}}<C$ separately, it is easy to conclude that $\Delta_{{n}}\left({\bm b}_{{n-1}}\right)\ge0$. However, this contradicts with (\ref{endA}).

Therefore, we have shown that, if the MNO does not cooperate with a particular VO ${l_n}$, $n=1,2,\ldots,N-1$, it won't cooperate with VO ${l_{n+1}}$. By applying such a fact consecutively, we prove that, if the MNO does not cooperate with a particular VO ${l_n}$, it won't cooperate with VO ${l_{n+1}},{l_{n+2}},\ldots,{L_N}$.

\textbf{Part C}: We then prove that there exists an unique $k\in{\left\{ 0 \right\} \cup \cal N}$ satisfying the following two inequalities:
\begin{align}
& \sum\limits_{t=1}^{T} {f_t\left( {\sum\limits_{n = {l_1}}^{{l_{k-1}}} {{X_{n}^t}}  + {X_{l_k}^t}} \right)} - \sum\limits_{t=1}^{T} {f_t\left( {\sum\limits_{n = l_1}^{{l_{k-1}}} {{X_{n}^t}} } \right)} + {Q_{l_k}} \ge 0,\label{conditions1}\\
& \sum\limits_{t=1}^{T} {f_t\left( {\sum\limits_{n = l_1}^{l_k} {{X_n^t}}  + {X_{l_{k + 1}}^t}} \right)}- \sum\limits_{t=1}^{T} {f_t\left( {\sum\limits_{n = l_1}^{l_k} {{X_n^t}} } \right)} + {Q_{l_{k + 1}}} < 0.\label{conditions2}
\end{align}

To prove this, we only need to use the following two facts: (i) function $f_t\left(\cdot\right),t=1,2,\ldots,T,$ is an increasing and concave function; (ii) under sequence $\bm l$, we have $Q_{l_n}\ge Q_{l_{n+1}}$ and $X_{l_n}^t\ge X_{l_{n+1}}^t$ for all $n$ and $t$. We omit the proof here.

\textbf{Part D}: We next prove that, if $k$ is the cooperation threshold, it satisfies both (\ref{conditions1}) and (\ref{conditions2}). In other words,  (\ref{conditions1}) and (\ref{conditions2}) are the necessary conditions for $k$ to be the cooperation threshold.

We first show that the cooperation threshold satisfies (\ref{conditions1}).

If $k$ is the cooperation threshold, the MNO only cooperates with the first $k$ VOs. Therefore, for the MNO's bargaining with VO $l_k$, we have
\begin{align}
\Delta_{k}\left({\bm 1}_{k-1}\right)\ge 0.\label{soon}
\end{align}
According to (\ref{deltaandW}), we can express $\Delta_{k}\left({\bm 1}_{k-1}\right)$ as
\begin{align}
\Delta_{k}\left({\bm 1}_{k-1}\right)=W_{k}\left({\bm 1}_{k-1} ,1\right)-W_{k}\left( {\bm 1}_{k-1},0\right) +Q_{l_k}.\label{longlong}
\end{align}
For $W_{k}\left({\bm 1}_{k-1} ,1\right)$, since the MNO does not cooperate with the last $N-k$ VOs, we have
\begin{align}
&b_{m}^*\left( B_{m-1}^k\left({\bm 1}_{k-1},1\right) \right)=0,\forall m=k+1,k+2,\ldots,N,\\
&\pi_{m}^*\left( B_{m-1}^k\left({\bm 1}_{k-1},1\right) \right)=0,\forall m=k+1,k+2,\ldots,N.
\end{align}
Therefore, by the definition of $W_{k}\left({\bm 1}_{k-1} ,1\right)$, we obtain
\begin{align}
{W_k}\left( { {\bm 1}_{k-1} ,1 } \right) = \Psi \left( {{{\bm 1}_{k-1} ,1,{\bm 0}_{N-k}} } \right) - \sum\limits_{n = {l_1}}^{{l_k}} {{Q_n}}.\label{longlong1}
\end{align}
For $W_{k}\left({\bm 1}_{k-1} ,0\right)$, based on Lemma \ref{lemma1}, we have
\begin{align}
\nonumber
W_{k}\left({\bm 1}_{k-1} ,0\right)& \ge {W_{k + 1}}\left( { {{\bm 1}_{k-1},0,0} } \right)\\
\nonumber
& \ge \ldots\\
\nonumber
& \ge {W_{N}}\left( { {{\bm 1}_{k-1},0,{\bm 0}_{N-k}}} \right)\\
& =\Psi \left( { {{\bm 1}_{k-1},0,{\bm 0}_{N-k}} } \right) - \sum\limits_{n = {l_1}}^{{l_{k-1}}} {{Q_n}}.\label{longlong2}
\end{align}
Based on (\ref{longlong}), (\ref{longlong1}), and (\ref{longlong2}), we conclude
\begin{align}
\Psi \left( { {{\bm 1}_{k-1},1,{\bm 0}_{N-k}} } \right)\ge \Psi \left( {{{\bm 1}_{k-1},0,{\bm 0}_{N-k}} } \right).\label{supercomplicated1}
\end{align}
Based on the definition of the social welfare, we have
\begin{eqnarray}
\sum\limits_{t=1}^{T} {f_t\left( {\sum\limits_{n = {l_1}}^{{l_{k-1}}} {{X_{n}^t}}  + {X_{l_k}^t}} \right)} - \sum\limits_{t=1}^{T} {f_t\left( {\sum\limits_{n = l_1}^{{l_{k-1}}} {{X_{n}^t}} } \right)} + {Q_{l_k}} \ge 0,
\end{eqnarray}
which is exactly (\ref{conditions1}).

Then we show that the cooperation threshold satisfies (\ref{conditions2}).

If $k$ is the cooperation threshold, the MNO does not cooperate with VO $l_{k+1}$. Therefore, for the MNO's bargaining with VO $l_{k+1}$, we have
\begin{align}
\Delta_{k+1}\left({\bm 1}_k\right)< 0.\label{soonB}
\end{align}
According to (\ref{deltaandW}), we can express $\Delta_{k+1}\left({\bm 1}_k\right)$ as
\begin{align}
\Delta_{k+1}\left({\bm 1}_k\right)=W_{k+1}\left( {\bm 1}_k,1\right) -W_{k+1}\left( {\bm 1}_k,0\right) +Q_{l_{k+1}}.\label{longshort}
\end{align}
For $W_{k+1}\left( {\bm 1}_k,0 \right)$, since the MNO does not cooperate with all remaining VOs, we have
\begin{align}
&b_{m}^*\left( B_{m-1}^{k+1}\left({\bm 1}_{k},0\right) \right)=0,\forall m=k+2,k+3,\ldots,N,\\
&\pi_{m}^*\left( B_{m-1}^{k+1}\left({\bm 1}_{k},0\right) \right)=0,\forall m=k+2,k+3,\ldots,N.
\end{align}
Therefore, by the definition of $W_{k+1}\left({\bm 1}_{k} ,0\right)$, we obtain
\begin{align}
W_{k+1}\left( {\bm 1}_k,0 \right) = \Psi \left( {{{\bm 1}_k,0,{\bm 0}_{N-k-1}} } \right) - \sum\limits_{n = {l_1}}^{{l_k}} {{Q_n}}.\label{longshort1}
\end{align}
For $W_{k+1}\left( {\bm 1}_k,1 \right)$, based on Lemma \ref{lemma1}, we have
\begin{align}
\nonumber
W_{k+1}\left( {\bm 1}_k,1 \right)& \ge {W_{k + 2}}\left( { {{\bm 1}_k,1,0} } \right)\\
\nonumber
& \ge \ldots\\
\nonumber
& \ge {W_{N}}\left( { {{\bm 1}_k,1,{\bm 0}_{N-k-1}} } \right)\\
& =\Psi \left( { {{\bm 1}_k,1,{\bm 0}_{N-k-1}} } \right) - \sum\limits_{n = {l_1}}^{{l_{k+1}}} {{Q_n}}.\label{longshort2}
\end{align}
Based on (\ref{longshort}), (\ref{longshort1}), and (\ref{longshort2}), we conclude
\begin{align}
\Psi \left( { {{\bm 1}_k,1,{\bm 0}_{N-k-1}} } \right)< \Psi \left( { {{\bm 1}_k,0,{\bm 0}_{N-k-1}} } \right).\label{supercomplicated2}
\end{align}
Based on the definition of the social welfare, we have
\begin{align}
\sum\limits_{t=1}^{T} {f_t\left( {\sum\limits_{n = l_1}^{l_k} {{X_n^t}}  + {X_{l_{k + 1}}^t}} \right)}- \sum\limits_{t=1}^{T} {f_t\left( {\sum\limits_{n = l_1}^{l_k} {{X_n^t}} } \right)} + {Q_{l_{k + 1}}} < 0,
\end{align}
which is exactly (\ref{conditions2}).

Therefore, we complete the proof of this part.

\textbf{Part E}: Now we summarize \textbf{part B}, \textbf{part C}, and \textbf{part D}.

In \textbf{part B}, we have proved the existence of the cooperation threshold. In \textbf{part D}, we have proved that the cooperation threshold should satisfy (\ref{conditions1}) and (\ref{conditions2}). In \textbf{part C}, we have proved that (\ref{conditions1}) and (\ref{conditions2}) together admit an unique $k$. Hence, we conclude that (\ref{conditions1}) and (\ref{conditions2}) are also the sufficient conditions for $k$ to be the cooperation threshold. Here we complete the whole proof.
\end{proof}

\section{Preliminary Lemmas IV}
In this section, we introduce a lemma that helps us prove Theorem \ref{homogethe}. 
Same as Appendix \ref{appendix:section:pro1}, we first assume that the bargaining sequence follows $1,2,\ldots,N$. We introduce the following lemma.
\begin{lemma}\label{lemma8}
If VOs are homogenous, for any vector ${\bm b}_{n-2},n=2,3,\ldots,N$, we have the following relation:
\begin{align}
\nonumber
& W_n\left({\bm b}_{n-2},0,1\right)-W_n\left({\bm b}_{n-2},0,0\right)\ge \\
& W_n\left({\bm b}_{n-2},1,1\right)-W_n\left({\bm b}_{n-2},1,0\right).\label{appendix:equ:pro8:a}
\end{align}
\end{lemma}
\begin{proof}
We prove it by mathematical induction.

\textbf{Part A:} It is easy to show that (\ref{appendix:equ:pro8:a}) holds for $n=N$.

\textbf{Part B:} We assume that (\ref{appendix:equ:pro8:a}) holds for $n=k$, and verify it for $n=k-1$. We define
\begin{align}
\nonumber
A\triangleq W_k\left({\bm b}_{k-3},0,0,0\right),B\triangleq W_k\left({\bm b}_{k-3},0,0,1\right),\\
\nonumber
C\triangleq W_k\left({\bm b}_{k-3},0,1,0\right),D\triangleq W_k\left({\bm b}_{k-3},0,1,1\right),\\
\nonumber
E\triangleq W_k\left({\bm b}_{k-3},1,0,0\right),F\triangleq W_k\left({\bm b}_{k-3},1,0,1\right),\\
\nonumber
G\triangleq W_k\left({\bm b}_{k-3},1,1,0\right),H\triangleq W_k\left({\bm b}_{k-3},1,1,1\right).
\end{align}
According to Lemma \ref{lemma1}, we have the following relations:
\begin{align}
W_{k-1}\left({\bm b}_{k-3},0,0\right)=\frac{1}{2}A+\frac{1}{2}\max\left\{A,B+Q\right\},\\
W_{k-1}\left({\bm b}_{k-3},0,1\right)=\frac{1}{2}C+\frac{1}{2}\max\left\{C,D+Q\right\},\\
W_{k-1}\left({\bm b}_{k-3},1,0\right)=\frac{1}{2}E+\frac{1}{2}\max\left\{E,F+Q\right\},\\
W_{k-1}\left({\bm b}_{k-3},1,1\right)=\frac{1}{2}G+\frac{1}{2}\max\left\{G,H+Q\right\}.
\end{align}
From Lemma \ref{lemma2}, we have $B=C=E$ and $D=F=G$. Furthermore, based on our assumption, we have $B-A\ge D-C$ and $F-E\ge H-G$. Therefore, we only need to consider the following four cases: 

\begin{itemize}
\item \textbf{Case 1:} $B+Q\ge A$, $D+Q\ge C$ (\emph{i.e.}, $F+Q\ge E$), and $H+Q\ge G$;
\item \textbf{Case 2:} $B+Q\ge A$, $D+Q\ge C$ (\emph{i.e.}, $F+Q\ge E$), and $H+Q< G$;
\item \textbf{Case 3:} $B+Q\ge A$, $D+Q< C$ (\emph{i.e.}, $F+Q< E$), and $H+Q< G$;
\item \textbf{Case 4:} $B+Q< A$, $D+Q< C$ (\emph{i.e.}, $F+Q< E$), and $H+Q< G$.
\end{itemize}
We need to verify the following relation for the four cases:
\begin{align}
\nonumber
& W_{k-1}\left({\bm b}_{k-3},0,1\right)-W_{k-1}\left({\bm b}_{k-3},0,0\right) \ge \\
& W_{k-1}\left({\bm b}_{k-3},1,1\right)-W_{k-1}\left({\bm b}_{k-3},1,0\right).\label{appendix:equ:pro8:b}
\end{align}

Here, we only show the analysis for \textbf{Case 1}. In this case, we have
\begin{align}
\nonumber
& W_{k-1}\left({\bm b}_{k-3},0,1\right)-W_{k-1}\left({\bm b}_{k-3},0,0\right)\\
\nonumber
& -W_{k-1}\left({\bm b}_{k-3},1,1\right)+W_{k-1}\left({\bm b}_{k-3},1,0\right)\\
\nonumber
& =\frac{1}{2}C+\frac{1}{2}D-\frac{1}{2}A-\frac{1}{2}B-\frac{1}{2}G-\frac{1}{2}H+\frac{1}{2}E+\frac{1}{2}F\\
\nonumber
& =\frac{1}{2}D-\frac{1}{2}A-\frac{1}{2}H+\frac{1}{2}E\\
\nonumber
& =\frac{1}{2}B-\frac{1}{2}A-\frac{1}{2}H+\frac{1}{2}D\\
\nonumber
& \ge \frac{1}{2}D-\frac{1}{2}C-\frac{1}{2}H+\frac{1}{2}D\\
\nonumber
& =  \frac{1}{2}F-\frac{1}{2}E-\frac{1}{2}H+\frac{1}{2}G\\
& \ge 0.
\end{align}
We can use the similar approach to prove (\ref{appendix:equ:pro8:b}) for the three remaining cases.

Combining \textbf{Part A} and \textbf{Part B} completes the proof.
\end{proof}

\vspace{-0.3cm}
\section{Proof of Theorem \ref{homogethe}}
\begin{proof}
Since VOs are homogenous, without loss of generality, we consider the bargaining sequence where the MNO bargains with VO $n$ at step $n$, $n\in{\cal N}$. Because VOs are sortable, based on Theorem \ref{theorem:sortable}, we have a threshold $k=0,1,\ldots,N,$ such that the MNO only cooperates with the first $k$ VOs and ${\hat \pi}_n=0$ for $n=k+1,\ldots,N$. Next we show that for VO $n-1$ and VO $n$, $n=2,3,\ldots,k$, we have $\Delta_{n-1}\left({\bm b}_{n-2}\right)\ge \Delta_{n}\left({\bm b}_{n-2},1\right)$ for any ${\bm b}_{n-2}$.

First, we define
\begin{align}
\nonumber
A\triangleq W_n\left({\bm b}_{n-2},0,0\right),B\triangleq W_n\left({\bm b}_{n-2},0,1\right),\\
\nonumber
C\triangleq W_n\left({\bm b}_{n-2},1,0\right),D\triangleq W_n\left({\bm b}_{n-2},1,1\right).
\end{align}

For $\Delta_{n-1}\left({\bm b}_{n-2}\right)$, we have
\begin{align}
\nonumber
& \Delta_{n-1}\left({\bm b}_{n-2}\right)=W_{n-1}\left({\bm b}_{n-2},1\right)-W_{n-1}\left({\bm b}_{n-2},0\right)+Q\\
\nonumber
& =\frac{1}{2}C+\frac{1}{2}D+\frac{1}{2}Q-\frac{1}{2}A-\frac{1}{2}B-\frac{1}{2}Q+Q\\
& =\frac{1}{2}C+\frac{1}{2}D-\frac{1}{2}A-\frac{1}{2}B+Q.
\end{align}

For $\Delta_{n}\left({\bm b}_{n-2},1\right)$, we have
\begin{align}
& \Delta_{n}\left({\bm b}_{n-2},1\right)=D-C+Q.
\end{align}

Therefore, we obtain
\begin{align}
\nonumber
& \Delta_{n-1}\left({\bm b}_{n-2}\right)- \Delta_{n}\left({\bm b}_{n-2},1\right)\\
& =\frac{1}{2}C-\frac{1}{2}D-\frac{1}{2}A+\frac{1}{2}B\ge0.
\end{align}

According to equation (\ref{solutionk}), we have $\pi_{n-1}^*\left({\bm b}_{n-2}\right)=\frac{1}{2} \Delta_{n-1}\left({\bm b}_{n-2}\right)$ and $\pi_{n}^*\left({\bm b}_{n-2},1\right)=\frac{1}{2} \Delta_{n}\left({\bm b}_{n-2},1\right)$. Since $\Delta_{n-1}\left({\bm b}_{n-2}\right)\ge \Delta_{n}\left({\bm b}_{n-2},1\right)$ for any ${\bm b}_{n-2}$, we conclude that $\pi_{n-1}^*\left({\bm b}_{n-2}\right)\ge\pi_{n}^*\left({\bm b}_{n-2},1\right)$ for any ${\bm b}_{n-2}$. Therefore, we have ${\hat \pi}_{n-1} \ge {\hat \pi}_{n}$ for $n=2,\ldots,k$. Together with ${\hat \pi}_n=0$ for $n=k+1,\ldots,N$, we complete the proof.
\end{proof}

\vspace{-0.3cm}
\section{Proof of Corollary \ref{corollary:homogenous}}
Corollary \ref{corollary:homogenous} can be easily proved by combining Theorem \ref{theorem:sortable} and Theorem \ref{homogethe}. The details are omitted here.

\vspace{-0.3cm}
\section{Examples on Heterogenous VOs}
We show examples where two VOs are homogenous in $Q_n$ but heterogenous in ${\bm X}_n$ in Figure \ref{figureextra:b}. We find that, the \emph{red} VO obtains a payoff of $1.25$ under an earlier bargaining position and a payoff of $1.5$ under a later bargaining position. That is to say, it has a higher payoff when it bargains with the MNO in the later position.
\vspace{-0.1cm}

\begin{figure}[h]
  \centering
  \includegraphics[scale=0.32]{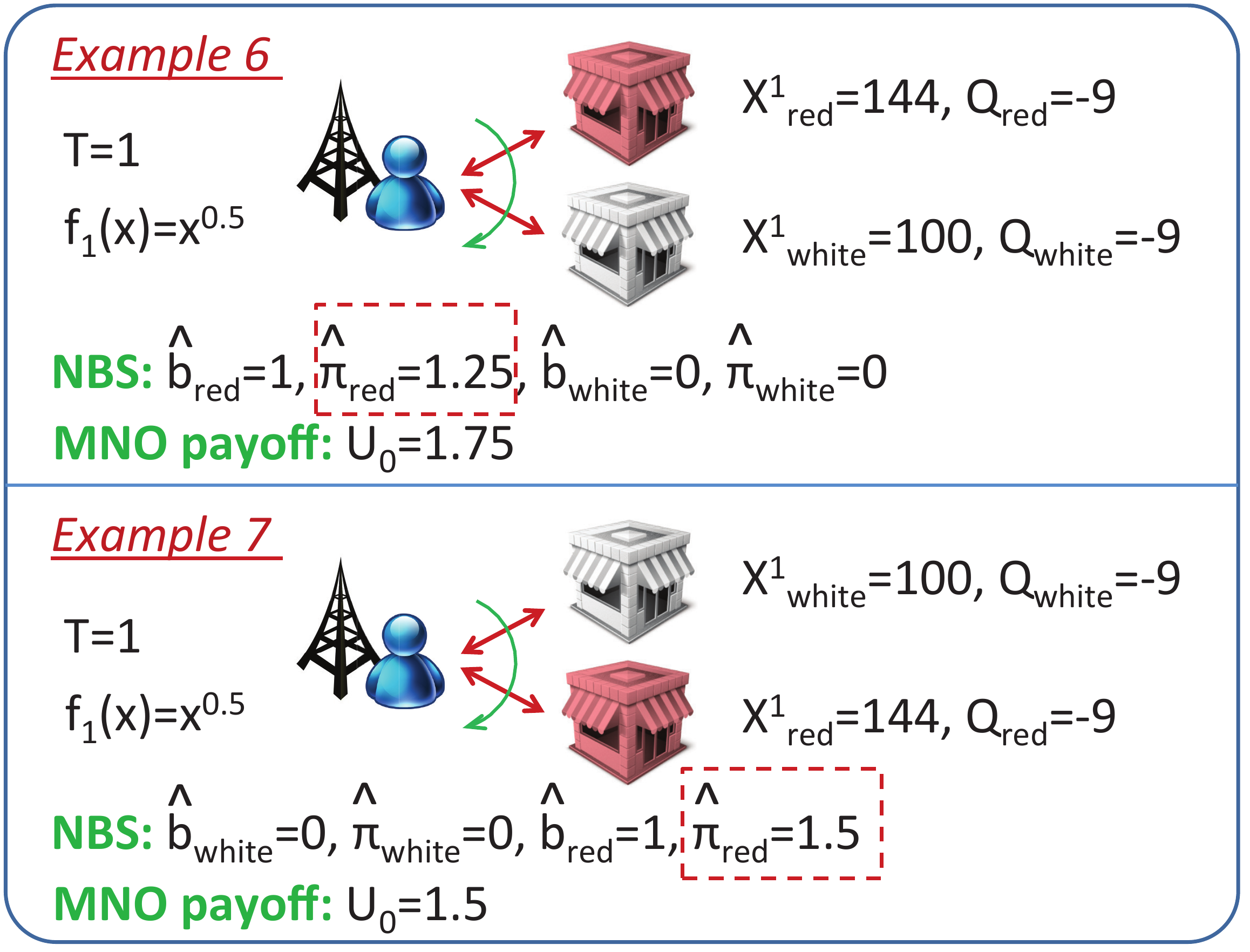}
  \centering
  \caption{Influence of Bargaining Sequence on Heterogenous VOs' Payoffs (Heterogenous $X_n$).}
  \label{figureextra:b}
\end{figure}

\end{document}